\documentclass[leqno]{amsart}
\usepackage{amsfonts,amssymb,amsmath,amsgen,amsthm}
\usepackage{hyperref}
\usepackage{bbm}
\usepackage{graphicx,color,epsfig}
\usepackage{subfigure}
\usepackage[subfigure]{ccaption}
\usepackage{float}

\theoremstyle{plain}
\newtheorem{theorem}{Theorem}[section]

\newtheorem{assumption}[theorem]{Assumption}
\newtheorem{lemma}[theorem]{Lemma}
\newtheorem{corollary}[theorem]{Corollary}
\newtheorem{proposition}[theorem]{Proposition}

\theoremstyle{remark}
\newtheorem{remark}[theorem]{Remark}
\newtheorem{notation}{Notation}

\def\Tend#1#2{\mathop{\longrightarrow}\limits_{#1\rightarrow#2}}

\font\tenms=msbm10
\font\sevenms=msbm7
\font\fivems=msbm5
\newfam\bbfam \textfont\bbfam=\tenms \scriptfont\bbfam=\sevenms
\scriptscriptfont\bbfam=\fivems

\newcommand{\beq}{\begin{eqnarray}}
\newcommand{\eeq}{\end{eqnarray}}
\newcommand{\bq}{\begin{equation}}
\newcommand{\eq}{\end{equation}}
\newcommand{\beqn}{\begin{eqnarray*}}
\newcommand{\eeqn}{\end{eqnarray*}}

\def\op{{\rm op}}

\def\DD{\mathop{\bf D\kern 0pt}\nolimits}

\def\SS{\mathop{\bf S\kern 0pt}\nolimits}
\def\ZZ{\mathop{\bf Z\kern 0pt}\nolimits}
\def\TT{\mathop{\bf T\kern 0pt}\nolimits}
\def\virgp{\raise 2pt\hbox{,}}
\def\cdotpv{\raise 1pt\hbox{ ;}}

\def\eps{\varepsilon}
\def\beq{\begin{equation}}
\def\eeq{\end{equation}}

\def\cdotv{\raise 2pt\hbox{,}}

\def\C{{\mathbf C}}
\def\R{{\mathbf R}}
\def\N{{\mathbf N}}

\def\virgp{\raise 2pt\hbox{,}}

\def\({\left(}
\def\){\right)}
\def\<{\left\langle}
\def\>{\right\rangle}
\def\le{\leqslant}
\def\ge{\geqslant}

\def\Tend#1#2{\mathop{\longrightarrow}\limits_{#1\rightarrow#2}}

\newcommand{\be}{\begin{equation}}
\newcommand{\ee}{\end{equation}}
\newcommand{\bea}{\begin{eqnarray}}
\newcommand{\eea}{\end{eqnarray}}
\newcommand{\bee}{\begin{eqnarray*}}
\newcommand{\eee}{\end{eqnarray*}}

\def\bs{\bigskip}

\def\eps{\varepsilon}

\def\pa{\partial}
\def\na{\nabla}
\def\un{{\mathbbmss{1}}}
\numberwithin{equation}{section}

\begin{document}

\title
[Kinetic model for Graphen]
{A kinetic model for the transport of electrons in a graphene layer}
\author[C. Fermanian Kammerer]{Clotilde Fermanian Kammerer}
\address[C. Fermanian]{LAMA UMR CNRS 8050,
Universit\'e Paris EST\\
61, avenue du G\'en\'eral de Gaulle\\
94010 Cr\'eteil Cedex\\ France}
\email{Clotilde.Fermanian@u-pec.fr}
\author[F. M\'ehats]{Florian M\'ehats}
\address[F. M\'ehats]{IRMAR,
Universit\'e Rennes 1 and IPSO Inria team\\
Campus de Beaulieu\\
35042 Rennes cedex\\ France}
\email{florian.mehats@univ-rennes1.fr}
\thanks{The authors would like to express their gratitude to Caroline Lasser for her help. This work was supported by the ANR-FWF Project Lodiquas ANR-11-IS01-0003 and by the ANR project Moonrise ANR-14-CE23-0007-01.
}
\begin{abstract}
In this article, we propose a new numerical model for computation of the transport of electrons in a graphene device. The underlying quantum model for graphene is a massless Dirac equation, whose eigenvalues display a conical singularity responsible for non adiabatic transitions between the two modes. We first derive a kinetic model which takes the form of two Boltzmann equations coupled by a collision operator modeling the non-adiabatic transitions. This collision term includes a Landau-Zener transfer term and a jump operator whose presence is essential in order to ensure a good energy conservation during the transitions. We propose an algorithmic realization of the semi-group solving the kinetic model, by a particle method. We give analytic justification of the model and propose a series of numerical experiments studying the influences of the various sources of errors between the quantum and the kinetic models.
 \end{abstract}  
\maketitle

\section{Introduction}

\subsection{Graphene structures} Recently, graphene based structures have been the object of intensive research in nanoelectronics, see for instance the reviews \cite{rmp1,rmp2} and references therein. Graphene is a single 2D sheet of carbon atoms in a honeycomb lattice and, differently from conventional semiconductors, the most important aspect of graphene's energy dispersion is its linear energy-momentum relationship. Electrons behave as massless relativistic particles, the conduction and valence bands intersecting at the zero energy point, with no energy gap. These features enable to observe at low energy some physical phenomena of quantum electrodynamics, such as Klein tunneling that is, the fact that Dirac
fermions can be transmitted through a classically forbidden region.

\medskip

We are here interested in numerical schemes describing the transport of electrons in a graphene device via a kinetic model. Kinetic models are usually easier to implement numerically and have a cheaper numerical cost, compared to out-of-equilibrium full quantum models. Indeed, they fit with Lagrangian approach  while the natural treatment of the quantum model requires small discretization steps, due to the smallness of physical parameters. In this paper, we will use a particle method to solve numerically the kinetic model. Moreover, the treatment of boundary conditions is simpler in this framework, which also enables to enrich the description by adding collisional effects via Boltzmann-like terms.
However, due to the absence of gap between the conduction and valence bands, it is not correct to describe separately electrons and holes, which remain coupled even at the semiclassical limit. The objective of this paper is to introduce a kinetic model for ballistic transport, which treats the possible transitions between bands and fits with easy numerical realizations. This kinetic model is derived rigorously in a linear setting and leads to algorithmic realizations which is tested numerically. 

\medskip 

Previous kinetic models  have been discussed  by O.~Morandi and F.~Sch\"urrer in~\cite{morandi} and a quite similar strategy as ours has been developed at the same moment where we were writing this paper by A. Faraj and S. Jin in~\cite{FS}.  We refer to Section~\ref{sec:num} below for further details.

\subsection{The quantum model}

The kinetic model that will be introduced below consists in a system of approximate equations based on the Wigner counterpart of an underlying quantum transport model. 
At the quantum level, the ensemble of particles is described by its density matrix $\varrho(T)$, solving the von Neumann equation
$$i\hbar \pa_T\varrho=[{\mathcal H},\varrho].$$
The Hamiltonian reads
$${\mathcal H}=-i \hbar v_F{\mathbf \sigma}\cdot \na_X+eU=v_F\hbar \begin{pmatrix} 0 &-i\pa_{X_1}-\pa_{X_2} \\ -i\pa_{X_1}+\pa_{X_2} & 0 \end{pmatrix}+eU,$$
where $X\in \R^2$, $v_F$ is the Fermi velocity, $\sigma=(\sigma_{X_1},\sigma_{X_2})$ denotes the Pauli matrices vector and $U=U(X)$ is  a smooth bounded potential with bounded derivatives, see \cite{rmp2} for physical references. 

\medskip

Let us first put this equation in dimensionless form. We introduce a characteristic space length $L$, a characteristic energy $\overline{E}$ and a characteristic density $\overline n$, then define the associated characteristic time by $\overline t=\frac{\hbar}{\overline{E}}$ and denote
$$x=\frac{X}{L},\quad t=\frac{T}{\overline t},\quad V=\frac{eU}{\overline{E}},\quad \varrho^\eps=\frac{\varrho}{\overline n L^2}.$$
The system in dimensionless form reads
\begin{equation}\label{eq:system}
i\eps \pa_t\varrho^\eps=[(A(\eps D)+V),\varrho^\eps],
\end{equation}
where the semiclassical dimensionless parameter is
$$\eps=\frac{\hbar v_F}{\overline{E}L}\ll 1,$$
where $D=-i\na_x$ and $A$ is the matrix 
$$A(\xi)=\begin{pmatrix}0 & \xi_1-i\xi_2\\ \xi_1+i\xi_2 & 0 \end{pmatrix}.$$
The matrix~$A(\xi)$ has two eigenvalues $|\xi|$ and~$-|\xi|$ with associated eigenprojectors~$\Pi^+$ and~$\Pi^-$,
$$\Pi^\pm={1\over 2} {\rm Id} \pm {1\over 2|\xi|} A(\xi)$$
where ${\rm Id}$ is the identity matrix. The singularity of the eigenvalues at the point $\xi=0$ is called conical singularity. As the function $U(x)$ above, the applied potential $V(x)$ is supposed to be smooth, bounded with bounded derivatives. 

\medskip

We shall assume that  for any $\eps>0$, the initial data $\varrho^\eps(0)$ is  a nonnegative trace-class operator. We shall denote by ${\mathcal L}^1(L^2(\R^2))$ the set  of trace-class operators on $L^2(\R^2)$. 
We shall assume that the family of operators $\left(\varrho^\eps(0)\right)_{\eps>0}$ is a bounded family of ${\mathcal L}^ 1(L^2(\R^d))$, that is 
\begin{equation}
\label{ass:rho0}
\exists C>0,\;\;\forall\eps>0,\;\;\| \varrho^\eps(0)\|_{{\mathcal L}^1(L^2(\R^2))}\leq C.
\end{equation}
Note that under these assumptions, we obtain 
$$\forall t\in\R,\;\;\forall\eps>0,\;\;\| \varrho^\eps(t)\|_{{\mathcal L}^1(L^2(\R^2))}\leq C.$$

\medskip 

Due to the smallness of $\eps$, any numeric scheme aiming at solving~\eqref{eq:system} has to cope with small discretization steps, in space and in time simultaneously, which induces considerable computational times. We aim here at taking the smallness of~$\eps$ as an opportunity to develop asymptotic analysis, based on Wigner transform approach. As a consequence, our numeric schemes will deal with order $1$ quantities and will not require $\eps$-dependent step of discretization.

\subsection{Wigner functions}

Denoting now by $\rho^\eps(t,x,y)$ the integral kernel of $\varrho^\eps$, the Wigner function is defined by 
$$w^\eps(t,x,\xi)=\frac{1}{(2\pi)^2}\int {\rm e}^{i\xi\cdot \eta} \rho^\eps\left(t,x-\eps {\eta\over 2},x+\eps {\eta\over 2} \right)d\eta.$$
Since $\rho^\eps(t)$ is Hilbert-Schmidt, its kernel is a function of $L^2(\R^2_x\times\R^2_y)$ and similarly for $w^\eps(t)$. Note however that this fact holds for any $\eps>0$ without any uniform bound. The fact that the family $\left(\rho^\eps(t)\right)_{\eps>0}$ is bounded in ${\mathcal L}^1(L^2(\R^2))$ implies that the family of distributions $(w^\eps(t))_{\eps>0}$ is bounded in the set of distributions (see Remark~\ref{rem:feps})

\medskip

We call diagonal part of the Wigner transform the scalar distributions 
\begin{equation}\label{def:Wigner}
w^\eps_\pm(t,x,\xi)={\rm tr} \left( \Pi^\pm(\xi) w^\eps(t,x,\xi) \Pi^\pm(\xi)\right)
\end{equation}
and, since $\Pi^+$ and $\Pi^-$ are rank~$1$ operators, we have 
$$w^\eps(t,x,\xi)= w^\eps_+(t,x,\xi) \Pi^+(\xi)+w^\eps_-(t,x,\xi) \Pi^-(\xi)+\Pi^+w^\eps(t,x,\xi) \Pi^- + \Pi^-w^\eps(t,x,\xi) \Pi^+.$$

\medskip

When $\eps$ is small, the off-diagonal contribution to the Wigner transform is known to be highly oscillating in time so that 
$$ \Pi^\pm w^\eps(t,x,\xi) \Pi^\mp \Tend\eps 0 0 \;\;{\rm in}\;\;{\mathcal D}'(\R_t\times (\R^{2d}\setminus\{\xi=0\})),$$
(see \cite{GMMP}). For this reason, we focus on the quasi-distribution functions  $w^\eps_\pm(t,x,\xi) $.
Far from the crossing set $\{\xi=0\}$, $w^\eps_+$ and $w^\eps_-$ satisfy approximated  transport equations
  \begin{equation}\label{eq:kin1}
\left\{ \begin{array} l
\displaystyle\partial_t w^\eps_+ +{\xi\over|\xi|}\cdot\nabla_xw^\eps_+-\nabla V(x)\cdot\nabla_\xi w^\eps_+=\mathcal O(\eps),\\[3mm]
\displaystyle\partial_t w^\eps_- -{\xi\over|\xi|}\cdot\nabla_xw^\eps_--\nabla V(x)\cdot\nabla_\xi w^\eps_-=\mathcal O(\eps),
\end{array}\right.\end{equation}
  in ${\mathcal D}'(\R_t\times (\R^{2d}\setminus\{\xi=0\}))$.   Besides, the equations~(\ref{eq:kin1}) imply that, outside $\{\xi=0\}$, the functions  $w^\eps_\pm$ are constant along the integral curves $(\Phi^t_\pm)_{t\in\R}$  of the vector fields
  \begin{equation}\label{def:Hpm}
  H_\pm (x,\xi) =\pm\,{ \xi\over |\xi|}\cdot\nabla_x -\nabla _xV\cdot \nabla_\xi.
  \end{equation}
  Such curves -- also called Hamiltonian curves of $\pm\,|\xi|+V(x)$ -- are well-defined and smooth as long as they do not reach $\{\xi=0\}$.
  They satisfy 
  $$\dot \Phi^t_\pm(x,\xi) = H_\pm\left(\Phi^t_\pm(x,\xi)\right),\;\; \Phi^0_\pm = (x,\xi).$$
  
  Notice that the evolution of $w^\eps_+$ and $w^\eps_-$  are decoupled at leading order outside $\{\xi=0\}$: this regime is said to be adiabatic.   As long as these curves do not reach $\{\xi=0\}$, each part of the Wigner transform at time~$t$ can be simply calculated by transporting the initial Wigner transform along the curves. The natural easiest numerical scheme then consists in three steps:
  \begin{enumerate}
  \item One samples the initial Wigner functions $w^\eps_+(0)$ and $w^\eps_-(0)$ to obtain a set of weighted points $(x^j_\pm, \xi^j_\pm,w^j_\pm)$, $1\leq j\leq N^\pm$, which can be done by Monte-Carlo methods for example as in ~\cite{KL} ;
  \item One proceeds to the transport of the weighted points along the trajectories and obtains at time $t$ a family of points $(\Phi^t_\pm(x_\pm^j,\xi_\pm^j),w^j_\pm)$, $1\leq j\leq N^\pm$, which  requires to solve numerically a system of ordinary differential equations which do not depend on $\eps$ ;
  \item At time $t$ expectation values can be computed according to 
  $$\int _{\R^{2d}} a(x,\xi)w^\eps_\pm(t,x,\xi) dx\,d\xi\sim  {1\over N^\pm} \sum _{1\leq j\leq N^\pm} w^j_\pm a(\Phi^t_\pm(x_\pm^j,\xi_\pm^j)).$$
  \end{enumerate}

  \medskip 
  
 It is proved in~\cite{FG1} (see Proposition~3 therein) that the curves $\Phi^t_\pm$ may reach $\{\xi=0\}$ in finite time, and that, if $\nabla V\not=0$ at the impact point, the curve can be prolongated in a unique way away from $\{\xi=0\}$ generating a continuous trajectory (which is not~${\mathcal C}^1$). These facts are recalled in details in Section~\ref{sec:prelim} below.
 
 \medskip

  The singularity of the eigenvalues of $A(\xi)$ when $\xi=0$ is known to produce non adiabatic transitions between the modes.  The presence of a non-zero mass in the Dirac equation would prevent this difficulty. Our aim here is to propose a kinetic model which is also valid close to $\xi=0$. We are going to add  a collision kernel to the equations~(\ref{eq:kin1}), which will couple the evolutions of $w^\eps_+(t)$ and $w^\eps_-(t)$, and, thus, will generate transitions between the modes.

\subsection{Conical singularities}
Systems presenting conical singularities have been the subject of extensive works since the early thirties with the works of Landau and Zener~\cite{La,Ze}. Such singularities arise in   particular when studying molecular dynamics in the frame of Born-Oppenheimer approximation (see~\cite{MS,ST} for example). Pioneer works have been performed in this context by G. Hagedorn and his collaborators, with a wave-packet approach \cite{Ha94,HJ}. Several ideas used here are due to these contributions. Ten years ago, classification of crossings for rather general systems was performed independently by \cite{CdV1,CdV2} and \cite{FG2}. In the latter reference and in~\cite{FG1}, the analysis of the crossing is made from the point of view of Wigner transform and can be adapted  to our setting. This kind of analysis has led to numerical realizations for molecular propagation (\cite{LT}, \cite{FL1} and~\cite{FL2}) and  we have been inspired by these results. Of course, the Dirac equation arising in the graphene context presents major difference, when compared  to the Schr\"odinger equation which models molecular propagation. However, the transitions due to the conical intersections can be treated similarly.  The collision kernel which solves the transitions arising from the conical intersections, is derived from the analysis of conical intersections performed in~\cite{FG1} and from the particle description derived in~\cite{LT,FL1,FL2,FL3} for molecular dynamics. Precise statements are given below.


\subsection{The approximate kinetic model}\label{sec:kin}

The collision kernel that we are going to add in order to couple  equations~(\ref{eq:kin1})  is realized by a Landau-Zener transfer term and a jump operator that occurs on a specific manifold.  More precisely,  we 
 consider the set
$\Sigma$ defined by 
\begin{equation}\label{def:Sigma}\Sigma=\{(x,\xi)\in\R^{4},\;\; \xi\cdot\nabla V(x)=0\}
\end{equation}
which is an hypersurface  of $\R^{4}$ under the assumption 
\begin{equation}\label{ass:V}
\nabla V(x)\not=0\;\;\forall x\in\R^2. 
\end{equation}
This set is the place where the gap between the two modes (i.e. the function $2|\xi|$) is minimal along the
trajectories (see Remark~\ref{rem:Sigma}).
We notice that, assuming~(\ref{ass:V}), the vector fields $H_\pm(x,\xi)$ defined in~(\ref{def:Hpm}) are transverse to $\Sigma$ in a neighborhood of $\{\xi=0\}$. This comes from the observation that
\begin{equation}\label{eq:passingthroughSigma}
\nabla V(x)\cdot (H_\pm)_\xi + d^2V(x) \xi\cdot (H_\pm)_x=-|\nabla V(x)|^2 \pm d^2V(x) \xi\cdot {\xi\over |\xi|}<0
\end{equation}
if $\xi$ is small enough. 
As a consequence, in a small gap region, when the trajectories reach their minimal distance to the gap, they pass through $\Sigma,$  arriving from the region $\{\xi\cdot\nabla V(x)>0\}$ and going to the region $\{\xi\cdot\nabla V(x)<0\}$. 

\medskip

We define $(f^\eps_{+}(t),f^\eps_{-}(t))$ as a pair of solutions to the following system:
\begin{equation}\label{eq:kinapp}
\left\{ \begin{array} l
\displaystyle \partial_t f^\eps_{+} +{\xi\over |\xi|}\cdot \nabla_x f^\eps_{+}-\nabla V(x)\cdot \nabla_\xi f^\eps_{+} = 
 K_{+}(f^\eps_{+},f^\eps_{-})\\[3mm]
\displaystyle\partial_t f^\eps_{-} -{\xi\over |\xi|}\cdot \nabla_x f^\eps_{-}-\nabla V(x)\cdot \nabla_\xi f^\eps_{-} = 
K_{-}(f^\eps_{+},f^\eps_{-})
\end{array}\right.\end{equation}
with initial conditions $f^\eps_{+}(0)=w^\eps_+(0)$ and $f^\eps_{-}(0)=w^\eps_-(0)$ and where $K_{\pm}$ are two collision kernels, defined below in \eqref{eq:Khe} and \eqref{eq:Keh}.

$ $

The collision process is involved above $\Sigma$; as a consequence, outside $\Sigma$, the functions $f^\eps_{\pm}$ are constant along the curves $(\Phi^t_\pm)_{t\in\R}$ introduced previously and we recover system~(\ref{eq:kin1}). Starting from an initial data localized far from $\Sigma$, the solution $f^\eps_{\pm}$ of system~(\ref{eq:kinapp}) is obtained by propagating the data by the flow $(\Phi^t_\pm)_{t\in\R}$ so that the plus and the minus modes have decoupled evolutions. Whenever trajectories reach $\Sigma$, the transition kernel will generate transfers between the modes. 
\medskip

Even for smooth initial data, the result of this process will not be smooth functions and they will present discontinuities on $\Sigma$. 
In order to localize on $\Sigma$ functions that present discontinuities through it, and thus have different traces, we have to distinguish two sides of $\Sigma$. For this purpose, we take advantage from the fact that, as noticed above, the flows $H_\pm$ are transverse to $\Sigma$ in suitably chosen neighborhoods~$\Omega$ of points $(x_0,\xi_0)$ such that $\nabla V (x_0)\not=0$ and $|\xi_0|$ is small enough. 
For a function $g(t,x,\xi)$ which is defined in $I\times\Omega$, $I$ open interval of $\R$, and continuous outside $I\times\Sigma$,
we denote by $g_{\Sigma,in} $,  the restriction to $\Sigma$ of the function $g \ {\bf 1}_{\{\xi\cdot\nabla V(x)\geq 0\}}$ and 
by $g_{\Sigma,out} $  the restriction to $\Sigma$ of $g\,{\bf 1}_{\{\xi\cdot\nabla V(x)\leq 0\}}.$ We shall call $g_{\Sigma,in}$ the {\it ingoing} trace of $g$ on $\Sigma$ and $g_{\Sigma,out}$ the {\it outgoing} one.
We will see in Section~\ref{sec:markov} below that we can extend this definition to functions solutions to~(\ref{eq:kinapp}) with $L^1$-initial data in such a way that the definition coincides whenever the considered solutions happen to be continuous outside $\R\times \Sigma$.

\medskip 

Let us now 
 describe these collision kernels
 $K_{\pm} $. They  depend on a transfer coefficient
\begin{equation}
\label{def:Teps}
T_{\eps}(x,\xi)  =  \exp\left(-{\pi\over  \eps} {  |\xi|^2 \over |\nabla V(x)|}\right),
\end{equation}
and on two jump operators 
$$
J_{\pm}(x,\xi)  =  
 \left(x\pm 2|\xi|{\nabla V(x)\over|\nabla V(x)|^2},\xi\right),\;\;(x,\xi)\in\Sigma.$$
Then,
the collision kernels $K_\pm$ are defined by 
\begin{eqnarray}\label{eq:Khe}
K_+(f,g) & = &  \Lambda_+(x,\xi)\delta_\Sigma(x,\xi)\left(T_\eps f_{\Sigma,in}-\left(T_\eps g_{\Sigma,in}\right)\circ J_+\right),
\\
\label{eq:Keh}
K_-(f,g) & = & \Lambda_-(x,\xi) \delta_\Sigma(x,\xi)\left(T_\eps g_{\Sigma,in} -\left(T_\eps f_{\Sigma,in}\right)\circ J_-\right),
\end{eqnarray}
 where the Jacobians $\Lambda_\pm$ are given by
 \begin{equation}\label{def=theta+}
\forall (x,\xi) \in\Sigma,\quad \Lambda_\pm(x,\xi) = 
{-|\nabla V(x)|^2\pm|\xi|^{-1}d^2V(x)\xi\cdot\xi\over \sqrt{|\nabla V(x)|^2+|d^2V(x)\xi|^2}}.
\end{equation}

\medskip

\begin{remark}\label{UepsR}
Note that the transfer coefficient $T_\eps(x,\xi)$ is exponentially small as soon as $|\xi|>R\sqrt\eps$ for some $R>0$.  Moreover, if $|\xi|\leq R\sqrt\eps$, we have
$$\displaylines{
J_\pm= {\rm Id} +\mathcal O(R\sqrt\eps)\;\;
{\rm and} \;\;
\Lambda_\pm(x,\xi) =-|\nabla V(x)|+\mathcal O(R\sqrt\eps),\cr}$$
under the assumption \eqref{ass:V}.
 \end{remark}
 
\medskip


\begin{assumption}\label{ass}
\begin{enumerate}
\item The initial data $(\varrho^\eps_0)_{\eps>0}$ satisfies~(\ref{ass:rho0}) and its Wigner transform $(w^\eps(0))_{\eps>0}$ is  localized away from~$\Sigma$.
\item  The potential is non degenerated: $\nabla V(x)\not= 0$. 
\item  We have $w^\eps_-(0)=\mathcal O(\eps ^{1/8})$ in $L^1(\R^{2d})$ and the symbol $a$  and the time $T$ are such that within the time interval $[0,T]$, each of the trajectories  $(\Phi^t_+)$
 arriving at the
support of $a$ at time $T$ has passed through $\Sigma$ at most once. 
\end{enumerate}
\end{assumption}

Our main result is the following theorem, which states that the functions $(f^\eps_{\pm})$ provides an approximation of the Wigner transforms $(w^\eps_\pm)$.
The following statement claims that, under Assumptions~\ref{ass}, the functions $(f^\eps_{\pm})$ provides an approximation of the Wigner transforms $(w^\eps_\pm)$. Besides, as we shall see later in the next section, the functions $(f^\eps_\pm)$ fits to easy numerical realization.

\begin{theorem}
\label{theo:main}
If (1) and (2) of Assumption~\ref{ass} is satisfied on the time interval $[0,T]$, 
then \eqref{eq:kinapp} admits a unique weak solution.\\
 Moreover if  $a\in{\mathcal C}_0^\infty(\R^2) $ is compactly supported outside~$\{\xi=0\}$ and   $\chi\in{\mathcal C}_c^\infty([0,T],\R)$ satisfy (3) of Assumption~\ref{ass}, then
there exist positive constants $C, \eps_0>0$ such that, for all
$0<\eps<\eps_0$, 
\begin{equation}
\label{eq:approx} \qquad\left| {\rm
tr}\int_{\R^{5}}\,\chi(t)
\left(f^\eps_{\pm}-w^\eps_{\pm} 
\right)\!(t,x,\xi)\,a(x,\xi) \, d x\,d \xi\,d t
\right| \leq C\eps^{1/8}.
\end{equation}
\end{theorem}

\begin{remark}\label{rem:theo}
\begin{enumerate}
\item The existence of solutions to these kinetic equations comes from the fact that, under (1) and (2) of Assumption~\ref{ass}, these equations have a particle description relying on a Markov semi-group that is explained in  Section~\ref{sec:markov} and is crucial for the proof of the  theorem. 
\item Note that the role of the indexes $plus$ and $minus$ can be inverted in (3) of  Assumptions~\ref{ass} and result of Theorem~\ref{theo:main} still holds. 
\item The $\eps^{1/8}$ approximation comes from our approach and we suspect that the exponent $1/8$ is not optimal.
\item The limitation induced by (3) to the range of validity of Theorem~\ref{theo:main} comes from the fact that the kinetic kernels $K_\pm$ are not adapted in some situations where the modes interfere too much. It appears nevertheless that  these kernels' description encounter a larger range of situation than those satisfying~(3), as it appears in the numerical realizations of Section~\ref{sec:numerics}. Some example of situation where (3) is not satisfied and where the description by the kernels~$K_\pm$ fails is given in~\cite{FL4} in the context of conical intersections for molecular dynamics. 
 \end{enumerate}
\end{remark}


\subsection{The algorithmic realization}\label{sec:num}
Thanks to a semi-group realization of the kinetic model which is performed in Section~\ref{sec:markov}, the mechanism describing the evolution of $f^\eps_{\pm}(t,x,\xi)$ has the simple algorithmic description:
\begin{enumerate}
\item Far from $\Sigma$, $f^\eps_+(t,x,\xi)$ propagates along the trajectories $\Phi^t_+$ and $f^\eps_-(t,x,\xi)$ propagates along the trajectories $\Phi^t_-$.
\item  Whenever a trajectory reaches $\Sigma$ at time $t^*$ in a point $(x^*,\xi^*)$, one may transmit some energy to the other mode according to a random process. One takes a random number $r$ between $0$ and $1$ and one compares $r$ and the transfer coefficient $T_\eps(x^*,\xi^*)$:
\begin{itemize}
\item If $r>T_\eps(x^*,\xi^*)$, one continues with the same trajectory and propagate the mass
$f^\eps_+(t^*,x^*,\xi^*)$
on the trajectory $\Phi^t_+(x,\xi)$.
\item If $r<T_\eps(x^*,\xi^*)$, one  initiates a trajectory $\Phi^t_-$ from the point  
\begin{equation}\label{jump}
(x^{t^*}_-,\xi^{t^*}_-):=J_+(x^*,\xi^*)=\left(x^*+2|\xi^*| {\nabla V(x^*)\over|\nabla V(x^*)|^2} ,\xi^*\right)
\end{equation}
and propagate the mass
$f^\eps_-(t^*,x^*,\xi^*)=f^\eps_+(t^*,x^*,\xi^*)$
on the new trajectory $\Phi^{t-t^*}_-(x_-^{t^*},\xi_-^{t^*})$.
\end{itemize}
\end{enumerate}
A similar process is performed on the other mode.
\begin{remark}\label{rem:(3)}
Note that the hypothesis (3) of Assumptions~\ref{ass} imply that at a transition point, only one of the trajectory is weighted. 
\end{remark}
The kinetic system proposed by O.~Morandi and F.~Sch\"urrer in~\cite{morandi} is obtained by expliciting some of the neglected terms in the pseudodifferential approach which gives~(\ref{eq:kin1}) at first approximation. Indeed, the $O(\eps)$ term in~(\ref{eq:kin1}) is no longer small when $\xi$ is close to~$0$, and O.~Morandi and F.~Sch\"urrer explicits this term which couples the equations. However, this pseudodifferential symbolic calculus can only been mathematically justified when $\xi$ is non zero and, as far as we know, O.~Morandi and F.~Sch\"urrer's approximated system can only be justified for non zero though small $\xi$. On the contrary, the approximation by system~\eqref{eq:kinapp} enjoys a mathematical justification.

\medskip

 In~\cite{FS}, 
A.~Faraj and S.~Jin uses a hopping algorithm which consists in transitions with the same rate $T_\eps$ as ours, however, they do not implement the jumps resulting from the operators $J_\pm$. We emphasize the importance of these jumps as shown in Figure~\ref{fig5} below. As pointed out in Remark~\ref{rem:energy}, these jumps aim at preserving the energy of the trajectories  during the transitions. There is also in~\cite{FS} an interesting numerical comparison of the model proposed O.~Morandi and F.~Sch\"urrer and the one of A.~Faraj and S.~Jin which shows the pertinence of the Landau Zener transition rate~$T_\eps$. 

\medskip

In Section~\ref{sec:numerics} below, we shall  present various numerical experiments in order to validate the kinetic model \eqref{eq:kinapp}. In particular, we shall study numerically several sources of error which are linked with the choice of the model, instead of numerical errors due to time and space discretization, considering that we have taken sufficiently small time steps and space steps, such that the error associated to these numerical parameters is negligible compared to the modeling errors.

\subsection{Organization of the paper} We begin by presenting in Section~\ref{sec:numerics} the numerical experiments arising from this analysis. Then, we explain  the underlying Markov semi-group realization which is at the core of the analysis in Section~\ref{sec:markov}. This allows to give a proof of Theorem~\ref{theo:main} in Section~\ref{sec:proof}, which justify the pertinence of the kinetic model and of its numerical realization. Finally, an Appendix is devoted to some technical aspects related with pseudo differential calculus.


\section{Numerical experiments}\label{sec:numerics}

\subsection{The simulated models}

In this section, we present various numerical experiments in order to validate our kinetic model \eqref{eq:kinapp}. We assume that the initial data has a Wigner transform supported in the domain $\Omega \subset \{|x_1|\leq a,\,|\xi|\leq M\}$. In pratice, we will take $a=M=10$. We will not discuss the behavior close to the boundary and we will consider a time schedule $[0,T]$ such that the trajectories issued from $\Omega$ do not reach any boundary. For the quantum model, we will use periodic boundary conditions. In the case of the barrier potential, we will also assume that the trajectories issued from points of $\Omega $ do not reach the top of the barrier with velocity $0$; thus, we are in the frame of Assumption~\ref{ass} and the above description of the trajectories is valid.

Several sources of errors can be identified in this model, if we compare it to the original quantum equation \eqref{eq:system}:
\begin{itemize}
\item[--] the error made in the computation of the initial data and on its sampling by a bunch of particles,
\item[--] the error made during the plus/minus transition processes; on the computation of the transmission coefficients and the position of particles after transitions (presence or not of the jump process),
\item[--] the error made during the transport phase, when the quantum transport in replaced by the classical transport induced by the Hamiltonian \eqref{def:Hpm} and when the coherence effects between particles are neglected,
\item[--] the error made when we neglect the transport of the antidiagonal part of the Wigner function.
\end{itemize}
In the numerical tests that we present, we concentrate our study on these modeling errors, instead of numerical errors due to time and space discretization. Hence, we have taken sufficiently small time steps and space steps, such that the error associated to these numerical parameters is negligible compared to the modeling errors.

In order to characterize these different modeling errors, let us identify the models that we simulate:
\begin{itemize}
\item[(i)] The {\em quantum graphene model} is the original quantum equation \eqref{eq:system}, computed with the Strang splitting method. In the first series of experiments (Subsections \ref{sub1}, \ref{sub2} and \ref{sub3}), the initial data is a gaussian coherent quantum state localized at the position $x^0_1$ in the $x_1$ direction, multiplied by a plane wave in the $x_2$ direction, with the momentum $\xi^0=(\xi^0_1,\xi^0_2)$, and polarized on the plus mode: 
\begin{equation}\label{initquantum}
\rho^\eps_0(x,y)= \psi^\eps_0(x_1-x_1^0,\xi^0_1)\left(\psi^\eps_0(y_1-x_1^0,\xi^0_1)\right)^*\exp\left(i\frac{(x_2-y_2)\xi_2^0}{\eps}\right),
\end{equation}
with \begin{equation}\label{initquantum2}
\psi_0^\eps(x_1,\xi^0_1)=\Pi^+(\eps D)\left(\begin{array}{c}\sqrt{2}\,u_0^\eps(x_1,\xi^0_1)\\0\end{array}\right),
\end{equation}
and
\begin{equation}
\label{gaussian}
u_0^\eps(x_1,\xi^0_1)=(\pi\eps )^{-1/4}\exp\left(-\frac{x_1^2}{2\eps}+i\frac{x_1 \xi^0_1}{\eps}\right).
\end{equation}
In the last experiments (Subsection \ref{sub4}), the initial data is a mixture of coherent quantum states
\begin{align}&&\rho^\eps_0(x,y)=\iint_{\R^2}\psi^\eps_0(x_1-x^0_1,\xi^0_1)\left(\psi^\eps_0(y_1-x^0_1,\xi^0_1)\right)^*\exp\left(i\frac{(x_2-y_2)\xi_2^0}{\eps}\right)\quad \nonumber\\
&&\times f^0(x_1^0,\xi_1^0)dx^0_1d\xi^0_1,\label{matinit}
\end{align}
where $f^0$ is a given distribution density and where, for simplicity, all the states are taken with the same momentum $\xi_2^0$ in the $x_2$ direction.
\item[(ii)] The {\em kinetic graphene model} is \eqref{eq:kinapp}, discretized according to the particle algorithm described in Section \ref{sec:num}, with randomly computed transitions between plus and minus modes, with or without jumps \eqref{jump}. The initial data $f_{+,0}^\eps$ and $f_{-,0}^\eps$ are described in Subsection \ref{sub2}. The time integrator is the triple jump method of order 4 \cite{GNI}.
\item[(iii)] The {\em quantum pseudo-graphene model} is the following modified quantum equation, computed with the Strang splitting method:
\begin{equation}\label{eq:pseudosystem}
i\eps \pa_t\varrho^\eps=[(\widetilde A(\eps D)+V),\varrho^\eps],\qquad \mbox{with }\widetilde A(\xi)=\begin{pmatrix} |\xi|&0\\0&-|\xi| \end{pmatrix}
\end{equation}
with the same initial data $\rho_0^\eps$ as (i). Far from the crossing set $\{\xi=0\}$, this model displays the same dynamics as (i), but the major difference is that no transition occur between plus and minus modes with \eqref{eq:pseudosystem}.
\item[(iv)] The {\em kinetic pseudo-graphene model} is the same as (ii), without transition process. It is the classical counterpart of the quantum pseudo-graphene model (iv).
\end{itemize}


\subsection{The Klein effect}\label{sub1}
Our aim in this subsection is to observe qualitatively the Klein effect (the tunneling of particles through a classically forbidden potential barrier) with both models: the quantum graphene model and the kinetic graphene model. For the simulations presented in this subsection, we have taken a sufficiently large number of particles ($4\times 10^6$) and assume that we are at numerical convergence. The potential is the following smooth barrier potential, depending only on $x_1$:
$$V(x_1)=4 \sin^3\left(\frac{\pi}{4}(x_1+1)\right)\mbox{ for } x_1\in [-1,3], \mbox{ and }V(x_1)=0\mbox{ otherwise}.$$ 

Let us describe the phenomenology that can be observed on Figures \ref{fig0} and \ref{fig1}. Here, $\eps=0.064$. We represent on Figure \ref{fig0} the contour plots of the trajectories of the plus and minus modes, computed with the quantum graphene model. On Figure \ref{fig1}, the plus and minus densities are represented at four instants, computed with the quantum graphene model (plotted in plain lines) and computed with our kinetic graphene model (plotted with the 'X'). It can be seen that the results given by both models are in very good agreement.
\begin{figure}[!htbp]
  \includegraphics[width=.8\textwidth]{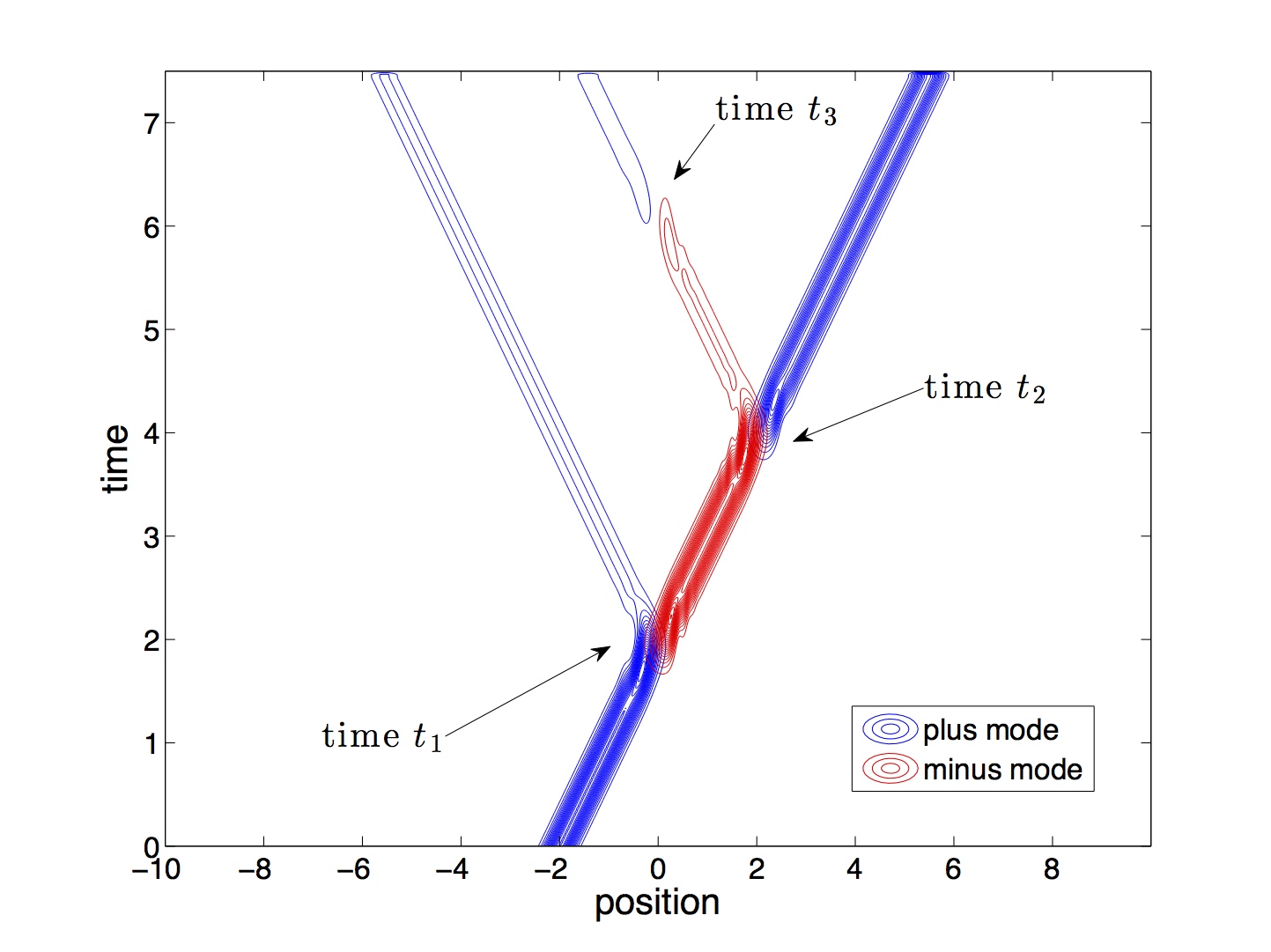}
    \caption{Contour plots of the plus and minus densities, quantum graphene model with $\eps=0.064$.}
\label{fig0}
\end{figure}
\begin{figure}[!htbp]
  \centerline{
  \subfigure[$t=0$]{\label{fig1a}\includegraphics[width=.6\textwidth]{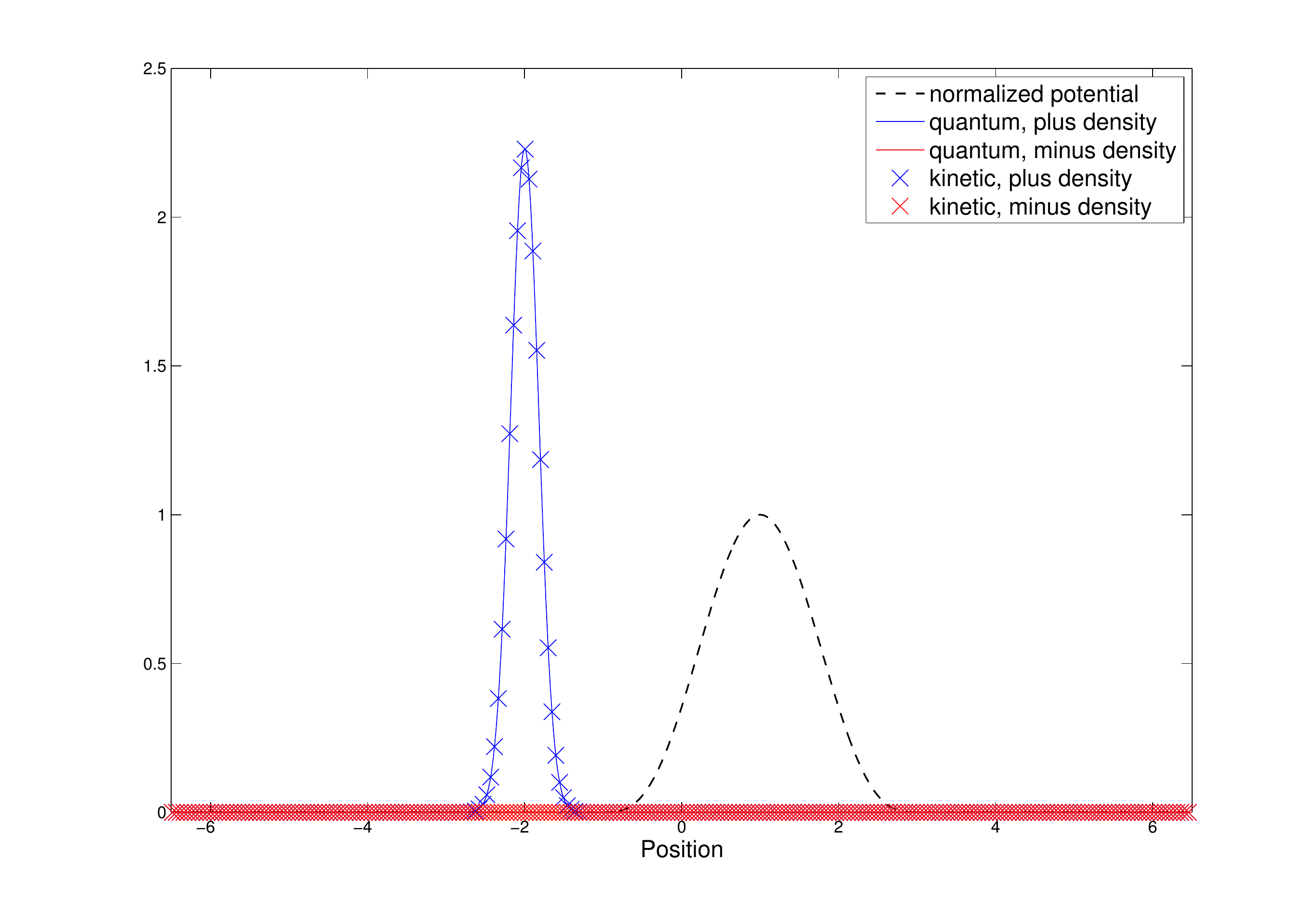}}\hspace*{-8mm}
  \subfigure[$t=3$]{\label{fig1b}\includegraphics[width=.6\textwidth]{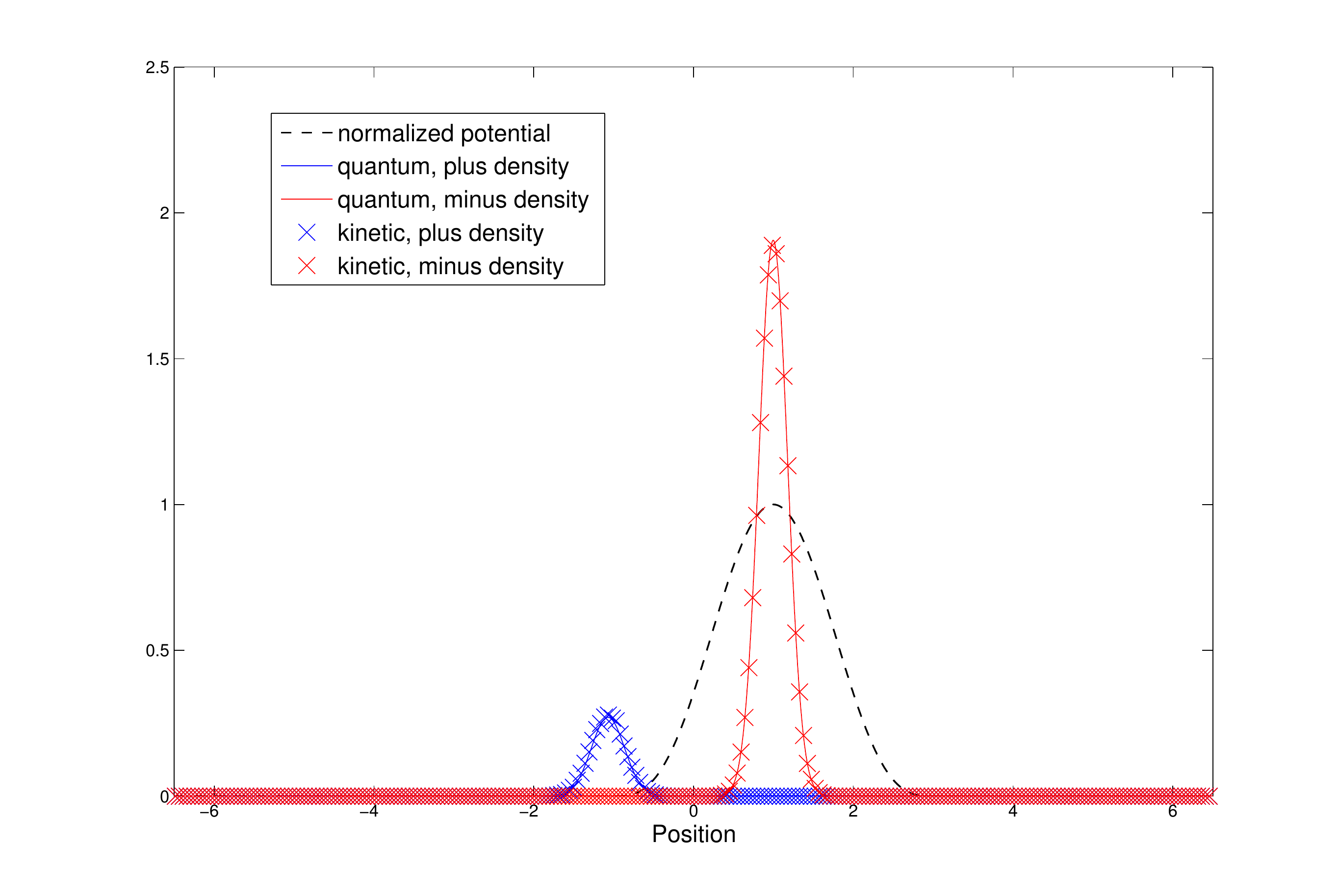}}}
  \centerline{\subfigure[$t=4.5$]{\label{fig1c}\includegraphics[width=.6\textwidth]{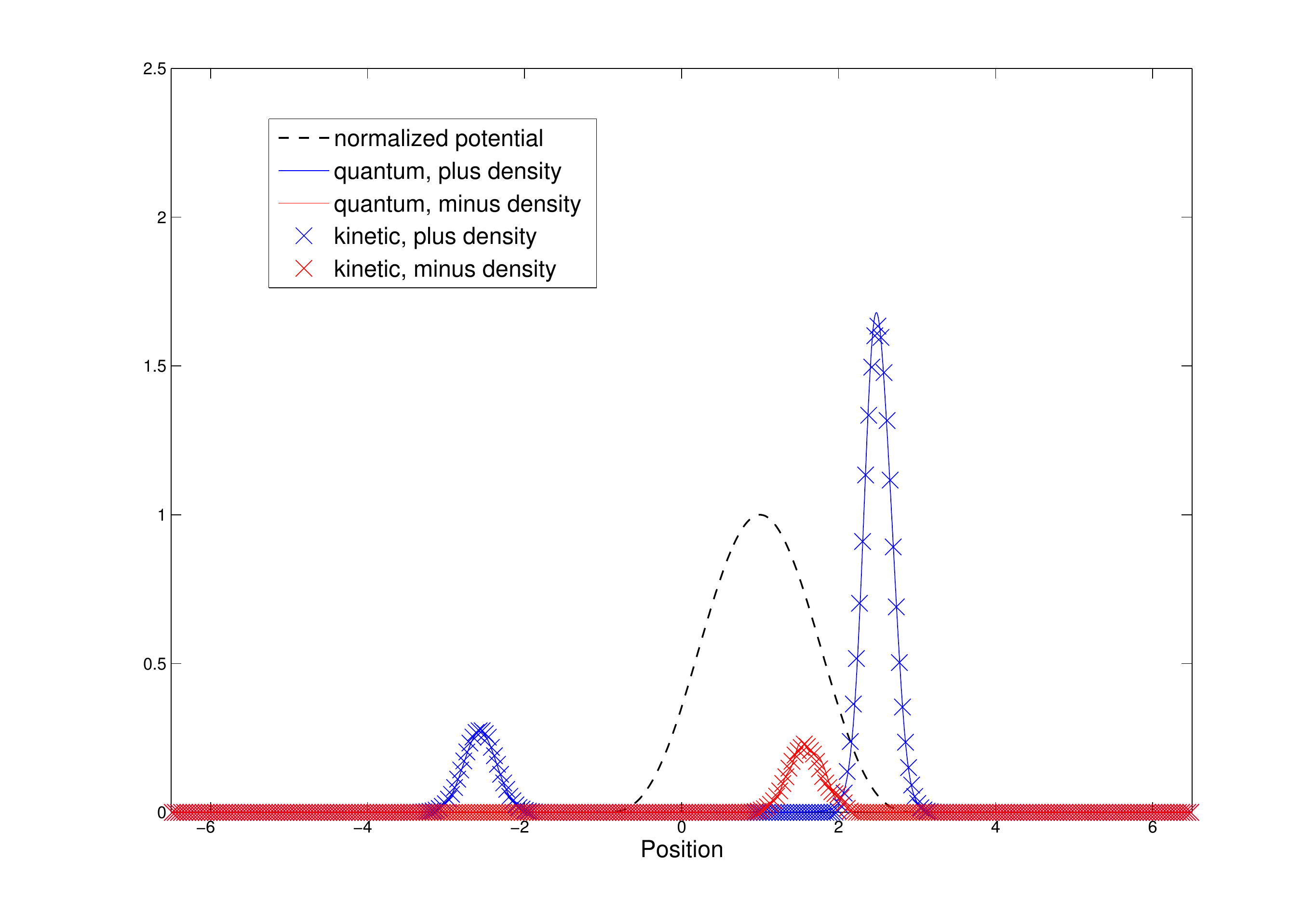}}\hspace*{-8mm}
  \subfigure[$t=7.5$]{\label{fig1d}\includegraphics[width=.6\textwidth]{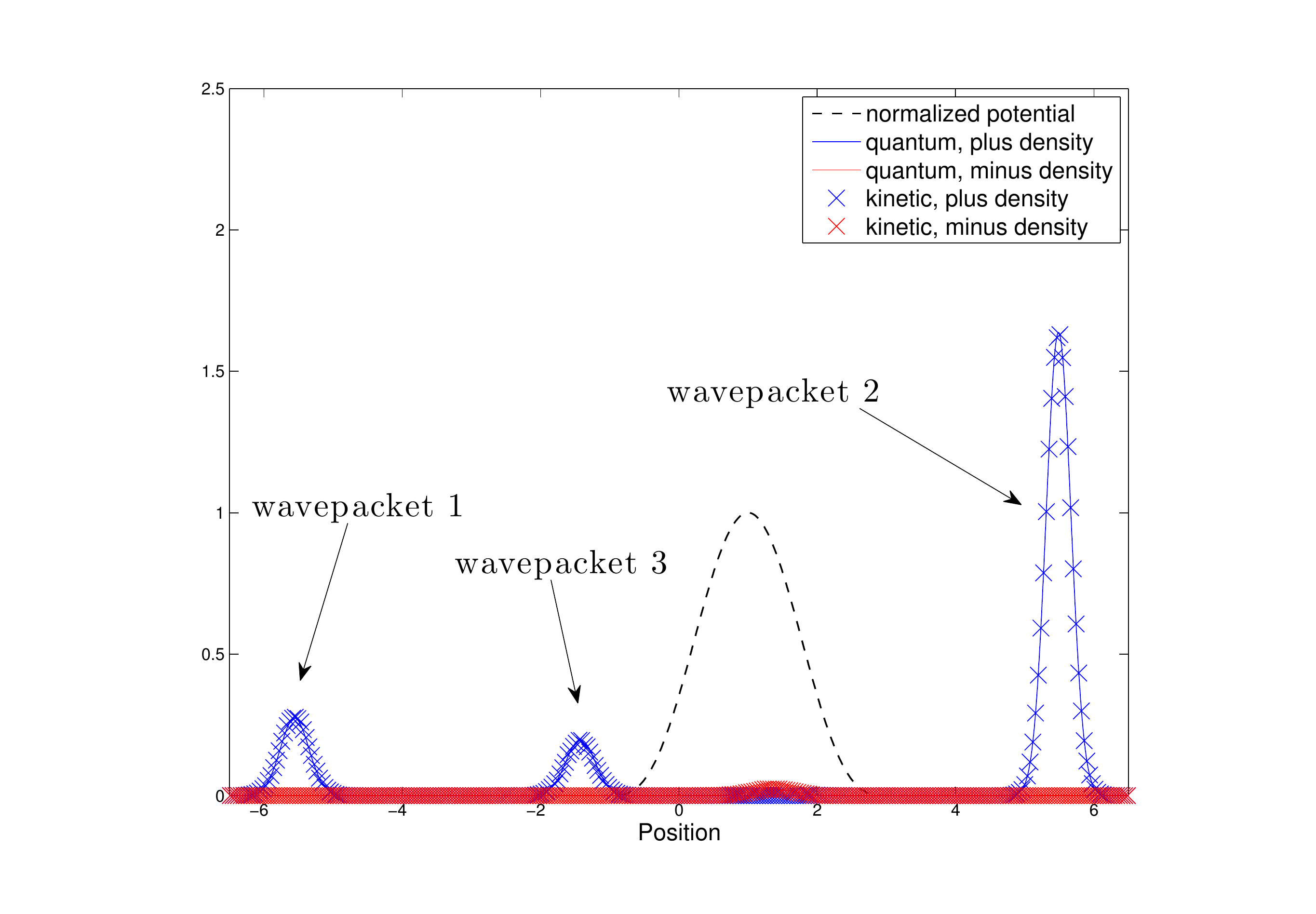}}}
    \caption{Propagation of a coherent wavepacket through a barrier for $\eps=0.064$.}
\label{fig1}
\end{figure}

Initially, a coherent wavepacket in the plus mode is at position $x^0_1=-2$ with the momentum $\xi^0=(\xi_1^0,\xi_2^0)=(1.3,0.1)$ (Figure \ref{fig1a}). The wavefunction first propagates freely, then enters inside the barrier and its momentum $\xi_1$ decreases until it vanishes. At the instant $t_1$, the wavefunction is partially reflected and partially transferred through the barrier in the minus mode. Then, the minus wavepacket propagates inside the barrier (see Figure \ref{fig1b}) until it reaches the other side of the barrier (instant $t_2$), through which it is partially transferred into a plus wavepacket and partially reflected. Finally, the remainding minus wavepacket propagates from the right to the left (see Figure \ref{fig1c}) and hits again the barrier (instant $t_3$), where it is almost integrally transferred into a third plus wavepacket outside the barrier.

At the end of the simulation (Figure \ref{fig1d}), almost all the mass have been redistributed into three plus wavepackets. In Table \ref{table1}, we give the numerical transfer rates (i.e. the ratio of mass in each wavepacket over the initial mass) of the initial mass into wavepackets 1, 2 and 3 for different values of $\eps$, computed with the quantum model and with the kinetic model. Here again, we observe a good agreement between our kinetic model and the reference one.
\begin{table}[h!]
\label{table1}
\caption{Transfer rates}\centering
\begin{tabular}{|c|c|c|c|c|c|c|}
\hline
$\eps$&0.128&0.064&0.032&0.016&0.008&0.004\\
\hline\hline
wavepacket 1 (quantum)&7.42\%&  14.18\%&   26.25\%&   45.55\%&   70.31\%&   91.17\%\\
\hline
wavepacket 1 (kinetic)&7.64\%&  14.64\%&   27.03\%&   46.73\%&   71.47\%&   91.92\%\\
\hline\hline
wavepacket 2 (quantum)&85.71\%&73.65\%&54.39\%&29.66\%&8.82\%&0.78\%\\
\hline
wavepacket 2 (kinetic)&85.69\%&73.49\%&54.17\%&29.46\%&8.79\%&0.75\%\\
\hline\hline
wavepacket 3 (quantum)&6.38\%&10.46\%&  14.28\% & 13.50\%&  6.20\%&   0.71\%\\
\hline
wavepacket 3 (kinetic)&6.20\%&  10.2\%&  13.98\%& 13.15\%&  6.01\%&  0.68\% \\
\hline
\end{tabular}
\end{table}
\subsection{Choice of the initial data}
\label{sub2}
Let us concentrate on the computation of the initial data for the kinetic model. Recall that, in the pure-state case, the initial data for the quantum graphene model is the density matrix given by \eqref{initquantum}, where the wave function \eqref{initquantum2} is the gaussian wavefunction \eqref{gaussian} projected on the plus mode.  Let us denote by $f_0^\eps(x,\xi)$ the Wigner function of $\varrho^\eps_0$ and by $\widetilde{f_0^\eps}(x,\xi)$ the Wigner function of the gaussian wavepacket
$$\widetilde {\rho_0^\eps}(x,y)=u^\eps_0(x_1-x_1^0,\xi^0_1)\left(u^\eps_0(y_1-x_1^0,\xi^0_1)\right)^*\exp\left(i\frac{(x_2-y_2)\xi_2^0}{\eps}\right),$$
which is
\begin{align}
\widetilde{f_0^\eps}(x,\xi)&=\frac{1}{2\pi}\int {\rm e}^{i\xi\cdot \eta}u^\eps_0\left(t,x_1-\eps {\eta\over 2}\right)\overline{u_0^\eps}\left(x_1+\eps {\eta\over 2}\right)d\eta\,\delta_{\xi_2=\xi_2^0}\nonumber\\
&=\frac{1}{\pi\eps}\exp\left(-\frac{(x_1-x^0_1)^2}{\eps}-\frac{(\xi_1-\xi^0_1)^2}{\eps}\right)\delta_{\xi_2=\xi_2^0}\,.\label{wig}
\end{align}
The function ${\rm tr} (\Pi^+(\xi)f_0^\eps(x,\xi)\Pi^+(\xi))$ is asymptotically close to $\widetilde{f_0^\eps}(x,\xi)$ as $\eps\to 0$. More precisely, using pseudo-differential calculus, the following expansion can be obtained:
\begin{align}
\label{devinit1}
{\rm tr} (\Pi^+(\xi)f_0^\eps(x,\xi)\Pi^+(\xi))&=\widetilde {f_0^\eps}(x,\xi)-\frac{\eps}{2}\frac{\xi_2}{|\xi|^2}\partial_x \widetilde {f_0^\eps}(x,\xi)+\mathcal O(\eps)\\\label{devinit2}
&=\widetilde {f_0^\eps}\left(x_1-\frac{\eps}{2}\frac{\xi_2}{|\xi|^2},\xi\right)+\mathcal O(\eps)\end{align}
in $L^1(\R^4)$. Notice that this expansion is in powers of $\eps^{1/2}$ (indeed, the $L^1$ norm of $\partial_x \widetilde {f_0^\eps}(x,\xi)$ is of order $\eps^{-1/2}$).

As initial data for our kinetic graphene model, in the pure-state case, let us experiment these two approximations of $f_0^\eps$: the first one is simply $\widetilde {f_0^\eps}(x,\xi)$ and the second one is the shifted function 
\begin{equation}
\label{f+0}
f_{+,0}^\eps(x,\xi)=\widetilde {f_0^\eps}\left(x_1-\frac{\eps}{2}\frac{\xi_2}{|\xi|^2},\xi\right).
\end{equation} The advantage of this last choice compared to the right-hand side of  \eqref{devinit1} is that the distribution function is always positive. On Figure \ref{fig2}, we plot in logarithmic scale the errors on the densities, i.e. the quantities 
$$\mathrm{err}_0=\int \left| |\psi_0^\eps(x)|^2-\int \widetilde{f_0^\eps}(x,\xi)d\xi\right|dx_1,\quad \mathrm{err}_1=\int \left| |\psi_0^\eps(x)|^2-\int f_{+,0}^\eps(x,\xi)d\xi\right|dx_1,$$
for $\eps=2^N\times 2.5\times10^{-3}$, with $N\in\{0,1,2,3,4,5,6,7,8,9\}$. Again, these integrals have been computed with enough discretization points such that the  numerical integration errors is negligible: the $(x,\xi)$ domain $[-10,10]^2$ is discretized with $2^{18}$ grid points in the $x$ direction and $2^{11}$ grid points in the $\xi$ direction with a non uniform cartesian mesh refined near the point $(x^0,\xi_1^0)$. As above, we have taken $(x^0_1,\xi_1^0, \xi^0_2)=(-2,1.3,0.1)$. Figure \ref{fig2} confirms the estimate \eqref{devinit1} and \eqref{devinit2}: the function $\widetilde {f_0^\eps}$ is an $\mathcal O(\eps^{1/2})$ approximation of $f_0^\eps$ and the shifted function $f_{+,0}^\eps(x,\xi)$ is an order $\mathcal O(\eps)$ approximation. From now on, we choose this shifted function $f_{+,0}^\eps$ defined by \eqref{f+0} (and $f_{-,0}^\eps=0$) as initial data for the kinetic graphene model.
\begin{figure}[!htbp]
  \includegraphics[width=.7\textwidth]{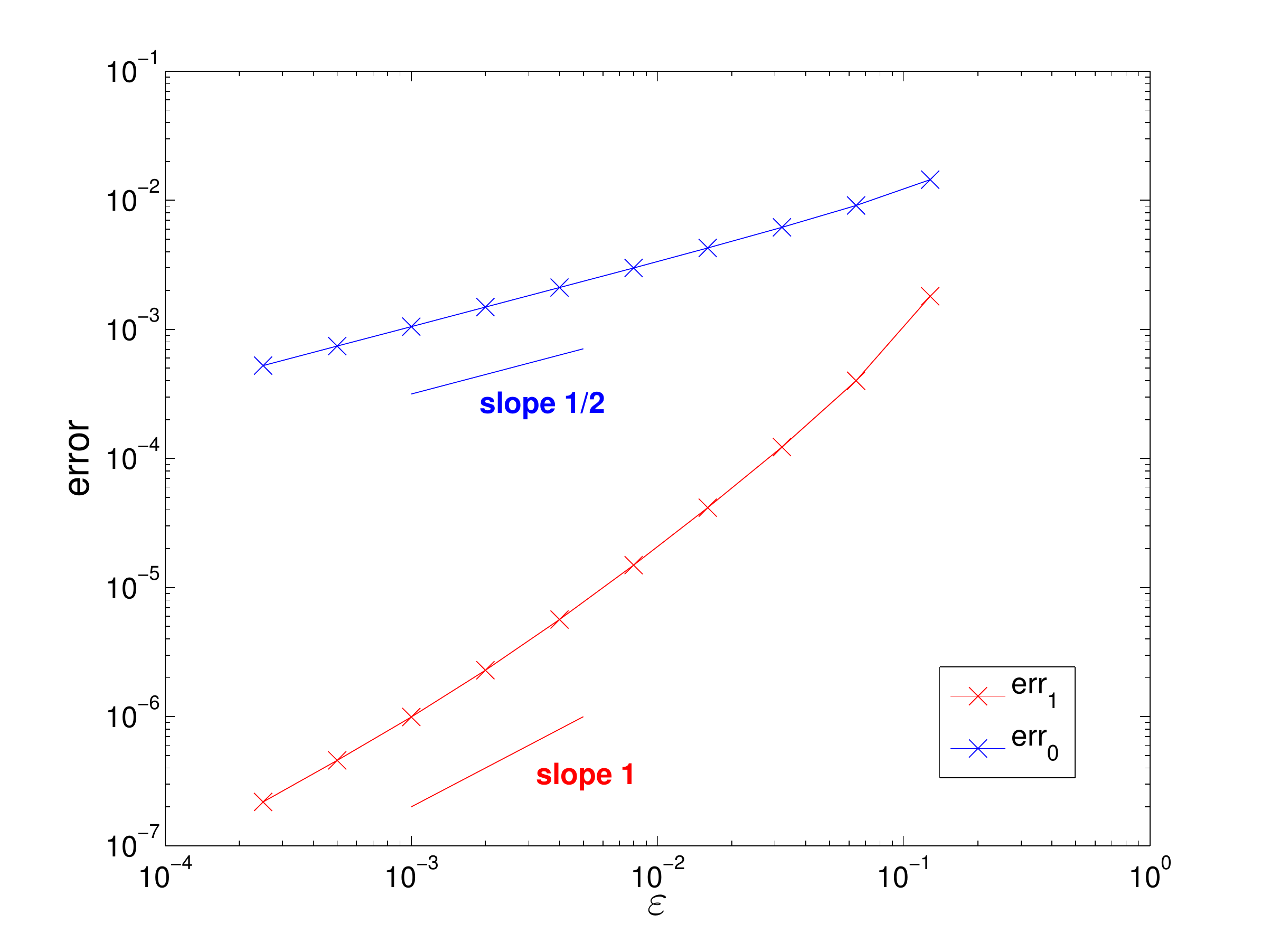}
    \caption{$L^1$ error on the initial density, for the two choices $\widetilde{f_0^\eps}$ and $f_{+,0}^\eps$.}
\label{fig2}
\end{figure}

Let us now discuss briefly the initial sampling step for the distribution function $f_{+,0}^\eps$ and compare the rates of convergence of Monte Carlo and quasi-Monte Carlo sampling. Monte Carlo sampling can be achieved by a rescaled and shifted sampling of a bidimensional Gaussian distribution, taking advantage of the tensorial structure of the function. Its convergence rate is known to be of order $\mathcal O(1/\sqrt{N_{part}})$, where $N_{part}$ denotes the number of particles. Quasi-Monte Carlo methods use quasi-random sequences, also known as low-discrepancy, which are deterministic approximation of the uniform distribution on $[0,1]^2$ and can be transformed into a Gaussian distribution by the cumulative distribution function. Such methods display better convergence rate, of the form $C(\log N_{part})^k /N_{part}$. In both cases, densities are reconstructed by a fifth order spline interpolation method.

On Figure \ref{fig3}, we have represented with the same scales, for the quasi-Monte Carlo method (using a 2D Hammersley set) and the Monte Carlo method, the $L^1$ error between the reconstructed density and the reference density $|\psi_0^\eps|^2$, for four different values of $\eps$. For the quasi-Monte Carlo method, we observe a convergence rate which is compatible with $\mathcal O(C(\log N_{part})^2 /N_{part})$ (with a saturation due to the difference between $f_{+,0}^\eps$ and $f_0^\eps$, studied above). For the Monte Carlo method, we observe a slower convergence, of the form $\mathcal O(1/\sqrt{N})$. Therefore, in the sequel we systematically use  the quasi-Monte Carlo method, with $4\times 10^6$ particles.
\begin{figure}[!htbp]
  \centerline{
  \subfigure[Quasi-Monte Carlo method]{\label{fig4a}\includegraphics[width=.6\textwidth]{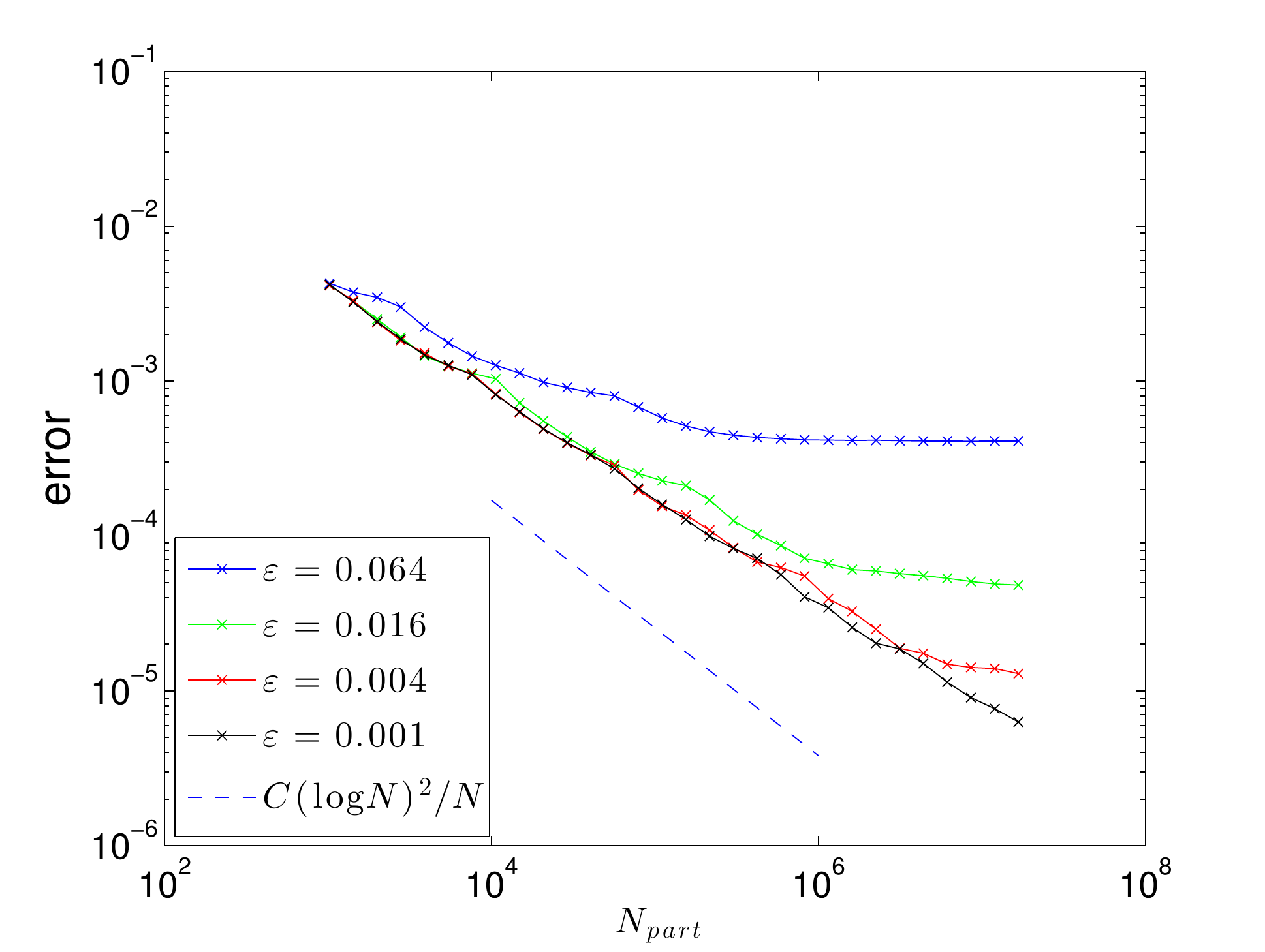}}\hspace*{-8mm}
  \subfigure[Monte Carlo method]{\label{fig4b}\includegraphics[width=.6\textwidth]{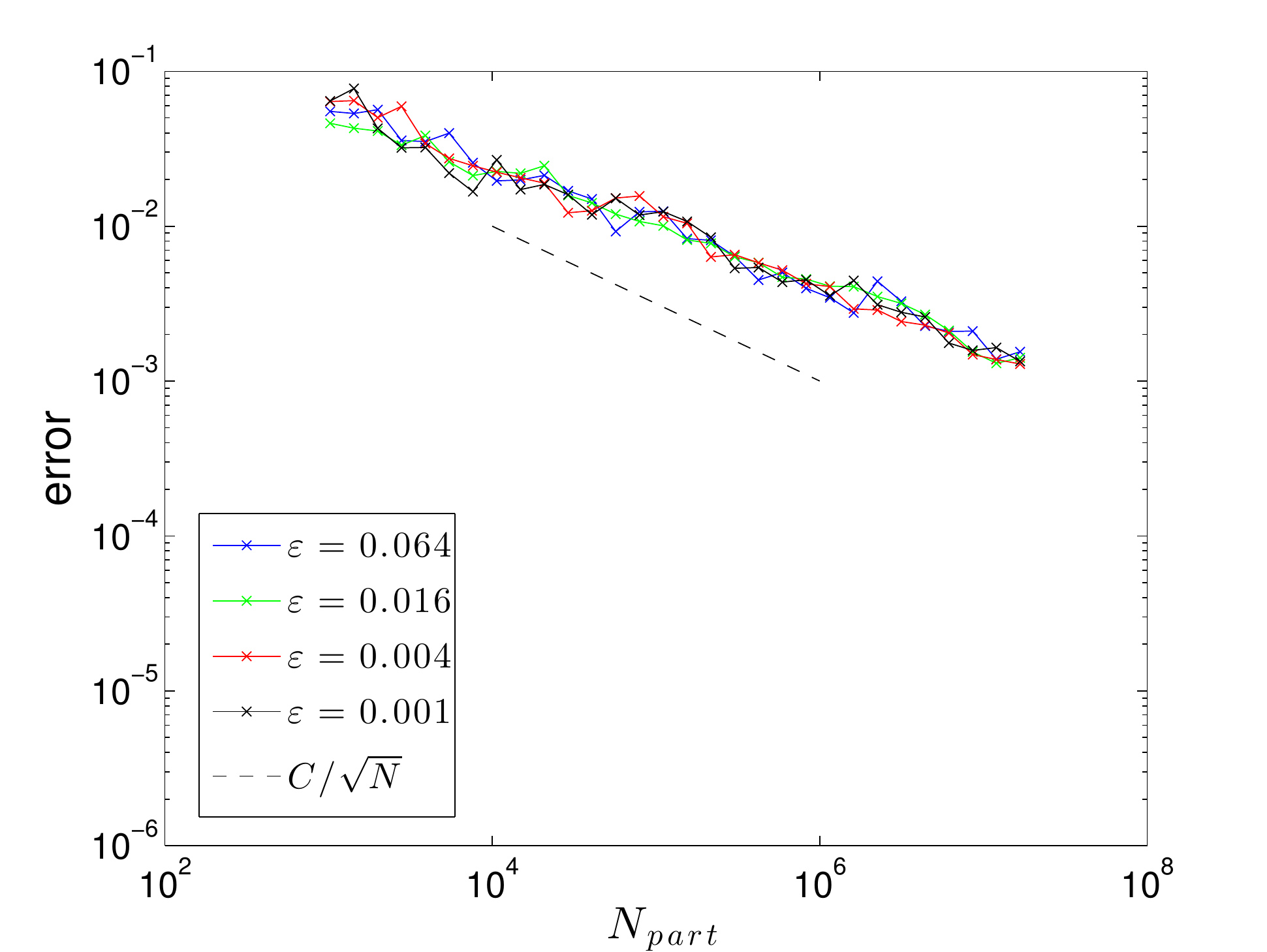}}}
    \caption{$L^1$ error on the initial density with respect to the number of particles for the quasi-Monte Carlo and the Monte Carlo methods}
\label{fig3}
\end{figure}
\subsection{Transport and transition phases}
\label{sub3}
In this section, we compare the dynamics computed with the quantum and the kinetic models. The initial data are the same as in the previous subsections: $\psi_0^\eps(x_1-x_1^0)$ given by \eqref{initquantum2} and $f_{+,0}^\eps$ defined by \eqref{f+0}. Here also, we take $(x^0_1,\xi_1^0, \xi^0_2)=(-2,1.3,0.1)$. The final time of the simulation is $t_f=4.5$.

\subsubsection*{Transport phase}
In a first step, in order to characterize the error made during the transport phase, we simulate the pseudo-graphene model, which displays the same transport properties as the graphene model but which induces no transition between plus and minus modes. We simulate the actions of two potentials:
\begin{equation}
\label{V1V2}
V_1(x_1)=\frac{1}{20}(x_1+10)^2\quad \mbox{and}\quad V_2(x_1)=\alpha\,\mathrm{atan}(2x_1+\frac{\pi}{2}),
\end{equation}
where $\alpha\approx 0.643$ has been adjusted such that, at the points where a particle is stopped by the potential barrier (resp. $x^*_1$ for $V_1$ and $x^*_2$ for $V_2$), we have $V'_1(x^*_1)=V_2'(x^*_2)$.
On Figure \ref{fig4}, we plot the error
\begin{align*}
\mathrm{error}=&\int \left| |\Pi^+\psi^\eps(t_f,x)|^2-\int{f_{+}^\eps}(t_f,x,\xi)d\xi\right|dx_1\\
&+\int \left| |\Pi^-\psi^\eps(t_f,x)|^2-\int{f_{-}^\eps}(t_f,x,\xi)d\xi\right|dx_1
\end{align*}
as a function of $\eps$. We observe two features. First, the error behaves as $\mathcal O(\eps^{1/2})$ in both cases. This can be explained by the fact that the derivative in $x$ and $\xi$ of the coherent wavepacket are of order $\eps^{-1/2}$ in $L^1$ norm. Second, the error for the transport in $V_2$ is 5 times higher than the error for the transport in $V_1$. This can be explained by the fact that $V_1$ is a harmonic potential, so the quantum transport operator coincides with its Wigner counterpart. Hence, for the transport by $V_1$, the main source of error comes from the fact that coherence effects between particles are not taken into account. For the transport by $V_2$, the main source of error comes from the replacement of the quantum transport operator by the classical one.
\begin{figure}[!htbp]
  \includegraphics[width=.7\textwidth]{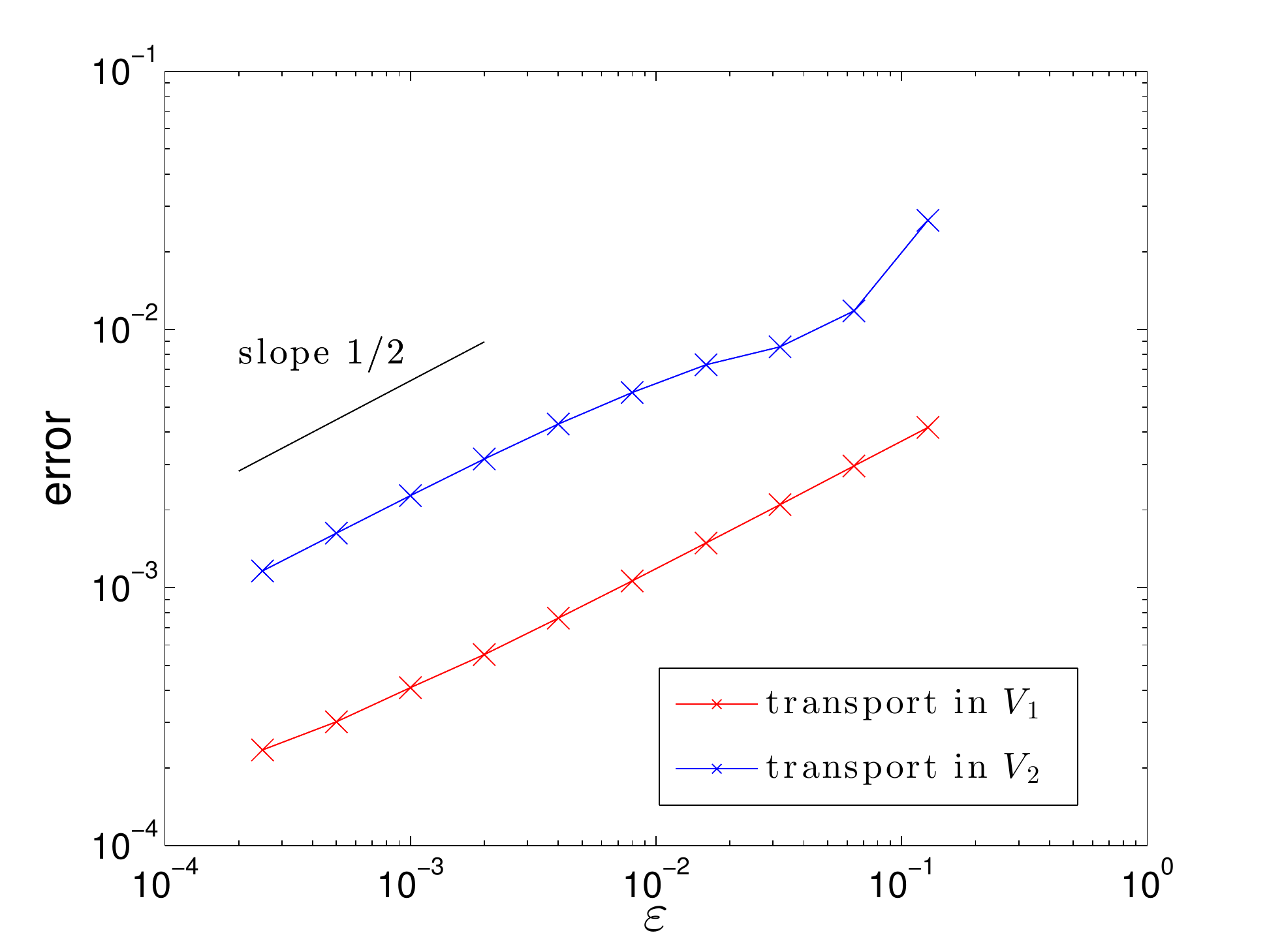}
    \caption{$L^1$ error on the densities for the pseudo-graphene model as a function of $\eps$, at time $t_f=4.5$, with the two potentials $V_1$ (harmonic) and $V_2$ (anharmonic).}
\label{fig4}
\end{figure}
\bs
\subsubsection*{Transition phase}
Let us now come back to the real graphene models, with transitions. We compute the transport by the potential $V_1$ of a coherent wavepacket initially at $(x^0_1,\xi_1^0, \xi^0_2)=(-2,1.3,0.1)$. We plot on Figure \ref{fig5bis} the plus and minus trajectories in the potential $V_1$. 
\begin{figure}[!htbp]
  \includegraphics[width=.8\textwidth]{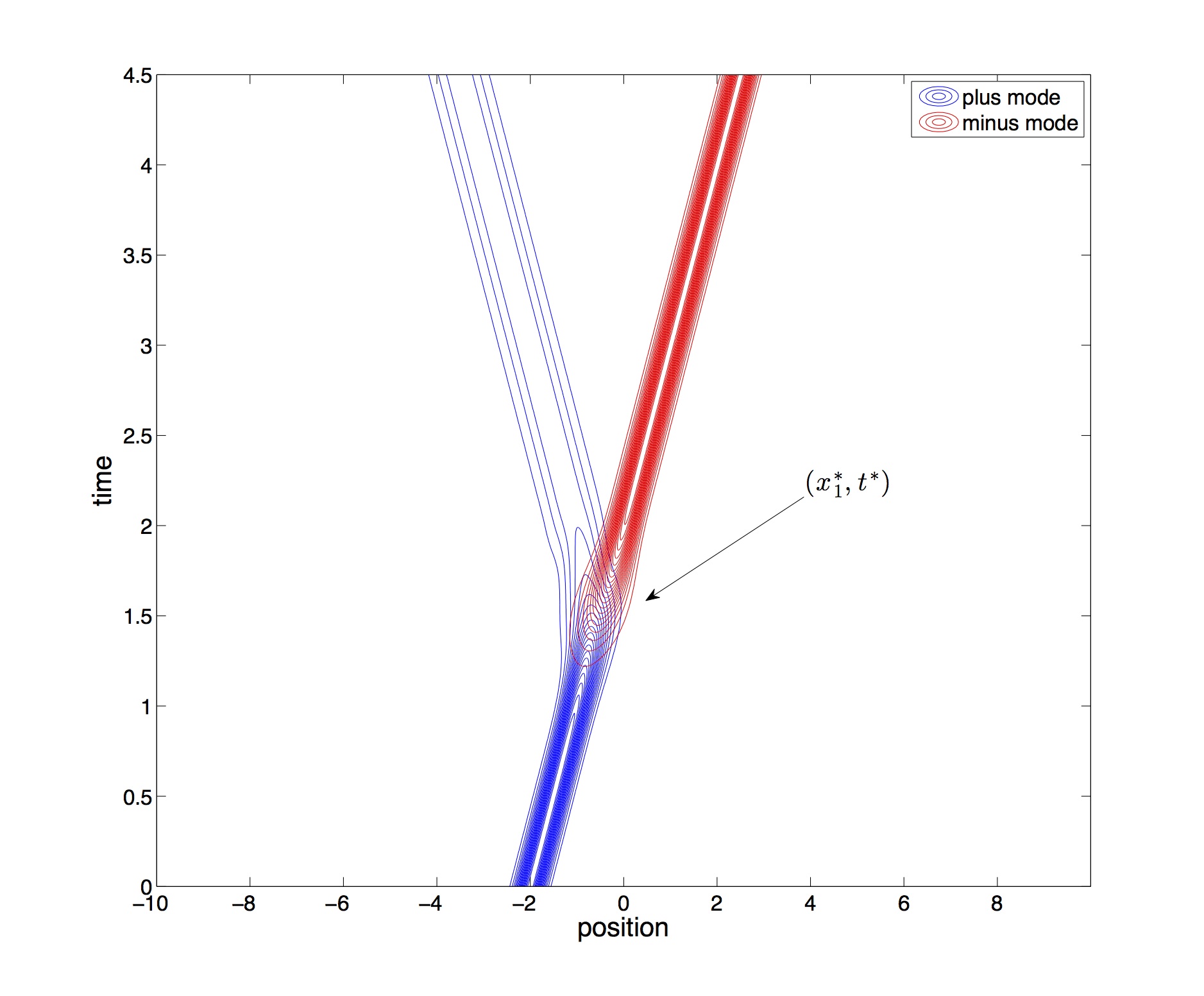}
    \caption{Contour plots of the plus and minus densities, quantum graphene model with $\eps=0.064$.}
\label{fig5bis}
\end{figure}
When the plus particles hits the potential barrier at the position $x_1^*$ and at time $t^*$, its momentum $\xi_1$ vanishes and a transition occurs: the mass is partially tranferred to a minus particle, with the transfer rate $T_\eps$. This phenomenology can be observed on Figure \ref{fig5bis} where, for $\eps=0.064$, the contour plots of plus and minus wavepackets are represented, computed with the quantum graphene model.  In order to check the formula \eqref{def:Teps}, we give in Table \ref{table2} the numerical transfer rate $T_\eps$ (i.e. the ratio of the transferred mass over the initial mass) and the quantity $-\frac{\pi (\xi_2^0)^2}{\eps \log T_\eps}$ (theoretically equal to $V'(x_1^*)$) for 6 values of $\eps$, computed with the quantum model and with the kinetic model. Note that $x_1^*$ is defined by $V_1(x_1^*)+\xi_2^0=V_1(x_1^0)+|\xi^0|$, which gives $x_1^*\approx -0.61$ and $V'_1(x_1^*)\approx 0.94$. The numerical results corroborate the predicted rates.
\begin{table}[h!]
\label{table2}
\caption{Transfer rates}\centering
\begin{tabular}{|c|c|c|c|c|c|c|}
\hline
$\eps$&0.128&0.064&0.032&0.016&0.008&0.004\\
\hline\hline
$T_\eps$ (quantum)&0.772&0.596&0.355&0.126&1.60$\times 10^{-2}$&2.55$\times 10^{-4}$\\
\hline
$T_\eps$ (kinetic)&0.769&0.593&0.351&0.123&1.52$\times 10^{-2}$&2.38$\times 10^{-4}$\\
\hline\hline
$-\frac{\pi (\xi_2^0)^2}{\eps \log T_\eps}$ (quantum)&0.9481&0.9486&0.9489&0.9491&0.9491&0.9492\\
\hline
$-\frac{\pi (\xi_2^0)^2}{\eps \log T_\eps}$ (kinetic)&0.9364&0.9387&0.9384&0.9383&0.9385&0.9415 \\
\hline\hline
\hline
\end{tabular}
\end{table}

But the rate $T_\eps$ is not the only parameter appearing in the transition phenomenon. We highlight the importance of the jump operator $J_\pm$, designed in order to ensure the energy conservation during the transition process. On Figure \ref{fig5}, we plot the $L^1$ error between the densities  for the quantum model and the densities for the kinetic model, with and without the jump. For clarity, we plot separately the error on the plus density (continuous red and blue lines) and the error on the minus density (dashed red and blue lines). The final time of the simulation is  still $t_f=4.5$. 

We observe on the minus density curves (dashed lines) that the jump process improves significantly the precision of the computation of the post-transition density: without the jump operator and for $\eps>10^{-2}$, the error on the minus density is of order 1. Note that, for $\eps<4\times10^{-3}$, the transfer rate becomes negligible and the major part of the error becomes the error made on the transport process, studied above: the plus density curves (continuous curves) are comparable to the red curve of Figure \ref{fig4} (the scales in these two figures are the same).
\begin{figure}[!htbp]
  \includegraphics[width=.7\textwidth]{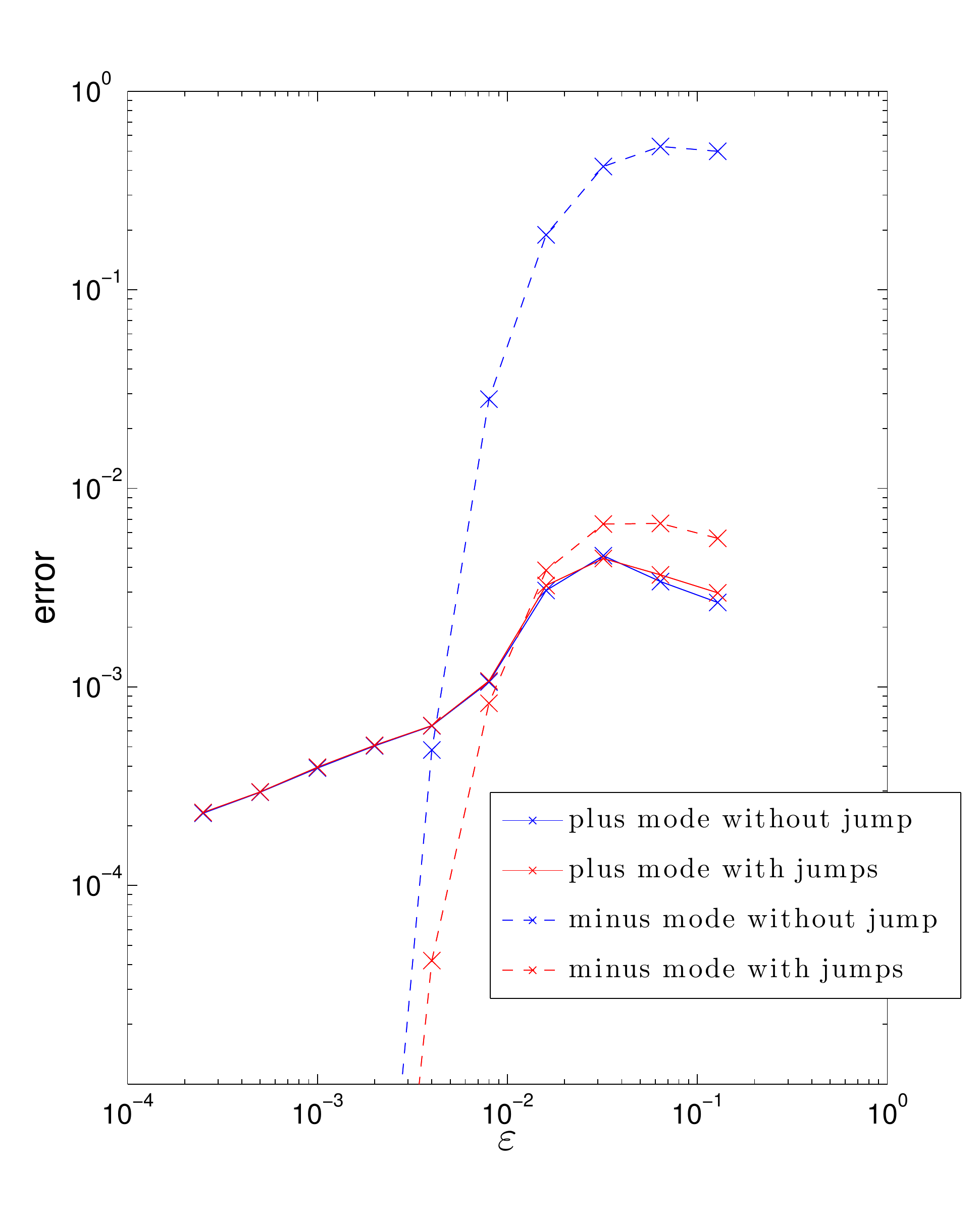}
    \caption{$L^1$ error on the densities  for the kinetic graphene model at time $t_f=4.5$, as a function of $\eps$, with and without the jump process.}
\label{fig5}
\end{figure}

\subsection{$N$-particles system}
\label{sub4}
In this last numerical experiments, we take the initial data as a mixture of coherent quantum states, with a double bump distribution function, i.e. the matrix density is \eqref{matinit} with
\begin{align*}
&&f^0(x^0,\xi^0)=\max\left(\cos(2\pi x)\un_{-7/4\leq x\leq -5/4}\,,\,\cos(2\pi x)\un_{-11/4\leq x\leq -9/4}\right)\qquad \\
&&\times(5/\pi)^{-1/2}\exp\left(-5(\xi_1-1.3)^2\right).
\end{align*}
The potential is $V_1$ defined in \eqref{V1V2}. We plot on Figure \ref{fig6} the plus and minus densities computed with the quantum and the kinetic graphene model, for $\eps=0.016$ at times $t=0,\,1.8,\,2.7$ and $4.5$ . The kinetic model is still used with $4\times 10^6$ particles, and the quantum model is now used with 5000 wavefunctions.
\begin{figure}[!htbp]
  \centerline{
  \subfigure[$t=0$]{\label{fig8a}\includegraphics[width=.6\textwidth]{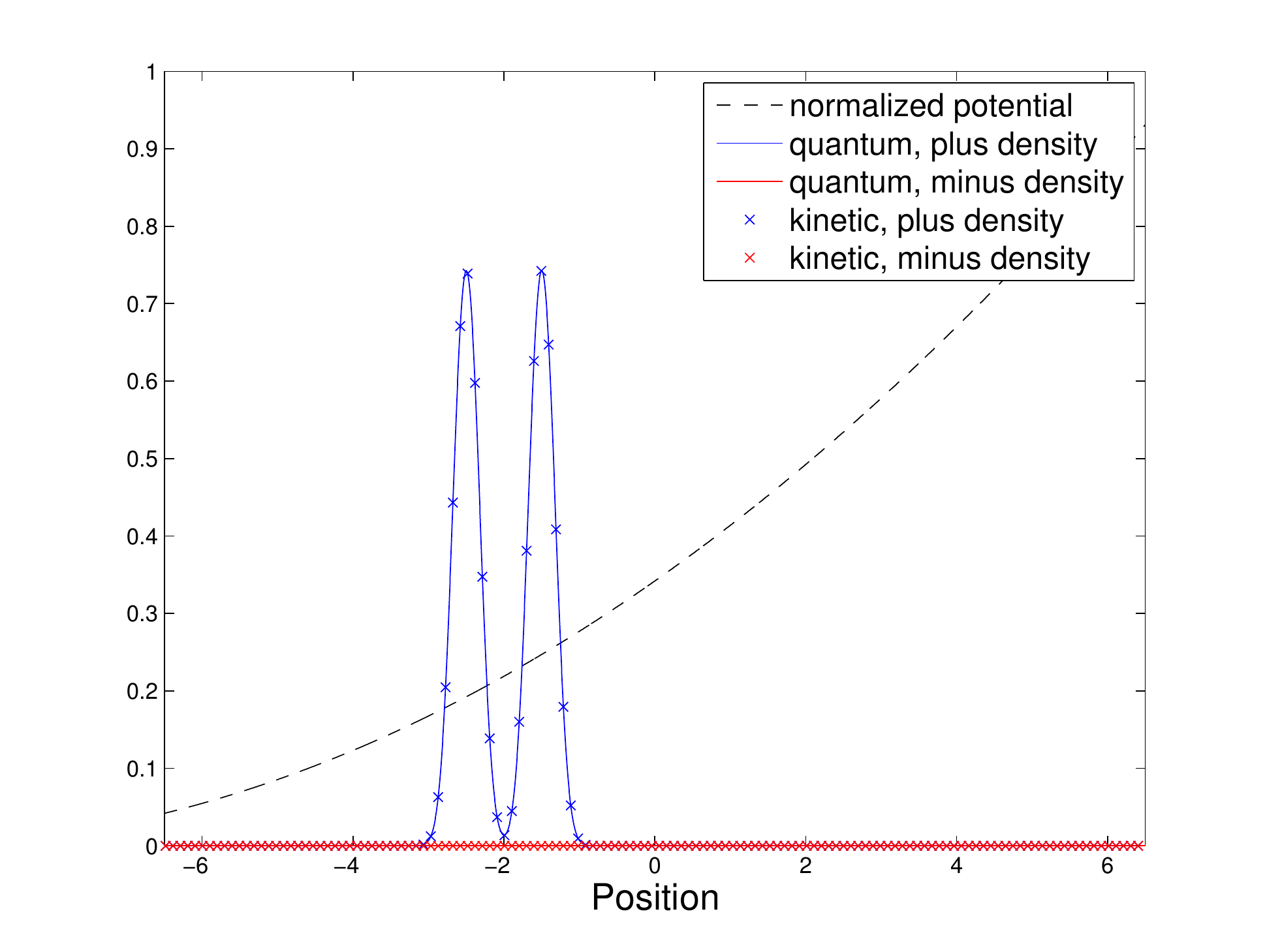}}\hspace*{-8mm}
  \subfigure[$t=1.8$]{\label{fig8b}\includegraphics[width=.6\textwidth]{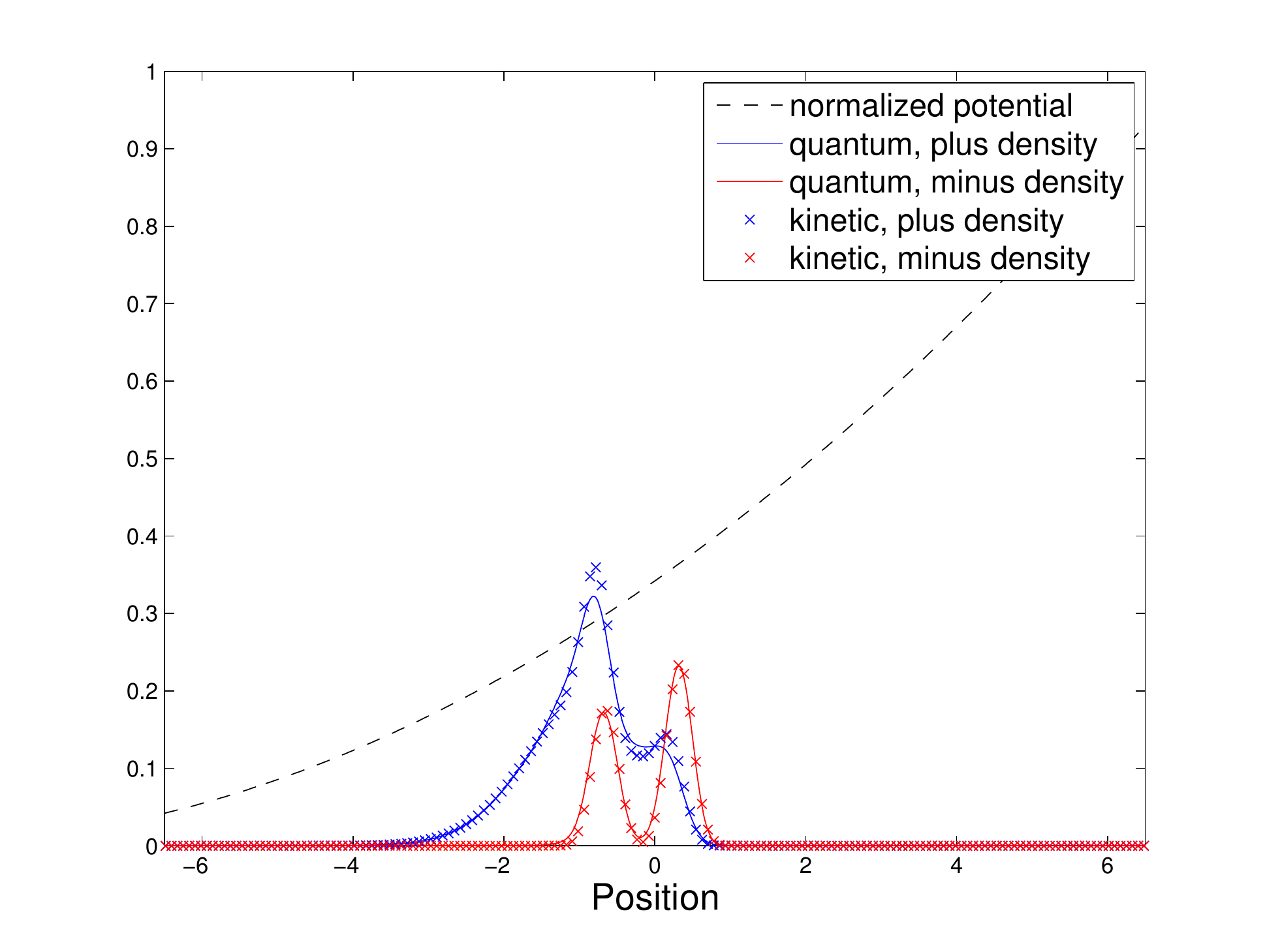}}}
  \centerline{\subfigure[$t=2.7$]{\label{fig8c}\includegraphics[width=.6\textwidth]{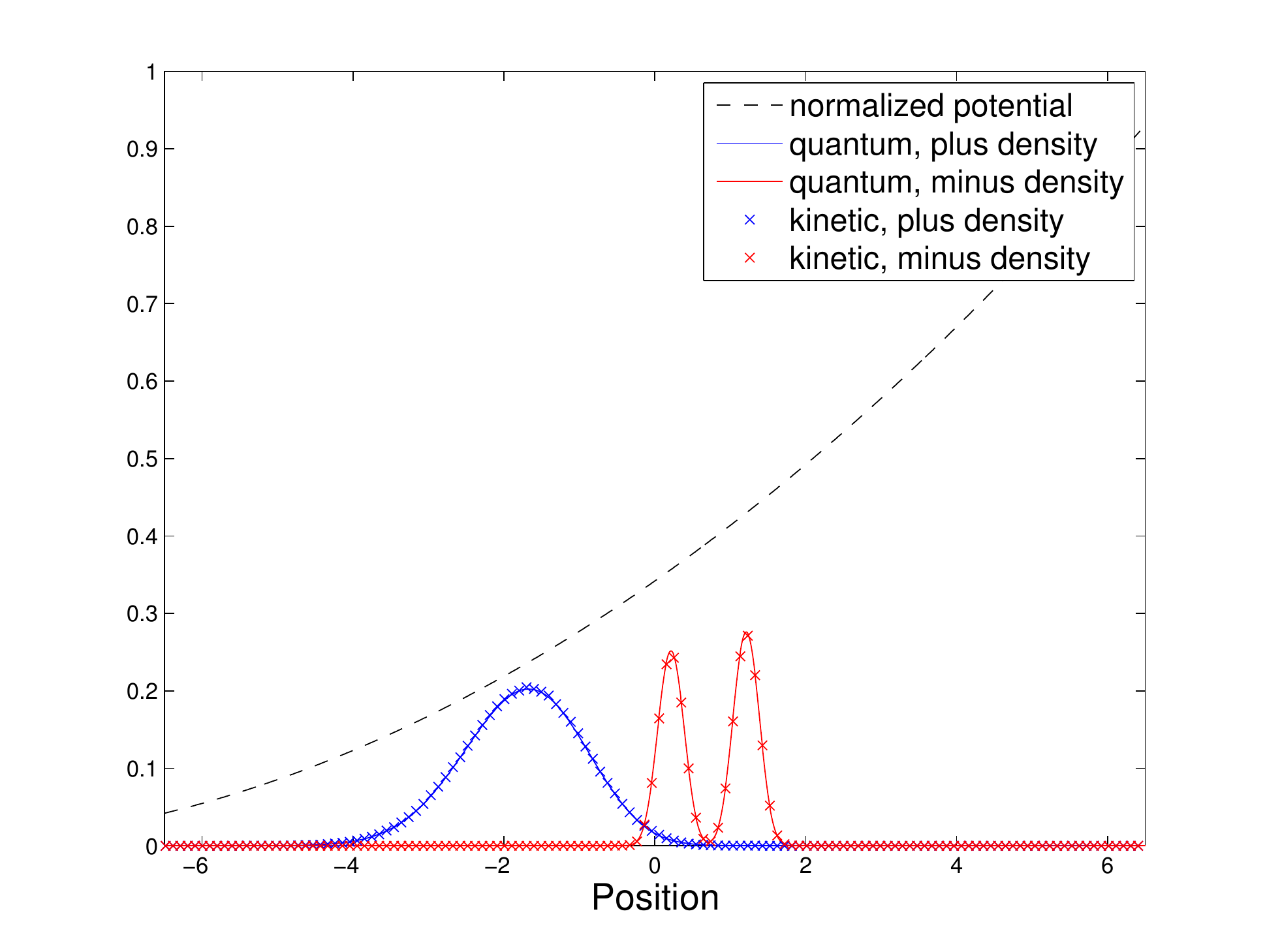}}\hspace*{-8mm}
  \subfigure[$t=4.5$]{\label{fig8d}\includegraphics[width=.6\textwidth]{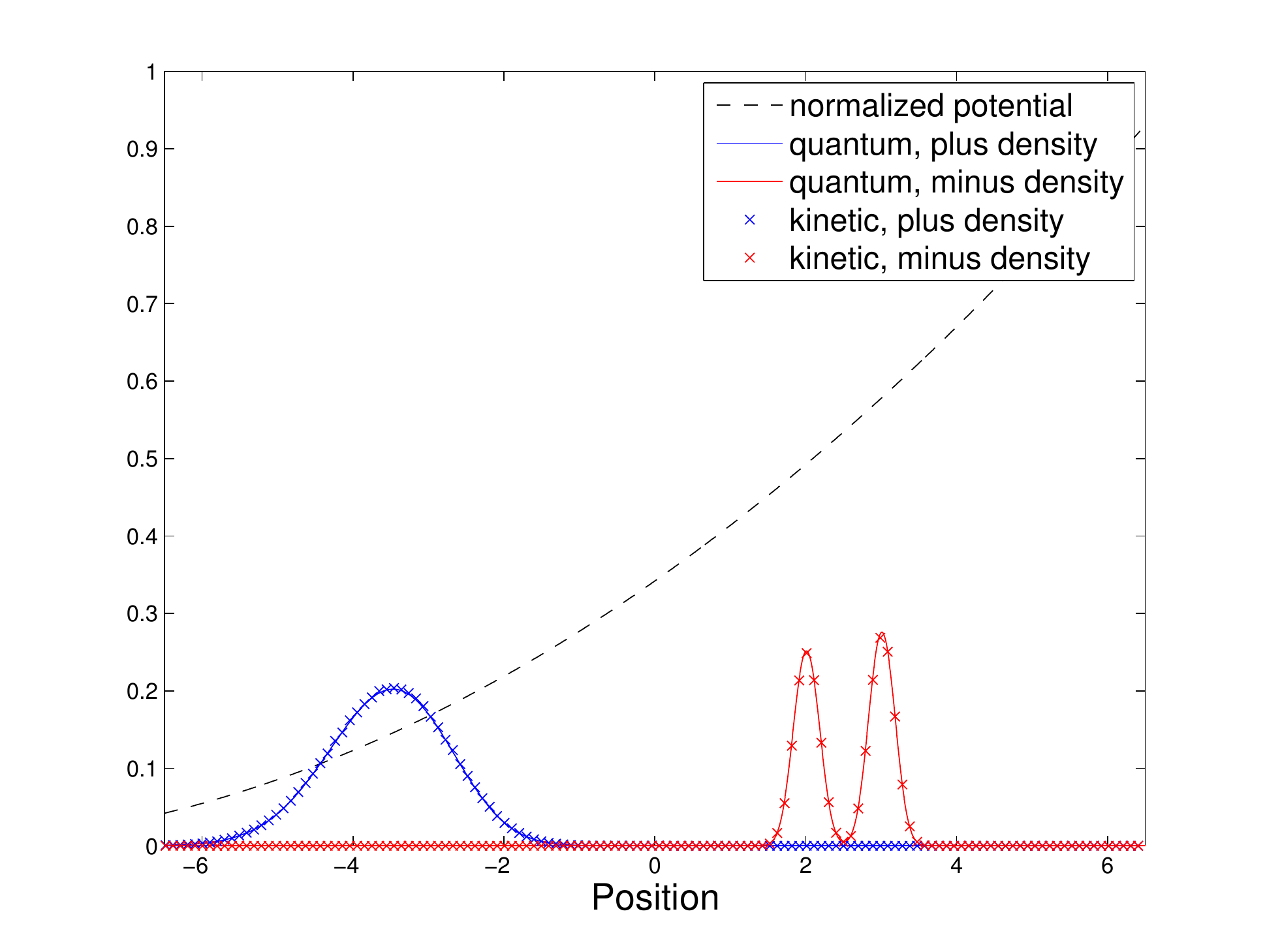}}}
    \caption{Propagation of a mixture of coherent wavepackets through a barrier for $\eps=0.016$.}
\label{fig6}
\end{figure}

Finally, in Table \ref{table3}, we provide the numerical tranfer rates between the plus and minus mixture of coherent states.
\begin{table}[h!]
\label{table3}
\caption{Transfer rates}\centering
\begin{tabular}{|c|c|c|c|c|c|c|}
\hline
$\eps$&0.128&0.064&0.032&0.016&0.008&0.004\\
\hline\hline
Transfer rate (quantum)&0.771&0.595&0.355&0.126&$1.62\times 10^{-2}$&$2.76\times 10^{-4}$\\
\hline
Transfer rate (kinetic)&0.768&0.592&0.350&0.123&$1.54\times 10^{-2}$&$2.57\times 10^{-4}$\\
\hline
\end{tabular}
\end{table}
Again, on Figure \ref{fig6} and in Table \ref{table3}, we observe a very good agreement between both models.


\section{Analytic justification of the algorithm}\label{sec:markov}

As emphasized  in the introduction, the algorithmic representation of the solutions of the kinetic equations~(\ref{eq:kinapp}) comes from a representation of these solutions via a Markov semi-group. We first present  this semi-group in Section~\ref{sec:description} and explain its connection with system~(\ref{eq:kinapp}). In particular, we reduce the proof of Theorem~\ref{theo:main} to a result on this semi-group, which will be the subject of Section~\ref{sec:proof}. For the convenience of the reader, the proofs of the two main results of this section are postponed in Sections~\ref{sec:proof1} and~\ref{sec:proof}, just after Section~\ref{sec:prelim} which is devoted to preliminaries. 

\subsection{The Markov semi-group description}\label{sec:description}
We consider the  Hamiltonian flows
$$\Phi^t_+= \left(x_+^t(x,\xi),\xi_+^t(x,\xi)\right)\;\;{\rm and}\;\;
\Phi^t_-= \left(x_-^t(x,\xi),\xi_-^t(x,\xi)\right)$$ with 
$\left(x^0_{\pm} (x,\xi),\xi^0_{\pm}(x,\xi)\right)=(x,\xi)$ and   
\begin{equation}\label{eq:clastraj}
\left\{\begin{array} l
\displaystyle {d\over dt} x^t_+=  {\xi^t_+ \over |\xi^t_+|}\;\;{\rm and}\;\;{d\over dt} \xi^t_+ = - \nabla V(x^t_+),\\[3mm]
\displaystyle {d\over dt} x^t_-= - {\xi^t_- \over |\xi^t_-|}\;\;{\rm and}\;\;{d\over dt} \xi^t_- =  -\nabla V(x^t_-).
\end{array}\right.
\end{equation}
As long as $\xi\not=0$,  the smoothness of the Hamiltonians $|\xi|\pm V(x)$  yields local existence and uniqueness of the trajectory passing through $(x,\xi)$ for any $x\in\R^2$. However, it may happens that $\xi^t_\pm\Tend{t}{ t^*}0$ for some $t^*\in\R$ and some index $+$ or $-$. If at the point $x^{t^*}_\pm$, the assumption~(\ref{ass:V}) is satisfied (that is if $\nabla V(x^{t^*}_\pm)\not=0$), then one can prove that there exists a unique continuation to the map $t\mapsto \Phi^t_\pm$ when $t>t^*$ (see Proposition~1 in~\cite{FG2} and Proposition~\ref{prop:trajectory} below where a precise statement and a proof are given for the convenience of the reader).
  As a consequence, the assumption~(\ref{ass:V}) guarantees the existence and uniqueness of the solutions to~(\ref{eq:clastraj}). However, these trajectories are no longer smooth  when passing through $\xi=0$; more precisely,  the vector $\dot \Phi^t$ has a discontinuity at $t=t^*$ whenever $\xi^{t^*}_\pm=0$. It is also interesting to notice that
if the latter assumption~(\ref{ass:V}) fails at $x^{t^*}_\pm$, then uniqueness is no longer guaranteed.

\medskip

\begin{remark}\label{rem:Sigma}
For any trajectory $\Phi^t_\pm$, the  quantity $|\xi^t_\pm|$ reaches its minimum when 
$$ {d\over dt} \left( |\xi^t_\pm|^2\right) =- \xi_\pm^t \cdot \nabla V(x^t_\pm)=0,$$
i.e. on points of the set~$\Sigma$ define in~(\ref{def:Sigma}).
\end{remark}

We are now going to introduce a branching process between both types of trajectories.
We  attach the labels $+$ and $-$ to  the phase space and for points 
$$(x,\xi,j)\in\R^{4}_\pm:= \R^{4}\times\{+1,-1\},$$ we consider trajectories
$$
{\mathcal T}_{\eps}^{(x,\xi,j)}:[0,+\infty)\rightarrow \R^{4}_\pm,
$$
which combine deterministic classical transport and random jumps 
between the levels at the manifold $\Sigma$. 
More precisely,
we set 
$$
{\mathcal
T}_{\eps}^{(x,\xi,j)}(t)=\left(\Phi^t_j(x,\xi),j\right)$$
 as long as
$\Phi^t_{j}(x,\xi)\not\in \Sigma$. Whenever the deterministic
flow $\Phi^t_j(x,\xi)$ hits the manifold~$\Sigma$ at a point
$(x^*,\xi^*)$, a random jump from 
$$(x^*,\xi^*,j)\;\;{\rm  to}\;\; \left(J_j(x^*,\xi^*),-j\right)$$ occurs with
probability $T^\eps(x^*,\xi^*)$.

\medskip

The jump aims at preserving at order $\mathcal O(\eps)$  the energy of the trajectory for points~$\xi$ where the transfer coefficient is relevant, that is points where $|\xi|\leq R\sqrt\eps$ according to Remark~\ref{UepsR}.  This is an important ingredient of the proof (see Remark~\ref{rem:energy}). Indeed, set 
$$E_\pm(x,\xi):=\pm|\xi|+V(x),$$
then, if $|\xi|\leq R\sqrt\eps$, at a jump from~$+$ to~$-$, we have 
$$
E_-\left(J_+(x,\xi)\right)  
 =  E_+(x,\xi)+\mathcal O(|\xi|^2)=E_+(x,\xi)+\mathcal O(R^2\eps). $$
Note that the importance of the jump has been illustrated numerically in~\cite{FL3} in the context of molecular propagation ; this jump was already performed in~\cite{HJ} for the construction of gaussian wave packets which are approximated solutions of a Schr\"odinger equation with matrix valued potential presenting a conical intersection.

\medskip

Since $\nabla V(x)\not=0$, the trajectories which reach the manifold~$\Sigma$ arrive there transversally to~$\Sigma$. As a consequence, in each  bounded time interval~$[0,T]$,
each path
$$
(x,\xi,j)\to {\mathcal T}_{\eps}^{(x,\xi,j)}(t)
$$
has a finite number of jumps and remains in a bounded region of the phase space~$\R^{4}_\pm$. Besides, away from
the jump manifold $\Sigma\times\{-1,+1\}$ each path is smooth.

\medskip 

Following \cite{dynkin,liggett}, we define the function $P_{\eps}(x,\xi,j;t,\cdot)$ as the function which associate to a measurable set $\Gamma\subset \R^4_\pm$  the probability $P_{\eps}(x,\xi,j;t,\Gamma)$  of
being at time $t$ in  $\Gamma$ having started in $(x,\xi,j)$. And we define a 
 time-dependent Markov
process $(\mathcal L_{\eps}^t)_{t\ge0}$ acting on bounded
measurable scalar functions $f: \R^{4}_\pm\to\C$ by
$$
({\mathcal L}_{\eps}^t\,f)(x,\xi,j):=
\int_{\R^{4}\times\{-1,+1\}}
f(q,p,k)\,P_{\eps}(x,\xi,j;t,d(q,p,k)).
$$
An explicit expression of ${\mathcal L}_\eps ^t$ is written on short interval times close to jump points in Section~\ref{sec:just} (see equations~(\ref{eq:g+}) and~(\ref{eq:g-}) below). 

\medskip

In order to define its action on Wigner functions, we need to identify pairs of functions $(a_+,a_-)$ with some function $a$ on $\R^4_\pm$, which is done by the identification 
\begin{equation}\label{eq:ident}
a(x,\xi,\pm1)=a_\pm(x,\xi),\;\;\forall (x,\xi)\in\R^4.
\end{equation}
Through this identification, the action of $({\mathcal L}^t_{\eps})_{t\ge0}$ on ${\mathcal C}_0^\infty(\R^4\setminus\{\xi=0\})$ is  given  by 
$$
({\mathcal L}_{\eps}^t \,a)(x,\xi):=\left(
\left({\mathcal L}^t_{\eps} a\right) (x,\xi,+1)\,,
\,\left({\mathcal L}^t_{\eps} a\right) (x,\xi,-1)\right),
$$
We extend this action to ${\mathcal D}'\left(\R^4\setminus\{\xi=0\}\right)$ by duality by setting 
$$\forall f\in{\mathcal D}'(\R^4),\;\;\forall a\in{\mathcal C}_0^\infty(\R^4\setminus \{\xi=0\}),\;\;  \langle ({\mathcal L}^t_{\eps} f)(\pm1),a \rangle= \langle f,({\mathcal L}^t_{\eps} a)(\pm1)\rangle.$$

\begin{proposition}[Resolution of the kinetic model]\label{prop:Markov}
Set $w^\eps(0)=\left( w^\eps_+(0),w^\eps_-(0)\right)$ and 
assume that (1) and (2) of Assumptions~\ref{ass} are satisfied. Then 
the function
$f^\eps(t,x,\xi)=({\mathcal L}_{\eps}^t w^\eps(0))(x,\xi)$ is the unique solution to
system~(\ref{eq:kinapp}) in ${\mathcal D}'\left(\R^4\setminus\{\xi=0\}\right)$.
\end{proposition}

Then, 
  Theorem~\ref{theo:main} is a  corollary of the following proposition.

\begin{proposition}[Approximation by the semi-group]\label{prop:FL}
Set $w^\eps(0)=\left( w^\eps_+(0),w^\eps_-(0)\right)$ and 
assume for $\chi\in{\mathcal C}_c^\infty([0,T],\R)$ and $a\in{\mathcal C}_0^\infty(\R^4\setminus\{\xi=0\})$, Assumptions~\ref{ass} are satisfied,  then,
there exist positive constants $C, \eps_0>0$ such that for all
$0<\eps<\eps_0$, 
\begin{equation*}
 \qquad\left| {\rm
tr}\int_{\R^{2d+1}}\,\chi(t)
\left(w^\eps_{\pm}(t,x,,\xi)-({\mathcal L}_{\eps}^t w^\eps(0))(x,\xi,\pm 1)
\right)\,a(x,\xi) \, d x\,d \xi\,d t
\right| \le C \,\eps^{1/8}.
\end{equation*}
\end{proposition}

\begin{remark}\label{rem:(3)}
Note that the hypothesis (3) of Assumptions~\ref{ass} imply that on the interval $[0,T]$, the trajectories which reach the support of $a$ has performed at most one jump. 
\end{remark}

The following Section~\ref{sec:prelim} states preliminary results and the two next subsections are devoted to the proof of Propositions~\ref{prop:Markov} and~\ref{prop:FL} respectively.


\subsection{Preliminaries}\label{sec:prelim}

In this section, we begin with a careful analysis of the geometry close to a point of $\Sigma$ in order to precise the setting in which the proofs will be performed.
 
\subsubsection{The generalized flow}\label{sec:prelimflow}
In this section, we gather some properties of  the flows~$\Phi^t_\pm$  that will be useful in the next sections. We first focus on the existence and uniqueness of the generalized trajectories of the Hamiltonian vector fields $H_\pm$ and recall the arguments of the proof given in~\cite{FG1}.

\begin{proposition}\label{prop:trajectory}
For any $x_0\in\R^2$ such that $\nabla V(x_0)\not=0$, there exists $\tau_0>0$ and a unique Lipschitz continuous map
 $$t\mapsto \left(x^t_\pm (x_0,0),\xi^t_\pm(x_0,0)\right),\;\; t\in[-\tau_0,\tau_0]$$
satisfying~(\ref{eq:clastraj}) for $t\not=0$ and such that $x^0_\pm(x_0,0)=x_0$, $\xi^0_\pm(x_0,0)=0$ and 
$$\displaylines{\dot x^t_\pm(x_0,0)\Tend{t}{0}\mp {\nabla V(x_0)\over |\nabla V(x_0)|},\;\;\dot \xi^t_\pm(x_0,\xi_0)\Tend{t}{0}- \nabla V_0(x_0).\cr}$$
\end{proposition}

\begin{corollary}\label{cor:fields}
With the notations of Proposition~\ref{prop:trajectory}, we have
 $$\displaylines{
 \lim_{t\rightarrow 0^-} H_+(x^t_\pm,\xi^t_\pm) = \lim_{t\rightarrow 0^+}  H_-(x^t_\pm,\xi^t_\pm) =-\nabla V(x_0)\cdot\nabla_\xi+{\nabla V(x_0)\over |\nabla V(x_0)|} \cdot \nabla _x,\cr
  \lim_{t\rightarrow 0^+} H_-(x^t_\pm,\xi^t_\pm) = \lim_{t \rightarrow 0^-}  H_-(x^t_\pm,\xi^t_\pm) =-\nabla V(x_0)\cdot\nabla_\xi-{\nabla V(x_0)\over |\nabla V(x_0)|} \cdot \nabla_ x.\cr
 \cr}$$
\end{corollary}

In the following and with the notations of Proposition~\ref{prop:trajectory}, we shall set 
\begin{eqnarray}\label{def:H}
H(x_0)&:=& \lim_{t\rightarrow 0^-} H_+(x^t_\pm,\xi^t_\pm) = \lim_{t\rightarrow 0^+}  H_-(x^t_\pm,\xi^t_\pm) ,\\ 
\label{def:H'}
H'(x_0)&:=&    \lim_{t\rightarrow 0^+} H_-(x^t_\pm,\xi^t_\pm) = \lim_{t\rightarrow 0^-}  H_-(x^t_\pm,\xi^t_\pm) .
\end{eqnarray}
We observe that if $\omega=d\xi\wedge dx$ is the canonical skew-symmetric $2$-form of  the cotangent space of $\R^2$, we have 
\begin{equation}\label{eq:sigmaHH'}
\omega(H,H')=2|\nabla V(x_0)|>0.
\end{equation}

\begin{proof}[Proof of Proposition \ref{prop:trajectory}]
Following  Proposition~3 in~\cite{FG1}, we introduce two flows $\Psi_1^t=(x_1^t,\xi_1^t)$ and $\Psi_2^t=(x_2^t,\xi_2^t)$ which are defined by
$$\left\{
\begin{array}{l}
\dot x_j^t=(-1)^{j+1}{\rm sgn}(t){\xi_j^t\over|\xi^t_j|},\;\; x_j^0=x_0,\\
\dot \xi_j^t=-\nabla V(x_j^t),\end{array}\right.$$
or equivalently 
$$\xi_j^t=-t\int_0^1 \nabla V(x_j^{ts})ds,\;\; x_j^t=x_0+(-1)^{j+1}\int_0^t
{\int_0^1 \nabla V(x_j^{s\sigma})d\sigma \over
\left| \int_0^1 \nabla V(x_j^{s\sigma})d\sigma\right| }ds.$$
The last system can be solved on short time by a fixed point argument in an open subset of $\{\nabla V(x)\not=0\}$ and the resulting map $x\mapsto\Psi_j^t(x,0)$ is smooth. 
As a consequence, 
 there exists a  neighborhood $\Omega$ of $(x_0,0)$ such that $\Omega\subset  \{\nabla V(x)\not=0\}$ and $\tau_0>0$ such that 
for all $(x,0)\in \Omega$, the maps
$$\displaylines{
\left(\Phi^t_+(x,0)\right)_{ t\in[0,\tau_0]}=\left(\Psi^t_1(x,0)\right)_ {t\in[0,\tau_0]}\;{\rm and} \;\left(\Phi^t_+(x,0)\right)_{ t\in[-\tau_0,0]}=\left(\Psi^t_2(x,0)\right)_{ t\in[-\tau_0,0]},\cr
\left(\Phi^t_-(x,0)\right)_{ t\in[0,\tau_0]}=\left(\Psi^t_2(x,0)\right)_ {t\in[0,\tau_0]}\;{\rm and} \;\left(\Phi^t_-(x,0)\right)_{ t\in[-\tau_0,0]}=\left(\Psi^t_1(x,0)\right)_ {t\in[-\tau_0,0]},\cr}$$
solve our problem. The flows $\Phi^t_\pm$ (which are well defined for $\xi\not=0$) extend to Lipschitz continuous maps
$$t\mapsto \Phi^t_\pm(x,\xi),\;\;t\in[-\tau_0,\tau_0],\;\;(x,\xi)\in\Omega.$$
\end{proof}

\begin{remark}
Following the arguments of Section~6.2 in~\cite{FGL}, one can prove that for $|t|<\tau_0$ and  $\alpha\in \N^d$, the maps $(x,\xi)\mapsto \partial_x^\alpha\Phi^t_\pm(x,\xi)$ are  continuous maps on $\Omega$ with bounded locally integrable time derivatives $\partial_t\partial_x^\alpha \Phi^t(x,\xi)$.
\end{remark}

\medskip

 \subsubsection{Local analysis}
 In what follows, 
we shall associate with points   $(x,\xi)$, which are close enough to the set~$\Sigma$, a number~$\tau_\pm(x,\xi)$ which is the time that  separates~$(x,\xi)$ from the point of  the trajectory~$\Phi^t_\pm(x,\xi)$ which belongs to~$\Sigma$. 

\begin{proposition}\label{reg:tau} Let $x_0\in\R^2$ such that $\nabla V(x_0)\not=0$, there exists an open set $\Omega\subset\{\nabla V(x)\not=0\}\subset \R^4_{(x,\xi)}$ containing $(x_0,0)$ and such that the relations $\Phi^{\tau_\pm(x,\xi)}(x,\xi)\in\Sigma$ define two
continuous functions on $\Omega$
$$( x,\xi) \mapsto \tau_\pm(x,\xi).$$
\end{proposition}

\begin{remark}
By definition of $\tau_\pm(x,\xi)$, we have 
$\nabla V(x_\pm^{\tau_\pm(x,\xi)})\cdot \xi_\pm^{\tau_\pm(x,\xi)}=0$ for any $(x,\xi)\in\Omega$.
\end{remark}

\begin{proof}
Let  us study the plus mode.  We observe that
\begin{equation}\label{eq:d^2}
{d\over dt} |\xi^t_+|^2= -\nabla V(x^t_+)\cdot \xi^t_+\;\;{\rm and}\;\;{d^2\over dt^2} |\xi^t_+|^2=| \nabla V(x^t_+)|^2-d^2 V(x^t_+) \xi^t_+ \cdot{ \xi^t_+\over |\xi^t_+|}
\end{equation}
Since $\nabla V(x_0)\not=0$, we can find a neighborhood $U$ of $(x_0,0)$, $c_0>0$ and $\tau_1>0$ such that 
\begin{equation}\label{defc0}
\forall (x,\xi)\in\Omega_1,\;\;\forall |t|\leq \tau_1,\;\; {d^2\over dt^2}| \xi^t_+ (x,\xi)|^2\geq c_0>0.
\end{equation}
Because of~(\ref{defc0}),  the  map $t \mapsto |\xi^t_+|$ reaches its minimum at most once in $U$. With any $(x,\xi) \in U$, we associate an interval $[t_{i}(x,\xi),t_{f}(x,\xi)]$  of maximal size such that $\Phi^t(x,\xi)\in U$  for all $t\in] t_i(x,\xi),t_f(x,\xi)[$. Because of the first relation of~(\ref{eq:d^2}), we are interested to the times where the curves $\Phi^t(x,\xi)$ crosses the hypersurface $\Sigma$, which happens at most once in $U$.  We set 
$$\Omega^+=\left\{(x,\xi)\in U,\;\exists (y,\eta)\in\Sigma\cap U,\;\exists s\in]t_{i}(y,\eta),t_{f}(y,\eta)[,\;(x,\xi)=\Phi^s_+(y,\eta)\right\}.$$
Then $\Omega^+$ is a neighborhood of $(x_0,0)$ 
 included in $U$ and such that the relation $\Phi^{\tau_+(x,\xi)}_+(x,\xi)\in\Sigma$ defines a map $\tau_+$ from $\Omega^+$ into $\R$.  We define similarly $\Omega^-$ and the map~$ \tau_-(x,\xi)$ and we choose $\Omega=\Omega^+\cap\Omega^-$ (see Fig. \ref{figure 1}). 

\begin{figure}[!htbp]
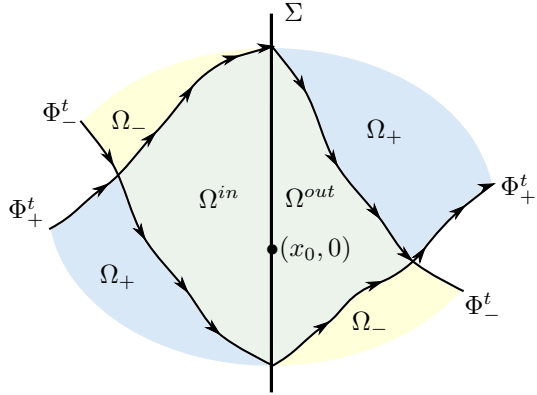
  \caption{Domains $\Omega_+$ (blue+green), $\Omega_-$ (yellow+green) and $\Omega$ (green) $=\Omega_+\cap\Omega_-=\Omega^{in}\cup\Omega^{out}$}
\label{figure 1}\end{figure}

\noindent Finally, it is classical to prove that $\tau_+$ is a continuous map on $\Omega$. We set 
$$z=(x,\xi)\;\;{\rm and}\;\; \phi(t,z)={d\over dt} |\xi^t_+|^2= -\nabla V(x_+^t(z))\cdot \xi_+^t(z)$$
 and we argue by contradiction. We assume 
that there exist $\alpha_0$, $z_0$ and a sequence~$(z_n)_{n\in\N^*}$ going to $z_0$ as $n$ goes to $+\infty$ and such that
$$|\tau_+(z_n)-\tau_+(z_0)|>\alpha_0.$$ Then, by the continuity of $\phi$, we get $\phi(\tau_+(z_0),z_n)\Tend{n}{+\infty} 0$ and since $\partial_t\phi\geq c_0$  (by~(\ref{defc0})), we have
$$c_0\left| \tau_+(z_0)-\tau_+(z_n)\right|\leq  \left| \phi(\tau_+(z_0),z_n)-\phi(\tau_+(z_n),z_n)\right|=|\phi(\tau_+(z_0),z_n)|.$$
Taking $n$ large enough, we get $|\tau_+(z_n)-\tau_+(z_0)|<\alpha_0/2$, whence a contradiction and the continuity of the map $z\mapsto \tau_+(z)$ is proved.\\
One argues similarly for the mode minus. \end{proof}

\begin{remark} The hypersurface $\Sigma$ parts $\Omega$ into two distinct connected regions  (see Fig. \ref{figure 1}):
\begin{eqnarray}\label{def:Omegainout}
{ \Omega}^{out}&:=&\{\tau_+<0, \; \tau_-<0\}\cap{\Omega}=\{\nabla V(x)\cdot \xi <0\}\cap \Omega\\
\nonumber
{\Omega}^{in}&:= &\{\tau_+>0,\;\tau_->0\}\cap\Omega=\{\nabla V(x)\cdot \xi >0\}\cap \Omega.
\end{eqnarray}
Indeed, if $\nabla V(x)\cdot\xi >0$, we have simultaneously, 
${d\over dt} |\xi^t_+|^2 <0$ and ${d\over dt} |\xi^t_-|^2 <0$, thus we have $\tau_+(x,\xi)>0$ and $\tau_-(x,\xi)>0$.
 Besides, for $(x,\xi)\in\Omega^{in}$, we have $\Phi^t(x,\xi)\in\Omega^{in}$ for $t\in[0,\tau_+(x,\xi)[$.
\end{remark}

\medskip

\subsubsection{Trace maps} We consider the open set $\Omega$ of Proposition~\ref{reg:tau} where $H^\pm$ are transverse to $\Sigma$. 
We can find some times $t_i$ and $t_f$, and four open sets $\Sigma_i^{\pm}$ and $\Sigma_f^\pm$, included in $\Omega$ and such that 
$\Sigma_f^\pm=\Phi^{t_f-t_i}(\Sigma_i^\pm)$. We set 
$${\mathcal V}=\{(t,\Phi^{t_f-t}_+(x,\xi),\;(x,\xi)\in\Sigma_f^+\}\cup\{(t,\Phi^{t_f-t}_-(x,\xi),\;(x,\xi)\in\Sigma_f^-\}.$$
We can assume that ${\mathcal V}\subset [t_i,t_f]\times \Omega$.
Because of the geometry of the trajectories, it is enough to study the equation in~
${\mathcal V}$. We have a partition of ${\mathcal V}$ as 
$${\mathcal V}=\widetilde\Sigma \cup {\mathcal V}^{in} \cup {\mathcal V}^{out}$$
with $\widetilde \Sigma = (\R\times\Sigma)\cap{\mathcal V}\subset [t_i,t_f]\times \Sigma$, ${\mathcal V}^{in}\subset \{\nabla V(x)\cdot\xi >0\}$ and ${\mathcal V}^{out}\subset \{\nabla V(x)\cdot \xi <0\}$.
We can also write 
\begin{eqnarray*}
{\mathcal V}^{in} & = & \{(t,x,\xi)\in {\mathcal V},\;\;\exists (s,y,\eta,j)\in[t,t_f]\times\Sigma\times\{-1,+1\},\;\;(x,\xi)=\Phi_j^{t-s}(y,\eta)\},\\
{\mathcal V}^{out} & = & \{(t,x,\xi)\in {\mathcal V},\;\;\exists (s,y,\eta,j)\in[t_i,t]\times\Sigma\times\{-1,+1\},\;\;(x,\xi)=\Phi_j^{t-s}(y,\eta)\}.
\end{eqnarray*}
\begin{lemma}\label{lem:trace}
With the above notations, let $f^\eps_\pm$ be two functions which are invariant by $\Phi^s_\pm$ in ${\mathcal V}^{in}$.
Assume that $f^\eps_\pm(t_i)$ is supported in $\Sigma_i^\pm$ and $f^\eps_\pm(t_i)\in L^1(\R^4)$. Then $f^\eps_\pm$ have traces on $\widetilde\Sigma$ that we denote by $(f^\eps_\pm)_{\Sigma,in}$ and $(f^\eps_\pm)_{\Sigma,in}\in L^1(\widetilde \Sigma)$.\\
Similarly, if $f^\eps_\pm$ are invariant by $\Phi^s_\pm$ in ${\mathcal V}^{out}$. Assume that $f^\eps_\pm(t_f)$ is supported in $\Sigma_f^\pm$ and $f^\eps_\pm(t_f)\in L^1(\R^4)$. Then $f^\eps_\pm$ have traces on $\widetilde \Sigma $ that we denote by $(f^\eps_\pm)_{\Sigma,out}$ and $(f^\eps_\pm)_{\Sigma,out}\in L^1(\widetilde \Sigma)$.
\end{lemma}
\begin{proof}
We first use the functions $\tau_\pm$ for defining the four following continuous maps 
\begin{eqnarray*}
\kappa^{i}_{\pm}&: & \Sigma_i^\pm\rightarrow\widetilde \Sigma,\\
& &  (x,\xi) \mapsto \left( t_i+\tau_\pm(x,\xi), \Phi^{\tau_\pm(x,\xi)}_\pm(x,\xi)\right),
\end{eqnarray*}
and 
\begin{eqnarray*}
\kappa^{f}_{\pm}&: & \Sigma_{f}^\pm\rightarrow\widetilde \Sigma,\\
& &  (x,\xi) \mapsto\left(t_f+\tau_\pm(x,\xi), \Phi^{\tau_\pm(x,\xi)}_\pm(x,\xi)\right).
\end{eqnarray*}
These four maps are homeomorphisms. Set 
$$(f^\eps_{\pm})_{\Sigma,in}= (\kappa^{i}_\pm)_*f^\eps_{\pm}(t_i).$$
These two functions are $L^1$ functions of $\widetilde \Sigma$ which are the entering traces of 
$f^\eps_{\pm}$ on $\R\times \Sigma$.
We argue similarly for the out-going traces.
\end{proof}

\subsection{Proof of Proposition~\ref{prop:Markov}}\label{sec:proof1}
In this section we prove the existence and uniqueness of solutions to the kinetic equations~(\ref{eq:kinapp}) and we analyze the link between the Markov semi-group and this system of equations. 
\subsubsection{Uniqueness of solutions to~(\ref{eq:kinapp})}\label{sect:uniqueness}
 As a corollary of the analysis of the previous subsection, we obtain that the solutions of~(\ref{eq:kinapp}) are unique if they do exist. 
Indeed, by (1) of Assumptions~\ref{ass}, the data has been chosen supported outside $\Sigma$ so that for short time the kinetic system~(\ref{eq:kinapp}) reduces to classical transport by the two flows. We cut the data in a sum of compactly supported pieces so that each of these pieces satisfy the hypothesis of Lemma~\ref{lem:trace} after a certain amount of time. We only need to consider one of these pieces and we take the notations of Lemma~\ref{lem:trace}. 
Because of the previous decomposition of the data,  we may assume that $f^\eps_\pm(t_i)$ is compactly supported in $\Sigma_{in}^\pm$. 
By Lemma~\ref{lem:trace}, if the solution does exist  on the  interval of time $[t_i,t_f]$, we must  have 
\begin{eqnarray}\label{eq:formule}
&f^\eps_+(t,x,\xi) = {\bf 1}_{\tau_+(x,\xi)+t>0}(f^\eps_+)_{\Sigma,in}\left(\tau_+(x,\xi)+t,\Phi_+^{\tau_+(x,\xi)}(x,\xi)\right) &\\
\nonumber
&+ {\bf 1}_{\tau_+(x,\xi)+t<0}(f^\eps_+)_{\Sigma,out}\left(\tau_+(x,\xi)+t,\Phi_+^{\tau_+(x,\xi)}(x,\xi)\right).&
\end{eqnarray}
As a consequence, the solution for the plus mode will be unique (if it exists) if and only if the outgoing trace $(f^\eps_+)_{\Sigma,out}$ is uniquely determined by the entering traces $(f^\eps_+)_{\Sigma,in}$ and $(f^\eps_-)_{\Sigma,in}$.
In order to study the link between the entering and outgoing traces, we use the following lemma, the proof of which is postponed at the end of Section~\ref{sec:proof1}:
\begin{lemma}\label{lemmetechnique}
In the set of distributions, we have 
\begin{eqnarray*}
&(\partial_t+H_\pm)\left(\tau_\pm(x,\xi)+t\right)=0,\qquad (\partial_t+H_\pm)\Phi_\pm^{\tau_\pm(x,\xi)}(x,\xi)=0,\\
&H_+\left({\bf 1}_{\tau_+(x,\xi)< 0}\right)=- \Lambda_+(x,\xi) \delta_\Sigma (x,\xi),
\end{eqnarray*}
where $\Lambda_+$ is given by~(\ref{def=theta+}).
\end{lemma}
As a consequence, equation~(\ref{eq:formule}) and Lemma~\ref{lemmetechnique} yields
$$(\partial_t+H_+)f^\eps_+=\Lambda_+ \left((f^\eps_+)_{\Sigma,in}-(f^\eps_+)_{\Sigma,out}\right).$$
Besides, by~(\ref{eq:kinapp}) and the definition of $K_+$,
\begin{eqnarray*}
(\partial_t+H_+)f^\eps_+&=&K_+((f^\eps_+)_{\Sigma,in},(f^\eps_-)_{\Sigma,in}) \\
&= &  \Lambda_+\left(T_\eps (f^\eps_+)_{\Sigma,in}-T_\eps (f^\eps_-)_{\Sigma,in}\circ J_+\right).
\end{eqnarray*}
We deduce 
$$\Lambda_+ \left((f^\eps_+)_{\Sigma,in}-(f^\eps_+)_{\Sigma,out}\right)= \Lambda_+\left(T_\eps (f^\eps_+)_{\Sigma,in}-T_\eps (f^\eps_-)_{\Sigma,in}\circ J_+\right).$$
As a consequence, 
$$(f^\eps_+)_{\Sigma,out}= (1-T_\eps) (f^\eps_+)_{\Sigma,in}+T_\eps (f^\eps_-)_{\Sigma,in}\circ J_+,$$
and the outgoing and ingoing traces are linked. 
A similar argument holds for the minus mode and as a consequence, the solution of~(\ref{eq:kinapp}) is unique if it exists. \qed


\subsubsection{Existence of solutions to system~(\ref{eq:kinapp})}\label{sec:just}
We aim at proving that the semigroup~${\mathcal L}^t_{\eps}$ provides the unique solution of~(\ref{eq:kinapp}). More precisely, we want to prove
 that if  $$g^\eps(t)=( {\mathcal L}^t_{\eps}w^\eps(0)),\;\; t\in\R^+,$$
 and $g^\eps_\pm(t,x,\xi)=g^\eps(t,x,\xi,\pm 1)$, then  $(g^\eps_+,g^\eps_-)$ satisfies system~(\ref{eq:kinapp}) in ${\mathcal D}'(\R^4\setminus\{\xi=0\})$. By density of compactly supported continuous functions in $L^1(\R^4)$, it is enough to prove it for continuous initial data.
  
 \medskip
 
 By definition, the solutions of system~(\ref{eq:kinapp}) include classical transport and jumps on $\Sigma$. 
 Let us first consider  a point $(x_0,\xi_0)$ which is far from $\Sigma$, i.e. such that $\nabla V(x_0)\cdot \xi_0\not=0$. Then,  there exists a neighborhood $\Omega_2$ of $(x_0,\xi_0)$  and $\eps_2,\tau_2>0$ such that 
 for $\eps\in(0,\eps_2)$,
 $$\forall (x,\xi)\in \Omega_2,\;\; \forall t\in\R,\;\;\forall s\in]t-\tau_2,t+\tau_2[,\;\;g^\eps_\pm(t,x,\xi)=g^\eps_\pm(s,\Phi^{-t+s}_\pm(x,\xi)).$$
 As a consequence, 
 $$\partial_t g_\pm^\eps +H_\pm g^\eps_\pm=0\;\; {\rm in}\;\; \Omega_2,$$
 and more generally, 
 $\partial_t g_\pm^\eps +H_\pm g^\eps_\pm=0$ for $\nabla V(x)\cdot\xi \not=0.$
 
 \medskip
 
 Consider now a point $(x_0,\xi_0)$ such that $\xi_0\cdot\nabla V(x_0)=0$.  Then, there exists a neighborhood $\Omega_2$ of $(x_0,\xi_0)$ and $\eps_2,\tau_2>0$ such that for $\eps\in(0,\eps_2)$, $(x,\xi)\in\Omega_2$, $t\in\R$ and $s\in(t-\tau_2,t+\tau_2)$ (the time $\tau_2$ corresponds to the length of an interval of time during which trajectories issued from points of $\Omega_2$ have at most one jump), 
\begin{eqnarray*}
g^\eps_{+}(t,x,\xi) & = & 
g_+^\eps \left(s,\Phi^{-t+s}_+(x,\xi)\right)\\
\nonumber &- & 
{\bf 1}_{-t+s\leq \tau_+(x,\xi)<0} 
T_\eps(\Phi_+^{\tau_+(x,\xi)}(x,\xi))
 g^\eps_{+}\left(\tau_+(x,\xi)+t,\Phi^{\tau_+(x,\xi)}_+(x,\xi)\right)\\
 \nonumber
 & +& 
{\bf 1}_{-t+s\leq \tau_+(x,\xi) < 0} 
\; T_\eps\circ J_+(\Phi_+^{\tau_+(x,\xi)}(x,\xi))\\
  & &
  \nonumber
   \qquad \qquad \times\,
   g^\eps_{-}\left(\tau_+(x,\xi)+t,J_+\left(\Phi^{\tau_+(x,\xi)}_+(x,\xi)\right)\right).
 \end{eqnarray*}
 We obtain 
   \begin{eqnarray}\label{eq:g+}
   \quad\quad
    g^\eps_{+}(t,x,\xi)
   & = & g_+^\eps \left(s,\Phi^{-t+s}_+(x,\xi)\right)\\
\nonumber
   & & +{\bf 1}_{-t+s\leq \tau_+(x,\xi) <0} \; \Bigl[
  (T_\eps g^\eps_{-})\left(\tau_+(x,\xi)+t,J_+(\Phi_+^{\tau_+(x,\xi)}(x,\xi))\right)
  \\
  \nonumber
 & & \qquad\qquad \,
  -(T_\eps g^\eps_{+})\left(\tau_+(x,\xi)+t,\Phi^{\tau_+(x,\xi)}_+(x,\xi)\right)
  \Bigr].   \end{eqnarray}
  Similarly, we have
\begin{eqnarray}\label{eq:g-}
  \quad\quad  g^\eps_{-}(t,x,\xi)
   & = & g_-^\eps \left(s,\Phi^{-t+s}_-(x,\xi)\right)\\
   \nonumber
   & & +{\bf 1}_{-t+s\leq \tau_-(x,\xi)< 0} \;
  \Bigl[
  (T_\eps g^\eps_{+})\left(\tau_-(x,\xi)+t,J_-(\Phi_-^{\tau_-(x,\xi)}(x,\xi))\right)
  \\
  \nonumber
 & & \qquad\qquad  \,
  -(T_\eps g^\eps_{-})\left(\tau_-(x,\xi)+t,\Phi^{\tau_-(x,\xi)}_-(x,\xi)\right)
  \Bigr].
\end{eqnarray}
Note that equations~(\ref{eq:g+}) and~(\ref{eq:g-}) are equivalent to the relation $g^\eps(t)={\mathcal L}^{t-s}_\eps(g^\eps(s))$ and give an explicit expression for the semigroup on a small time $\tau:=t-s$ during which the trajectories jump at most once. 

\medskip

The result is then straightforward by Lemma~\ref{lemmetechnique}. For the $+$~mode (the proof for the $-$~mode is similar),
we have
$$\displaylines{
(\partial_t+ H_+) g^\eps_+  =  {\bf 1}_{\tau_+(x,\xi)+t\geq 0} H_+\left({\bf 1}_{\tau_+(x,\xi)\leq 0}\right)
\Bigl[
(T_\eps g^\eps_-)\left(\tau_+(x,\xi)+t,J_+( \Phi_+^{\tau_+(x,\xi)}(x,\xi))\right)
\hfill\cr\hfill 
-(T_\eps g^\eps_+)\left(\tau_+(x,\xi)+t,\Phi^{\tau_+(x,\xi)}(x,\xi)\right)\Bigr]. \cr}$$
Observing that 
$\Phi_+^{\tau_+(x,\xi)}(x,\xi)=(x,\xi)$ for~$(x,\xi)\in\Sigma,$
we obtain for $t\geq 0$, 
$$(\partial_t+ H_+) g^\eps_+  = \Lambda_+(x,\xi)\delta_\Sigma(x,\xi)\left[ T_\eps (g^\eps_+)_{\Sigma,in}-
(T_\eps (g^\eps_-)_{\Sigma,in})\circ J_+\right] $$
where we have used that for $\tau_+(x,\xi)=0$ and $t\geq 0$, we have ${\bf 1}_{\tau_+(x,\xi)+t\geq 0}=1$.
This implies that $g^\eps_+$ satisfies~(\ref{eq:kinapp}).

\qed

\subsubsection{Proof of Lemma~\ref{lemmetechnique}} Let us begin with the first line. 
Recall that 
we have
$$
(\partial_t+H_\pm)\Phi^{-t}_\pm=0.
$$
Similarly, in view of 
$\displaystyle{\tau_\pm\left(\Phi_\pm^{-t}(x,\xi)\right)=\tau_\pm(x,\xi)+t}$
and writing 
$$\Phi_\pm^{\tau_\pm}(x,\xi)=\Phi_\pm^{\tau_\pm(x,\xi)+t} \left(\Phi_\pm^{-t}(x,\xi)\right),$$
we get
$\displaystyle{(\partial_t +H_\pm)\left(\Phi_\pm^{\tau_\pm}(x,\xi)\right)=0}$.\\
As a consequence, we have  in ${\mathcal D}'(\Omega_{out})$, 
$H_+\tau_+=-1$,
which implies 
$$H_+\left(\Phi^{\tau_+(x,\xi)}(x,\xi)\right)= (H_+\tau_+)\dot \Phi^{\tau_+(x,\xi)}(x,\xi)+(H_+\Phi)^{\tau_+(x,\xi)}(x,\xi)=0.$$
%
Let us now prove the second line.
Note that $\tau_+(x,\xi)=0$ is an equation of the hypersurface~$\Sigma$. For calculating $H_+\left({\bf 1}_{\tau_+(x,\xi)<0}\right) $, 
we write $$z=(x,\xi),\;\;H_+a=F(z)\cdot \nabla_z a,$$ where $a$ is a test-function, and we observe that $\nabla_z\cdot F(z)=0$. Therefore, using that $a$ is compactly supported, Green's formula reads 
\begin{eqnarray*}
\int _{\{\tau_+(z)<0\}} F(z)\cdot \nabla_z a(z) dz& =& \int_\Sigma F(z)\cdot n(z) \,a(z)d\sigma -\int_{\{\tau_+(z)<0\}} \nabla_z \cdot F(z)\,  a(z) dz\\
&= & \int_{\Sigma} F(z)\cdot n(z) \,a(z) d\sigma
\end{eqnarray*}
where $n(z)$ is the unitary exterior normal vector to $\Sigma$: 
 $$n(x,\xi)= \left(|d^2V(x)\xi|^2+| \nabla V(x)|^2\right)^{-1/2} \left( \begin{array} {l} d^2V(x)\xi \\ \nabla V(x)\end{array}\right).$$
 Note that $H_+$ is transverse to $\Sigma$ and points towards the region $\tau_+<0$. Besides, one can check that 
$
 H_+(x,\xi)\cdot n(x,\xi)=\Lambda_+(x,\xi),$ where $\Lambda_+$ is defined by~(\ref{def=theta+}).
At this stage of the proof, we have obtained for any smooth compactly supported function $a$,
\begin{eqnarray*}
\int \left(H_+ {\bf 1}_{\tau_+(x,\xi)<0}\right)\, a dxd\xi  & = &  -  \int {\bf 1}_{\tau_+(x,\xi)<0} (H_+a )\, dxd\xi\\
& = & -\int_\Sigma a \, (H_+\cdot n) d\sigma\\
& = &  - \int_\Sigma \Lambda_+ a\;d\sigma
,
\end{eqnarray*}
which gives the result.
\qed


\subsection{Proof of Proposition~\ref{prop:FL}}\label{sec:proof}

Recall that the Proposition~\ref{prop:FL} implies Theorem~\ref{theo:main}, which give the mathematical justification of the algorithm that we propose therein.

\subsubsection{Strategy}\label{subset:strategy} Let us first describe the strategy of the proof of 
Proposition~\ref{prop:FL}. 
The proof relies on a characterization of $w^\eps(t)$ via pseudodifferential operators.
Recall that if  $a\in{\mathcal C}_0^\infty(\R^2)$, the semiclassical pseudodifferential operator of symbol $a$ is defined by the Weyl quantization rule
$$
\op_\eps(a) f(x) = (2\pi)^{-2} \int {\rm e}^{i\xi\cdot (x-y)} a\left({x+y\over 2},\eps\xi\right) f(y) dy d\xi, \quad f\in {\mathcal S}(\R^2).$$
This operator extends to functions $f\in L^2(\R^2)$ and one can prove that $\left(\op_\eps(a)\right)$ is a uniformly bounded family of operators of ${\mathcal L}(L^2(\R^2))$ since there exists a constant $C>0$ such that 
\begin{equation}\label{eq:CV}
\forall a\in{\mathcal C}_0^\infty(\R^4),\;\;\| \op_\eps(a)\|_{{\mathcal L}(L^2(\R^2))}\leq C \, \sup _{|\beta|\leq 3} \sup_{\xi \in\R^2}\int _{\R^2}| \partial_x^\beta  a(x,\xi)|dx .
\end{equation}
We refer to the books~\cite{AG,DS,Z} for a complete study of pseudodifferential operators. The estimate~(\ref{eq:CV}) is  not the standard Calderon-Vaillancourt estimate (see~\cite{CV}) that is usually used. It has the advantage not to differentiate in the variable $\xi$ and is inspired from~\cite{GL} (see also the survey~\cite{AFF}). A short proof is given in the Appendix for the convenience of the reader, we also recall the single symbolic calculus result that we shall use. 

\medskip

Denote by $\C^{2,2}$ the set of $2\times 2$ complex matrices and consider 
 symbols $a$ that are  matrix-valued: $a\in{\mathcal C}_0^\infty(\R^2,\C^{2,2})$. Then, the operator $\op_\eps(a)$ is a matrix-valued operator acting on functions of $L^2(\R^2,\C^2)$. If $w^\eps$ is the Wigner transform of the matrix density $\varrho^\eps$, we have the relation 
\begin{equation}\label{eq:link}
\langle a, w^\eps\rangle = \,{\rm tr} \left(\op_\eps(a) \varrho^\eps\right),
\end{equation}
where the bracket between the two matrices $a=(a_{i,j})$ and $w^\eps=(w^\eps_{i,j})$ is defined by 
\begin{equation}\label{def:bracket}
\langle   a,w^\eps\rangle = \sum_{i,j} \int_{\R^4} a_{
i,j}(x,\xi)w^\eps_{j,i}(x,\xi) dxd\xi.
\end{equation}
 We shall use this description of $\langle a,w^\eps\rangle$ in order to prove Proposition~\ref{prop:FL}. 

\begin{remark}\label{rem:feps}
The relations~(\ref{eq:link}) and~(\ref{eq:CV}) imply that, under Assumption \ref{ass:rho0}, the family $(w^\eps(t))_{\eps>0}$ is a bounded family in the set of distributions.
\end{remark}

\medskip 

Since the initial density matrix $\varrho^\eps(0)$ is supposed to be a Hilbert-Schmidt operator, there exists a sequence $(\lambda_j)_{j\in\N}$ of $\ell^2(\N)$ and a sequence of bounded normalized  families $(\psi^\eps_{0,j})_{j\in\N}$ of $L^2(\R^2)$ such that 
$$\varrho^\eps(0)=\sum_{j\in\N} \lambda_j |\psi^\eps_{0,j}\rangle \langle\psi^\eps_{0,j}|.$$
As a consequence, for $t\in\R$,
$$\varrho^\eps(t)=\sum_{j\in\N} \lambda_j |\psi^\eps_{j}(t)\rangle \langle\psi^\eps_{j}(t)|,$$
where for any $j\in\N$, the family $(\psi^\eps_j(t))_{\eps>0}$ is a family of solutions to the Dirac equation 
\begin{equation}\label{eq:function}
i\eps\partial_t\psi^\eps_j=\left(A(\eps D)+V(x) \right)\psi^\eps_j,
\end{equation}
with initial data $\psi^\eps_j(t)=\psi^\eps_{0,j}$. 
Besides, 
the relation~(\ref{eq:link}) yields 
$$\langle a, w^\eps\rangle = \,{\rm tr} \left(\op_\eps(a) \varrho^\eps\right) =\sum_{j\in\N} \lambda_j\left(\op_\eps(a)\psi^\eps_j(t)\;,\;\psi^\eps_j(t)\right)_{L^2(\R^2_x)}.$$
We denote by $w^\eps_j(t)$ the Wigner transform of the family $(\psi^\eps_j(t))_{\eps>0}$ which is defined by the relation 
$$\forall a\in{\mathcal C}_0^\infty(\R^4),\;\;  \langle a , w^\eps_j(t) \rangle =  \left(\op_\eps(a)\psi^\eps_j(t)\;,\;\psi^\eps_j(t)\right)_{L^2(\R^2_x)}.$$
The Wigner function $w^\eps_j(t)$ is a $2$ by $2$ matrix and the bracket involved in the preceding relation is also the one defined in~(\ref{def:bracket}). 
In the following, we will  characterize $w^\eps_j(t)$ in terms of the flow ${\mathcal L}^t_{\eps}$. More precisely, we are going to prove that for any $j\in\N$,  $ w^\eps_j(t)$ satisfies Proposition~\ref{prop:FL}, which gives the result for $w^\eps(t)$. 

\medskip

For this purpose, 
we use the results of~\cite{LT,FL1,FL2,FL3} which are stated  for a Schr\"odinger equation with matrix-valued potential.  This comes from the following observation: whenever $V(x)=|x|^2$,  the operator $A(\eps D) + V(x)$ becomes a Schr\"odinger operator with a matrix-valued potential by taking the Fourier transform. As a consequence, the methods developed in~\cite{LT,FL1} for Schr\"odinger equation with matrix-valued potential can be adapted to our setting. Furthermore, conical intersections have been classified in~\cite{CdV1} and~\cite{FG2} and the Dirac-type equation~(\ref{eq:system}), like the Schr\"odinger equations of~\cite{FL1,FL2,FL3}, enters in the same class of crossings. Thus, it is not surprising that similar methods do apply. Note however  that the jumps were omitted in~\cite{LT} and~\cite{FL1}; as mentioned in~\cite{FL3}, these jumps are required for the correctness of the proof of~\cite{FL1}. 

\medskip

Then, the  main steps of the proof will consist in:
\begin{enumerate}
\item  The transport outside $\Sigma$.
\item Localization in energy and use of space-time variables.
\item A normal form which reduces to a simple model called the Landau-Zener system.
\item The computation of the  transitions on $\Sigma$ that is performed via the normal form and the Landau-Zener system.
\end{enumerate}

In the following, we use Remark~\ref{UepsR} for taking into account only the jumps which occur inside the set 
$${\mathcal U}_{\eps,R}=\{ (x,\xi)\in \R^4,\;\;|\xi|\leq R\sqrt\eps\}.$$
Let us introduce the semi-group ${\mathcal L}_{\eps,R}$ which restricts the jumps to those occurring inside ${\mathcal U}_{\eps,R}$, the semi-group~${\mathcal L}_{\eps,R}$ differs from ${\mathcal L}_{\eps}$ by exponentially small terms. It is this semi-group that we shall consider now.  The real number $R$ will be chosen as $R=\eps^{-{1/ 8}}$ (see Notation~\ref{notation} below). 

\medskip 

Let us now detail these steps. For simplicity, we omit the index ``$j$'' and simply consider a family $(\psi^\eps(t))_{\eps>0}$, uniformly bounded in $L^2(\R^2)$, of solutions to the Dirac equation~(\ref{eq:function}) with initial data $(\psi^\eps_0)_{\eps>0}$ and we denote by $w^\eps(t)$ its Wigner transform at time $t$. We also denote by $w^\eps_\pm(t)$ the scalar quantities $w^\eps_{\pm}(t)={\rm tr}\left(\Pi^\pm w^\eps(t)\right)$.

\subsubsection{The transport outside the transition region} 
The analogue of Proposition~2.3 in~\cite{FL3} is the following

\begin{proposition}
\label{prop:propagation1} Let $c\in{\mathcal C}_c^\infty(\R^{4},\C)$,
and let $b\in{\mathcal C}^\infty(\R^2,\C)$ with $\nabla b$ compactly supported.
If there exist $C>0$ and $s_0>0$ such that
$$
\forall r\in[-s_0,s_0]:\quad\Phi_\pm^{r}({\rm supp}(c))\cap {\mathcal U}_{\eps,R}=\emptyset
$$
then for all $\chi\in{\mathcal C}_c^\infty(\R,\R)$ and for all $s\in[t-s_0,t+s_0]$
\begin{eqnarray*}
\lefteqn{
{\rm tr} \int_{\R^{2d+1}}\chi(t)\,c(x,\xi)\,
b\!\left({\textstyle{\xi\over R\sqrt\eps}}\right) w^\eps_\pm(t,x,\xi)\,d x\,d \xi\,d t
=}\\
&&
{\rm tr}\int_{\R^{2d+1}}\chi(t)\,c(x,\xi)\,
b\!\left({\textstyle{\xi\over
R\sqrt\eps}}\right)
\left(
w^\eps_\pm\left(s,\Phi^{-t+s}_\pm\right)\!(x,\xi)\right)\,d
 x\,d \xi\,d t\\
&&\hspace*{12em}+
\mathcal O(1/(\sqrt\eps R^{5}))+\mathcal O(1/R^2)+\mathcal O(\eps).
\end{eqnarray*}
\end{proposition}

This proposition is a refined version of the resolution of the kinetic system~(\ref{eq:kin1}). Indeed, 
let $\Omega\subset \R^4$ be an open subset of  $\{|\xi|>\delta_0\}$  for some $\delta_0>0$, and $s_0$ such that for $(x,\xi)\in\Omega$, 
the trajectories $(\Phi^s_\pm(x,\xi))_{s\in[-s_0,s_0]}$ remain in  $\{|\xi|>\delta_0/2\}$, then~(\ref{eq:kin1})   gives for all $s\in[t-s_0,t+s_0]$, 
$$w^\eps_\pm(t,x,\xi)=w^\eps_\pm(s,\Phi^{-t+s}_\pm(x,\xi))+\mathcal O(\eps)\;\;{\rm in}\;\;{\mathcal D}'(\Omega).$$
Proposition~\ref{prop:propagation1} authorizes to be at a distance of order $\mathcal O(R\sqrt\eps)$ of $\{\xi=0\}$ 

\medskip

Since for scalar symbols $a$, we have
\begin{eqnarray*}
{d\over dt} \langle a\Pi^\pm\;,\; w^\eps(t)\rangle 
& = &  \, \left( {1\over i\eps} \left[\op_\eps(a\Pi^\pm)\;,\;A(\eps D) +V\right]\psi^\eps(t)\;,\;\psi^\eps(t)\right)_{L^2(\R^2_x)},
\end{eqnarray*}
the proof of this proposition relies on  a good understanding of the operator  
$$L_\eps= {1\over i\eps} \left[\op_\eps(a\Pi^\pm)\;,\;A(\eps D) +V\right].$$ 
The main ingredients are the two following observations:
\begin{itemize}
\item For 
$\displaystyle{a(x,\xi)= c(x,\xi)\,
b\!\left({\textstyle{\xi\over R\sqrt\eps}}\right),}$
the symbol $a(x,\xi)\Pi^+(\xi)$ is smooth and 
we have 
\begin{equation}\label{est:loinde0}
\forall \alpha,\beta\in\N^2,\;\;\exists C_{\alpha,\beta}>0,\;\;\left| \partial_x^\alpha\partial_\xi^\beta \left(a(x,\xi)\Pi^+(\xi)\right)\right|\leq C\,(R\sqrt \eps)^{-|\beta|},
\end{equation}
so that we can use the symbolic calculus theorems of the Appendix, paying attention to the rest terms. 
\item If $B$ is an off-diagonal symbol, that is a symbol which satisfies 
$$B=\Pi^+B\Pi^-+\Pi^-B\Pi^+,$$
the quantities
$$\int \chi(t)\left(\op_\eps(B)\psi^\eps(t)\;,\;\psi^\eps(t) \right)dt,$$
which seems to be of order~$\mathcal O(1)$, can be proved to be of smaller order than expected by re-using the equation satisfied by $\psi^\eps(t)$. 
\end{itemize}

\begin{proof}[Proof of Proposition \ref{prop:propagation1}]
Let us now focus on the proof itself. 
Using Proposition~\ref{prop:symbol} and observing that $A(\xi)=|\xi|(\Pi^+-\Pi^-)$ with $\Pi^++\Pi^-=1$,  we obtain
$$L_\eps = -\op_\eps\left((\nabla V(x)\cdot \nabla_\xi a -{\xi\over|\xi|}\cdot \nabla_x a )\Pi^+\right) +\op_\eps (B)+\mathcal O(R^{-2})+\mathcal O(\eps),$$
with 
\begin{eqnarray*}
B &= & -a\,\nabla V\cdot \nabla \Pi^+ -{1\over 2} \left(|\xi| (\{a\Pi^+,\Pi^+\}-\{\Pi^+,a\Pi^+\})\right)\\
& & \qquad -{1\over 2}\left(|\xi|(\{a \Pi^+,\Pi^-\}-\{\Pi^-,a\Pi^+\})\right)\\
& = & -a \,\nabla V\cdot \nabla \Pi^+-|\xi| \left(\{a\Pi^+,\Pi^+\}-\{\Pi^+,a\Pi^+\}\right)\\
& = &-a\,\nabla V\cdot \nabla \Pi^++|\xi|\left(\Pi^+ \nabla_x a \cdot\nabla \Pi^+
+\nabla_x a\cdot \nabla \Pi^+ \Pi^+\right)\\
& = & -a\, \nabla V\cdot \nabla \Pi^++|\xi|\nabla_x a \cdot\nabla \Pi^+.
\end{eqnarray*}
Here we have used $\nabla\Pi^+ =\Pi^+\nabla \Pi^++\nabla\Pi^+\Pi^+.$
As a consequence, $B$ is an off-diagonal symbol. We write $B=B_0+B_1$ with $B_0= -a \nabla V \cdot \nabla \Pi^+$, we have
\begin{equation}\label{est:Bj}
\forall \alpha,\beta\in\N^d,\;\;\exists C_{\alpha,\beta}>0,\;\;\forall x,\xi\in\R^{4},\;\; \left|\partial^\beta_\xi\partial_x^\alpha B_j\right|\leq C_{\alpha,\beta} (R\sqrt\eps)^{-|\beta|-1+j}.
\end{equation}
The result comes from the next lemma which concludes the proof. \end{proof}
\begin{lemma}\label{lem:B1B0}
For any $\chi\in{\mathcal C}_0^\infty(\R)$, we have 
$$\displaylines{
\int \chi(t) \left(\op_\eps(B_1)\psi^\eps(t),\psi^\eps(t)\right) =\mathcal O(R^{-2})+\mathcal O(\eps)
,\cr
\int \chi(t) \left(\op_\eps(B_0)\psi^\eps(t),\psi^\eps(t)\right) =\mathcal O(R^{-5}\eps^{-1/2})+\mathcal O(\sqrt\eps)
.\cr}$$
\end{lemma}
\begin{proof}
We begin with $B_1$. Since $B_1$ is off-diagonal, we can write 
\begin{eqnarray*}
B_1 & = & [(\Pi^-B_1\Pi^+-\Pi^+B_1\Pi^-) (2|\xi|)^{-1} , A(\xi)]\\
&=&[(\Pi^-B_1\Pi^+-\Pi^+B_1\Pi^-)(2 |\xi|)^{-1}, \tau +V(x)+A(\xi)].
\end{eqnarray*}
After quantization, we get 
$$\displaylines{\qquad \op_\eps(B_1) =\left[ \op_\eps((\Pi^-B_1\Pi^+-\Pi^+B_1\Pi^-)  (2|\xi|)^{-1}) , {\eps\over i}\partial_t+V(x)+A(\eps D)\right]\hfill\cr\hfill 
+\mathcal O(R^{-2}) +\mathcal O(\sqrt\eps/R).\cr}$$
Once applied to $\psi^\eps$ which satisfies the Dirac equation~(\ref{eq:function}), we obtain the announced relation. 

\medskip

Note that we have obtained more generally that if $B_j$ is off-diagonal and satisfies the relation~(\ref{est:Bj}), then 
\begin{eqnarray}\label{eq:Bjbis}
\int \chi(t) \left(\op_\eps(B_j)\psi^\eps(t),\psi^\eps(t)\right) &=&\mathcal O\left((R\sqrt\eps)^{-1+j}) \left(R^{-2}+\sqrt\eps/R\right)\right)\\
\nonumber &=&\mathcal O\left((R\sqrt\eps)^{-1+j}) \left(R^{-2}+\eps\right)\right)
\end{eqnarray}
In particular, for  $B_0$, we obtain 
$$
\int \chi(t) \left(\op_\eps(B_0)\psi^\eps(t),\psi^\eps(t)\right) =\mathcal O(R^{-3}\eps^{-1/2})+\mathcal O(\sqrt\eps/R)
$$
that we want to improve. 
Therefore, we go one step further in the symbolic calculus and we write
\begin{eqnarray*}
\op_\eps(B_0) & = &  
 \left[ \op_\eps((\Pi^-B_0\Pi^+-\Pi^+B_0\Pi^-)  (2|\xi|)^{-1}) , {\eps\over i}\partial_t+V(x)+A(\eps D)\right]\\
 & &  +{\eps\over i}\,\op_\eps\left(\left\{ (\Pi^-B_0\Pi^+-\Pi^+B_0\Pi^-) (2|\xi|)^{-1} ,\tau +V(x) \right)\right\}\\
& & +{\eps\over 2i} \op_\eps\Bigl((\left\{(\Pi^-B_0\Pi^+-\Pi^+B_0\Pi^-) (2|\xi|)^{-1} , A(\xi)\right\}\\
& & 
-\left\{A(\xi),(\Pi^-B_0\Pi^+-\Pi^+B_0\Pi^-)  (2|\xi|)^{-1}\right\}\Bigr)\\
& & +\mathcal O(R^{-4})+\mathcal O(\sqrt\eps/R^3).
\end{eqnarray*}
Paying attention to all these terms, we observe that 
$$\displaylines{\op_\eps(B_0)=
 \left[ \op_\eps((\Pi^-B_0\Pi^+-\Pi^+B_0\Pi^-)  (2|\xi|)^{-1}) , {\eps\over i}\partial_t+V(x)+A(\eps D)\right]\hfill\cr\hfill
-i {\eps}\, \op_\eps\left(\nabla V\cdot \nabla_\xi \left((\Pi^-B_0\Pi^+-\Pi^+B_0\Pi^-)  (2|\xi|)^{-1}\right)\right) +\mathcal O(R^{-2})+\mathcal O(\sqrt\eps).\cr}$$
The matrix 
$$B_{-2}= \nabla V\cdot \nabla_\xi\left((\Pi^-B_0\Pi^+-\Pi^+B_0\Pi^-)  (2|\xi|)^{-1}\right)$$
satisfies~(\ref{est:Bj}) with $j=-2$ and we claim that $B_{-2}$
is also off-diagonal.
As a consequence,  equation~(\ref{eq:Bjbis}) gives
$$\int \chi(t) \left(\op_\eps(B_{-2})\psi^\eps(t),\psi^\eps(t)\right) dt =\mathcal O(R^{-5}\eps^{-3/2})+\mathcal O(R^{-3}\eps^{-1/2}),$$
which concludes the proof of Lemma~\ref{lem:B1B0}.

\medskip 

It remains to prove the claim, a simple calculus shows that
$$
B_{-2}  =  - {1\over 2} |\xi|^{-3}(\nabla V\cdot\xi ) (\Pi^-B_0\Pi^+-\Pi^+B_0\Pi^-) 
 + (2|\xi|)^{-1} \nabla V\cdot\nabla_\xi (\Pi^-B_0\Pi^+-\Pi^+B_0\Pi^-).
$$
Therefore, 
\begin{eqnarray*}
\Pi^\pm B_{-2}\Pi^\pm& =&
(2|\xi|)^{-1} \Pi^\pm\left( \nabla V\cdot \nabla\Pi^-B_0\Pi^++\Pi^-B_0 \,\nabla V\cdot \nabla\Pi^+\right)\Pi^\pm\\
& &- (2|\xi|)^{-1} \Pi^\pm\left( \nabla V\cdot \nabla\Pi^+B_0\Pi^-+\Pi^+B_0 \,\nabla V\cdot \nabla\Pi^-\right)\Pi^\pm\\
&=& |\xi|^{-1} \Pi^\pm\left(- \nabla V\cdot \nabla\Pi^+B_0\Pi^++\Pi^+B_0 \,\nabla V\cdot \nabla\Pi^+\right)\Pi^\pm\\
&=& |\xi|^{-1} \Pi^\pm\left[B_0\, ,\,\nabla V\cdot \nabla \Pi^+\right]\Pi^\pm\\
&=&0
\end{eqnarray*}
since $B_0=a\, \nabla V\cdot \nabla \Pi^+$, which proves that $B_{-2}$ is off-diagonal.

\end{proof}

\begin{notation}\label{notation}
In the following, it will be convenient to denote by $\eta_\eps$ any rest term smaller than $\mathcal O(1/(R^{5}\sqrt\eps))+\mathcal O(1/R^2)+\mathcal O(\sqrt\eps)+\mathcal O(R^3\sqrt\eps)$. The term in $R^3\sqrt\eps$ will be useful in the following. Note that when $R=\eps^{-1/8}$, we have $\eta_\eps=\mathcal O(\eps^{1/8})$. 
\end{notation}

\subsubsection{Localization in energy}

The memory of the mode  by use of  a matrix-valued symbol of the form $a\Pi^+$ or $a\Pi^-$, with $a$ scalar, can be replaced  by  a localization in energy. This requires to work in space time variables and has the advantage that we are reduced to use scalar symbols. Using scalar symbols will be convenient in the next section when we will perform a normal form and use a Fourier Integral Operator. The   energy surfaces of the space-time phase space $\R^{6}_{t,x,\tau,\xi}$ are the sets
\begin{equation}
\label{numero}
E^\pm=\{(t,x,\tau,\xi)\in\R^{6},\;\;\tau= \mp|\xi| -V(x)\}.
\end{equation}
Recall that the dual variable of the time~$t$ is interpreted as an energy~$\tau$. 

\medskip 

In the following, we shall use semi-classical pseudo differential operators with symbols depending on the variable~$(t,x,\tau,\xi)\in\R^6$ with the choice of the Weyl quantization in the time variables, as it was already the case for the space variables. 

\medskip

The localization in energy is done by use of a cut-off function $\theta\in{\mathcal C}_0^\infty(\R)$ such that $0\leq \theta\leq 1$, $\theta(\tau)=0$ for $|\tau|>1$ and $\theta(\tau)=1$ for $|\tau|<1/2$. This function~$\theta$ is fixed from now on. 

\begin{lemma}\label{lem:loc}
Let $a\in{\mathcal C}^\infty_0(\R^{4})$ and set
$c_{\eps,R}(x,\xi )=a(x,\xi) (1-\theta)(2\xi/(R\sqrt\eps))$, then for all $\chi\in{\mathcal C}_0^\infty(\R)$,
$$\displaylines{
 \int_\R \chi(t)  \left( \op_\eps(c_{\eps,R}\Pi^\pm ) \psi^\eps(t)\;,\;\psi^\eps(t)\right) _{L^2(\R^2)}dt =
\mathcal O(1/R^2)+\mathcal O(\sqrt\eps)\hfill\cr\hfill
+ \, \left(\op_\eps\! \left(\chi(t) c_{\eps,R}(x,\xi)
\theta\!\left({\tau\pm|\xi|+V(x)\over
R\sqrt\eps}\right)\right) \psi^\eps\;,\;\psi^\eps
\right)_{L^2(\R^3_{t,x})}.\cr}$$
\end{lemma}

\begin{remark}\label{rem:energy}
It is because the localization in energy is made in balls of size $\sqrt\eps$ that we need to perform the jumps: they guarantee that the energy of the created trajectory do not differ at order $\mathcal O(\sqrt\eps)$ but at least at order $\mathcal O(\eps)$. 
\end{remark}

\begin{remark} \label{rem:loc}
Note that the presence of the eigenprojector in the symbol induces restriction on both components of the function $\psi^\eps(t)$. Indeed,  by the symbolic calculus of the Appendix and by equation~(\ref{est:loinde0}) we have 
$$\displaylines{ \int_\R \chi(t)  \left( \op_\eps(c_{\eps,R}\Pi^\pm ) \psi^\eps(t)\;,\;\psi^\eps(t)\right) _{L^2(\R^2)}dt \hfill\cr\hfill=
 \int_\R \chi(t)  \left( \op_\eps(c_{\eps,R}\Pi^\pm ) \psi^\eps(t)\;,\;\Pi^\pm\psi^\eps(t)\right) _{L^2(\R^2)}dt+\mathcal O(R\sqrt\eps).\cr}$$
\end{remark}

\begin{proof}[Proof of Lemma \ref{lem:loc}] We set 
$$\theta^\pm_{\eps,R} (x,\tau,\xi)= \theta\!\left({\tau \pm |\xi|+V(x)\over
R\sqrt\eps}\right).$$
Following the lines of the proof of Lemma~5.1 in \cite{FL2}, we observe that since $1-\theta$ vanishes identically close to $0$, one can write 
$$1-\theta^+_{\eps,R}(x,\tau,\xi)= {1\over R\sqrt\eps}( \tau + |\xi|+V(x)) G\left({\tau + |\xi|+V(x)\over
R\sqrt\eps}\right)$$
for some smooth function $G$, with 
$$( \tau + |\xi|+V(x)) \Pi^+ (\xi)=\Pi^+(\xi) (\tau+A(\xi)+V(x)).$$
Therefore, we can use the equation satisfied by $\psi^\eps(t)$, symbolic calculus and the estimate~(\ref{est:loinde0}) to obtain
$$
\left(\op_\eps\left(\chi c_{\eps,R}\Pi^+\right)\psi^\eps\;,\;\psi^\eps\right)
=
\left(\op_\eps\left(\chi c_{\eps,R}\theta^+_{\eps,R}\Pi^+\right)\psi^\eps\;,\;\psi^\eps\right) + \mathcal O\left(R^{-2}\right)+\mathcal O(\eps).$$

\medskip

It remains to get rid of the matrix $\Pi^+(\xi)$. 
In view of 
 \begin{equation*}\label{eq2}
 \chi c_{\eps,R} \theta^+_{\eps,R} = \chi c_{\eps,R} \theta^+_{\eps,R} \Pi^+ + \chi c_{\eps,R} \theta^+_{\eps,R} \Pi^-,
 \end{equation*} 
we only need to prove that $\left(\op_\eps( \chi c_{\eps,R} \theta^+_{\eps,R} \Pi^-)\psi^\eps,\psi^\eps\right) =\mathcal O(\eta_\eps).$
 We observe that
 $
\theta^+_{\eps,R}c_{\eps,R}=\theta^+_{\eps,R}(1-\theta^-_{\eps,R})c_{\eps,R},$
and
\begin{eqnarray*}
(1-\theta^-_{\eps,R})\Pi^-&=&\frac{1}{R\sqrt\eps} G\left(\frac{\tau-|\xi|+V(x)}{R\sqrt\eps}\right) (\tau-|\xi|+V(x))\Pi^-\\
& = & \frac{1}{R\sqrt\eps} G\left(\frac{\tau-|\xi|+V(x)}{R\sqrt\eps}\right)\Pi^- (\tau+A(\xi)+V(x)).
\end{eqnarray*}
By using again the equation, symbolic calculus and estimate~(\ref{est:loinde0}), we can write 
\begin{eqnarray*}
\lefteqn{
\left(\op_\eps\left(\chi a^+ \theta^+_{\eps,R} \Pi^-\right)\psi^\eps\;,\;\psi^\eps\right)}\\
&=& 
\left(\op_\eps\left(\chi a^+\theta^+_{\eps,R}(1-\theta^-_{\eps,R})\Pi^-\right)\psi^\eps\;,\;\psi^\eps\right) + \mathcal O\left(R^{-2}\right)+\mathcal O(\eps)
= \mathcal O(\eta_\eps).
\end{eqnarray*}
The proof for the minus mode is similar.
\end{proof}

\subsubsection{The normal form} For computing the transitions, we use a normal form result. For this, we need to work microlocally in space-time phase space variables. Following~\cite{CdV1,FG1,F}, close to a point $(t_0,x_0,\xi_0=0,\tau_0=-V(x_0))$, there exist a change of coordinates
$$\kappa:\;(s,z,\sigma,\zeta)\mapsto (t,x,\tau,\xi)$$
with $s,\sigma\in\R$ and $z=(z_1,z_2),\;\zeta=(\zeta_1,\zeta_2)\in\R^2$, 
and a matrix $B$ such that 
\begin{equation}\label{def:B}
\left(\tau+V+A(\xi)\right) \circ \kappa = \,^tB\left(-\sigma +\widetilde A(s,z_1)\right)B,
\end{equation}
with
$$\widetilde A(s,z_1)=\displaystyle{ \begin{pmatrix}s & z_1 \\ z_1 & -s \end{pmatrix}}.$$
Moreover, this change of coordinates preserves the symplectic structure of the phase space $\R^3_{t,x}\times\R^3_{\tau,\xi}$: the variables $\sigma$ and $\zeta$ are respectively the dual variables of $s$ and~$z$. Besides, in view of Section 6.2 of~\cite{FG1},
 there exists a function $\gamma>0$  such that $\,^tBB= \gamma {\rm Id}$.
 
 \medskip 
 
The construction of the canonical transform is based on the vectors $H$ and $H'$ defined in~(\ref{def:H}) and~(\ref{def:H'}) (see~\cite{CdV1} and the analysis performed in \cite{FL1,F}).  The variable~$s$ is chosen such that the trajectories which reach $\{\xi=0\}$ are included in $\{s<0\}$ and those which leave $\(\xi=0\)$ are included in $\{s>0\}$. Besides, one 
extends the vectors $H$ and $H'$ as vectors of $T^*\R^3_{t,x}$  by adding the coordinate $1$  along $\partial_t$ and the coordinate $0$ along $\partial_\tau$ and we keep calling them $H$ and $H'$. The resulting vectors are the limit  on $\{\xi=0\}$ and along the flows of the Hamiltonian vector fields associated with the functions $\tau + V(x)\pm|\xi|$. They are sent by $d\kappa$ on the limit  on $\{s=z_1=0\}$ and along the flow of the Hamiltonian vector fields associated with $\gamma^2(-\sigma\pm\sqrt{s^2+z_1^2})$. A simple calculus shows that since the canonic symplectic form $\omega$ is preserved by canonical transform, the relation~(\ref{eq:sigmaHH'}) and the fact that 
$$\omega(-\partial_s-\partial_\sigma,-\partial_s+\partial_\sigma)=2$$
 imply that $H'$ is sent on $\gamma^2(-\partial_s+\partial_\sigma)$, the limit  as $s$ goes to $0^-$ of the Hamiltonian field associated with $\gamma^2(-\sigma - \sqrt{s^2+z_1^2})$, and $H$ is sent on on $\gamma^2(-\partial_s-\partial_\sigma)$, the limit  as $s$ goes to $0^-$ of the Hamiltonian field associated with $\gamma^2(-\sigma +\sqrt{s^2+z_1^2})$. This observation allows to relate the modes after the change of coordinates. 

\medskip
 
 As a consequence, in these new variables $(s,z,\sigma,\zeta)$, 
the geometry of the crossing is simple  and  we have 
\begin{eqnarray}
\nonumber
&S:=\{\xi=0,\;\tau+V(x)=0\}=\kappa\left(\{ s=0,\;z_1=0,\;\sigma=0\}\right),&\\
\label{def:E+-}
&E^\pm=\kappa\left(\{-\sigma\mp \sqrt {s^2+z_1^2}=0\}\right),&
\end{eqnarray}
where the energy sets $E^\pm$ are defined by \eqref{numero}.

\medskip

Finally, in the construction of the canonical form~$\kappa$,  the function $z_1(x,\xi)$ can be related to the variables $x$ and $\xi$ according to 
\begin{equation}\label{eq:z}
z_1(x,\xi)=\xi\wedge \,{\nabla V(x)\over |\nabla V(x)|^{3/2}} +\mathcal O(|\xi|^2).
\end{equation}
Similar formula can be written for the functions $s$ and $\sigma$. However, in the sequel, we will only use the formula  for $z_1$.

\medskip

Then, thanks to Theorem~3 of~\cite{CdV1}, it is possible to pass equation~(\ref{def:B}) at the quantum level:  
  there exists a unitary operator $K_{\eps}$ and a matrix $B_1$ such that
$$\displaylines{\qquad
K_\eps\op_\eps(\,^tB_\eps)\op_\eps\left(-\sigma+\widetilde A(s,z_1)\right)\op_\eps(B_\eps)(K_\eps)^*\hfill\cr\hfill
=\op_\eps(\tau+V(x)+A(\xi))+\mathcal O(\eps^{2}),\qquad\cr}$$
where 
$B_\eps=B+\eps B_1$. The operator $K_\eps$ is a Fourier Integral Operator associated with the canonical transform $\kappa$ (see~\cite{DS} or~\cite{FG1}). 
It allows to pass at the quantum level the relation~(\ref{def:B}) induced by the change of variables~$\kappa$.
An important property of these Fourier Integral Operators is that they are compatible with pseudo differential calculus in the sense that for all $a\in{\mathcal C}_0^\infty(\R^6)$,
\begin{equation}\label{OIF-1}K_\eps \op_\eps(a\circ\kappa^{-1}) K_\eps^*= \op_\eps(a)+ \mathcal O(\eps N_\eps(a)),
\end{equation}
where 
$$N_\eps(a)= \sup_{2\leq|\alpha|+|\beta|\leq N_0,\;|\alpha|} \,\sup_{(t,x,\tau,\xi)\in\R^{6}}\left| \partial_{t,x}^\alpha\partial_{\tau,\xi}^\beta a(t,x,\tau,\xi) \right|$$ for some $N_0\in\N$. 
In particular, in view of the remarks developed in the Appendix, when one applies this relation to a two-scaled symbol of the form 
$$a_{\eps,R}(x,\xi)= \chi(t) c_{\eps,R}(x,\xi)
\theta\!\left({\tau\pm|\xi|+V(x)\over
R\sqrt\eps}\right)\Pi^\pm(\xi),$$
one gets 
\begin{equation}\label{OIF-2}
K_\eps \op_\eps(a_{\eps,R}\circ\kappa^{-1}) K_\eps^*= \op_\eps(a)+ \mathcal O(\sqrt \eps),
\end{equation}
We will use this property to translate the quantities that we want to study in the variables $(t,x,\tau,\xi)$ in these new variables $(s,z,\sigma,\zeta)$. More precisely,  we  set $$v^\eps=\op_\eps(B_\eps)K_\eps^*\psi^\eps,$$
then $v^\eps$ solves (microlocally in $L^2(\R^3_{s,z})$)
  the system 
\begin{equation}\label{eq:systreduit}
{\eps\over i} \partial_s v^\eps =\widetilde A(s,z_1) v^\eps +\mathcal O(\eps^{2})
\end{equation}
and 
$$\left(\op_\eps((\,^tBaB)\circ\kappa)\psi^\eps,\psi^\eps\right)_{L^2(\R^3_{t,x})}=
\left(\op_\eps(a) v^\eps,v^\eps\right)_{L^2(\R^3_{s,z})}+\mathcal O(\eta_\eps)$$
where $\eta_\eps$ denotes a rest term as defined in Notation~\ref{notation}. 
In particular, for scalar functions $a$, we have 
  $$\left(\op_\eps((\gamma a)\circ\kappa)\psi^\eps,\psi^\eps\right)_{L^2(\R^3_{t,x})}=
\left(\op_\eps(a) v^\eps,v^\eps\right)_{L^2(\R^3_{s,z})}+\mathcal O(\eta_\eps)$$
In what follows, we shall focus on the analysis of this family $v^\eps$.

 \medskip 
 
Let us now write the Markov process ${\mathcal L}_{\eps,R}$ in the new coordinates, we shall denote by $\widetilde {\mathcal L}_{\eps,R}$ the resulting semi-group. 
\begin{itemize}
\item 
As we have already observed, by the geometric properties of canonical transforms, 
 the {\it Hamiltonian trajectories} of our system are preserved by $\kappa$ and one is able to identify each branch of the trajectories:   the trajectories for the plus mode  are Hamiltonian trajectories of $-\sigma-\sqrt{s^2+z_1^2}$ and the trajectories for the minus mode are those of the Hamiltonian $-\sigma+\sqrt{s^2+z_1^2}$. We denote by $\widetilde \Phi_\pm$ these trajectories and we observe that they write
 \begin{equation}\label{eq:flotaleph}
 \widetilde \Phi^\aleph_\pm (s,z,\sigma,\zeta)=\left( s-\aleph,z,\widetilde\sigma^\aleph_\pm(s,z_1,\sigma),
 (\widetilde\zeta_1)^\aleph_\pm(s,z_1,\sigma),
 \zeta'
\right) ,
 \end{equation}
where we set $\zeta=(\zeta_1,\zeta')$ and $\widetilde\sigma^\aleph_\pm(s,z_1,\sigma)=\sigma\mp\sqrt{(s-\aleph)^2+z_1^2}\pm\sqrt{s^2+z_1^2}$ by the conservation of the energy and  $(\widetilde\zeta_1)^\aleph_\pm(s,z_1,\sigma)=\zeta_1+O(\aleph)$. 
 \item The {\it transitions} occur when the gap is minimal along the trajectories, that is when $s=0$. Besides, when the transitions occur,  one has $\xi\cdot\nabla V(x)=0$, which implies   $|\xi\wedge \nabla V(x)|= |\xi||\nabla V(x)| 
 $ and the relation~(\ref{eq:z}) then gives 
 $$T_\eps(x,\xi)={\rm exp}\left(-{\pi\over \eps}{|\xi\wedge\nabla V(x)|^2\over |\nabla V(x)|^3} \right)=  T_{LZ}\left(\frac{z_1}{\sqrt \eps}\right)(1+\mathcal O(R^3\sqrt \eps))$$
(provided $z_1=O(R\sqrt\eps$) where $T_{LZ}(\eta)={\rm e}^{-\pi\eta^2}$. 
 Since the transition coefficients  $T_\eps(x,\xi)$  and $T_{LZ}(z_1/\sqrt\eps)$ differ of a term of order $\eta_\eps$, we define the flow $\widetilde {\mathcal L}_{\eps,R}$ with the transition rate $T_{LZ}(z_1/\sqrt\eps)$.
\item The {\it hopping region} is  chosen as
 $$\widetilde{\mathcal U}_{\eps,R}=\left\{ | z_1| \leq C_1 R\sqrt \eps  \right\}$$
 because of the precise form of the transition rate and in view of the preceding remarks. 
 \item Finally, we observe that the {\it drift} is made in the direction of $H_+-H_-$.  By the description above, $d\kappa$ sends $H_+-H_-$  on a vector collinear to $\partial_\sigma$. As a consequence, we deduce that there exists a map 
 $$(s,z,\sigma,\zeta)\mapsto \delta\sigma_\pm(s,z,\sigma,\zeta)$$ such that
 \begin{equation}\label{def:tildedrift}
 \widetilde J_\pm:=\kappa^{-1}\circ J_\pm\circ \kappa \, (0,z,\sigma,\zeta)=(0,z,\sigma+\delta\sigma_\pm,\zeta).
 \end{equation}
  Using~(\ref{def:E+-}), we deduce  
$\delta\sigma_\pm= \mp 2|z_1|$.
\end{itemize}

\medskip

Let us now reformulate our problem in these new variables. Recall that we work in the region $\widetilde{\mathcal U}_{\eps,R}$.
Let $b^\pm(s,z,\zeta)$ be two smooth functions compactly supported in $\{s>0\}$ and such that the trajectories reaching their support have only experienced one transition during an interval of time of length $\aleph$. We also suppose that the functions 
 $b^\pm(s+\aleph,z,\zeta)$ are supported in $\{s<0\}$.
 We consider the symbol $c_{\eps,R}^{out}=(c_{\eps,R}^{+,out},c_{\eps,R}^{-,out})$ defined by 
 $$c_{\eps,R}^{\pm,out}(s,z,\sigma,\zeta)=   b^\pm (s,z,\zeta)
   \theta\left(\frac{\widetilde\lambda^\pm(s,z_1,\sigma)}{R\sqrt\eps}\right)$$
 where $\theta$ is the cut-off function of Lemma~\ref{lem:loc} and $\widetilde\lambda^\pm(s,z_1,\sigma)$ is the energy
 $$\widetilde\lambda^\pm(s,z_1,\sigma)=-\sigma\mp\sqrt{z_1^2+s^2}.$$
Note that the localization in energy yields that $\sigma= \mp \sqrt{z_1^2+s^2} +O(R\sqrt\eps)$
  in the zone of interest; for this reason we do not need to assume that $b^\pm$ depends on the variable $\sigma$.
 
  We now want to compute $\widetilde{\mathcal L}_{\eps,R}^\aleph c_{\eps,R}^{out}$ the pull back by the semi-group $\widetilde {\mathcal L}_{\eps,R}$ in the normal coordinates.  The observable $c_{\eps,R}^{out}$ has two parts $c_{\eps,R}^{+,out}$ and $c_{\eps,R}^{-,out}$ and we have to consider the random trajectories that reach the support of each of these functions. More precisely, for $c^{+,out}_{\eps,R}$, we consider the plus trajectories that reach its support; however, these trajectories may have known a jump  and either they result from plus trajectories, either they result from minus trajectory. Similar description holds for trajectories reaching the support of $c^{-,out}_{\eps,R}$

At that point of the analysis, we notice that by point (3) of Assumptions~\ref{ass}, in the ingoing region, one of the mode is negligible. 
Without loss of generality, we can assume that the contribution of trajectories which arise from the minus mode is negligible. For summarizing, the picture is the following: 
\begin{itemize}
\item for calculating the backward image by the semigroup of $c_{\eps,R}^{+,out}$, that we shall denote by $c_{\eps,R}^{+,in}$, we only need to consider the plus trajectories which reach its support, 
\item for calculating the backward image by the semigroup of $c_{\eps,R}^{-,out}$, that we shall denote by $c_{\eps,R}^{-,in}$, we need to consider the minus trajectories which reach its support and these trajectories arises from plus trajectories which have had a jump.
\end{itemize} 
We  denote by $\widetilde\Phi^s_\pm$  the Hamiltonian trajectories associated with $\widetilde\lambda^\pm$ and  we observe that along a trajectory, the variable $z_1$ is constant and the variable $\zeta_1$ is constant up to a $\mathcal O(|z_1|)$ term (for $s$ of order $1$). Besides, the variable $\sigma$ is determined by the conservation of the energy. 

\medskip

Let us now calculate $c_{\eps,R}^{+,in}$. By applying the transition rate at time~$0$, we have
$$c_{\eps,R}^{+,in}(s,z,\sigma,\zeta)=\left(1-T_{LZ}\left({z_1\over \sqrt\eps}\right)\right) c_{\eps,R}^{+,out}\circ \widetilde \Phi^{-\aleph}_+\left(s,z,\sigma,\zeta\right).$$
We deduce 
$$\displaylines{c_{\eps,R}^{+,in} (s,z,\sigma,\zeta)=\left(1-T_{LZ}\left({z_1\over \sqrt\eps}\right)\right) 
b^+(s+\aleph,z,\zeta_1+\mathcal O(R\sqrt\eps),\zeta')\hfill\cr\hfill \times\,
\theta\left(\frac{\tilde\lambda^+\left(\widetilde \Phi^{-\aleph}_+(s,z,\sigma,\zeta)\right)}{R\sqrt\eps}\right).\cr}$$
Using moreover the conservation of the energy $\tilde\lambda^+$ along trajectories, we obtain 
\begin{equation}\label{cin1}
c_{\eps,R}^{+,in} (s,z,\sigma,\zeta)  = \left(1- T _{LZ}\left(\frac{z_1}{\sqrt\eps}\right)\right)  b^+ \left(s+\aleph, z,\zeta\right)
 \theta\left(\frac{\widetilde\lambda^+(s,z_1,\sigma)}{R\sqrt\eps}\right)+\mathcal O(R\sqrt\eps).
 \end{equation}

  The component  $c_{\eps,R}^{-,in}$ is more intricate since it incorporates classical transport through both modes, application of the transfert coefficient and of the drift. Indeed, the branches of minus trajectories which reach the support of $c_{\eps,R}^{-,out}$ results from plus trajectories that have been drifted.
   By applying the transition rate at time~$0$, we obtain
$$c_{\eps,R}^{-,in}(s,z,\sigma,\zeta)=T_{LZ}\left({z_1\over \sqrt\eps}\right) 
c_{\eps,R}^{-,out}\circ \widetilde \Phi^{s}_-\circ \widetilde J_-\circ \widetilde \Phi^{-s-\aleph}_+
\left(s,z,\sigma,\zeta\right).$$
We deduce 
$$\displaylines{c_{\eps,R}^{-,in} (s,z,\sigma,\zeta)=T_{LZ}\left({z_1\over \sqrt\eps}\right)
b^-(s+\aleph,z,\zeta_1+\mathcal O(R\sqrt\eps),\zeta')\hfill\cr\hfill\times\, 
\theta\left(\frac{\tilde\lambda^-\left(\widetilde \Phi^{s}_-\circ \widetilde J_-\circ \widetilde \Phi^{-s-\aleph}_+(s,z,\sigma,\zeta)\right)}{R\sqrt\eps}\right).\cr}$$
The crucial point is that (\ref{def:tildedrift}) implies that for all $(z,\sigma,\zeta)\in\R^5$,
  $$\widetilde \lambda^-(\widetilde J_-(0,z,\sigma,\zeta))=\widetilde\lambda^+(0,z_1,\sigma),$$
so, by using also the conservation of the energy along trajectories, we obtain
   $$\widetilde \lambda^-\left(\widetilde \Phi^{s}_-\circ \widetilde J_-\circ \widetilde \Phi^{-s-\aleph}_+(s,z,\sigma,\zeta)\right)=\widetilde \lambda^+(s,z_1,\sigma). $$
 As a consequence, 
\begin{equation}\label{cin2}
c_{\eps,R}^{-,in} (s,z,\sigma,\zeta) = 
  T_{LZ} \left(\frac{z_1}{\sqrt\eps}\right)   b^- \left(s+\aleph, z,\zeta\right)\\
 \theta\left(\frac{\widetilde\lambda^+(s,z_1,\sigma)}{R\sqrt\eps}\right)+\mathcal O(R\sqrt\eps).
 \end{equation}

 \medskip
 
As a conclusion, in order to prove our result, we only need to prove  the following relations:
\begin{equation}\label{claim}
\left(\op_\eps(c_{\eps,R}^{\pm,out})v^\eps,v^\eps\right)_{L^2(\R^{3}_{s,z})}=
\left(\op_\eps(c_{\eps,R}^{\pm,in})v^\eps,v^\eps\right)_{L^2(\R^3_{s,z})}+\mathcal O(\eta_\eps).
\end{equation}


 \subsubsection{The transitions}
The claim~(\ref{claim}) is proved  by use  of
 the following Landau-Zener type formula (see~\cite{La,Ze,FG1,FL1}). 

\begin{proposition}
\label{prop:scat}
Let $v^\eps$ be a solution of~(\ref{eq:systreduit}). 
There exist two vector-valued functions 
$k^{\eps,\pm}\in L^2(\R^d,\C^2)$
such that  $k^{\eps,\pm}{\bf 1}_{|z_1|\leq R\sqrt\eps}$ are bounded and such that 
we have for $\pm s>0$ and  $|z_1|\leq R\sqrt\eps$
\begin{eqnarray*}
 v^\eps_1(z,s) & = & 
{\rm e}^{is^{2}/(2\eps)} 
\left|{\tfrac{s}{\sqrt \eps}}\right|^{i\frac{z_1^2}{2\eps}}
k^{\eps,\pm}_{1}(z)+\mathcal O\left(\frac{R^2\sqrt \eps}{s}\right),\\
 v^\eps_2(z,s) & = & 
{\rm e}^{-is^{2}/(2\eps)}
\left|{\tfrac{s}{\sqrt \eps}}\right|^{-i\frac{z_1^2}{2\eps}}
k^{\eps,\pm}_{2}(z)+\mathcal O\left(\frac{R^2\sqrt\eps}{s}\right)
\end{eqnarray*}
Moreover, 
$k^{\eps,+}=S\left(\frac{z_1}{\sqrt\eps}\right)k^{\eps,-}$
where the unitary matrix $S$ is given by 
$$\displaylines{
S(\lambda)=
\begin{pmatrix}a(\lambda^2) & -
\lambda \overline b(\lambda^2) \\ \lambda b(\lambda^2) & a(\lambda^2)\end{pmatrix},\cr
a(\lambda)={\rm e}^{-\pi\lambda/2},\qquad
b(\lambda)=
\frac{2i{\rm e}^{i\pi/4}}{\lambda\sqrt\pi}2^{-i\lambda/2}
{\rm e}^{-\pi\lambda/4}\Gamma(1+i{\textstyle\frac\lambda2})
\sinh({\textstyle\frac{\pi\lambda}{2}}).\cr}
$$ 
\end{proposition}

We are now in position to conclude the proof of Proposition~\ref{prop:FL} by proving~(\ref{claim}). Note that we have $|s|>c_0>0$ on the support of our symbols.
We first take advantage of the localization near the energy surfaces  to translate it as a focalization: by the analogous of Lemma~\ref{lem:loc} and Remark~\ref{rem:loc} is $(s,z)$ variables, we obtain 
$$\displaylines{
\left(\op_\eps( c_{\eps,R}^{\pm,out}) v^\eps,v^\eps\right)_{L^2(\R^3_{s,z})} 
= \left(\op_\eps( c_{\eps,R}^{\pm,out}\widetilde\Pi^\pm(s,z_1))v^\eps,v^\eps\right) _{L^2(\R^3_{s,z})}+\mathcal O(\eta_\eps)\hfill\cr\hfill
 =  \left(\op_\eps( c_{\eps,R}^{\pm,out}\widetilde\Pi^\pm(s,z_1))v^\eps,\widetilde\Pi^\pm(s,z_1)v^\eps\right) _{L^2(\R^3_{s,z})}+\mathcal O(\eta_\eps)\cr}$$
where we denote by $\widetilde\Pi^\pm(s,z_1)$ the eigenprojectors of the matrix $$\widetilde A(s,z_1)=\displaystyle{ \begin{pmatrix}s & z_1 \\ z_1 & -s \end{pmatrix}}$$ associated with the eigenvalues $\mp\sqrt{s^2+z_1^ 2}$.  
For $|z_1|\leq R\sqrt\eps$, we have
\begin{eqnarray}
\nonumber
\widetilde\Pi^+(s,z_1)&=&
\begin{pmatrix}0&0\\ 0&1\end{pmatrix}+\mathcal O(R\sqrt\eps)
\;\;{\rm in}\;\;\{s>0\},\\
\label{eq:Pitilde}
\widetilde\Pi^+(s,z_1)&=&\begin{pmatrix}1&0\\0&0\end{pmatrix}+\mathcal O(R\sqrt\eps)
\;\;{\rm in}\;\;\{s<0\}
\end{eqnarray}  
(whence similar asymptotics for $\widetilde\Pi^-$ since ${\rm Id}=\widetilde \Pi^++\widetilde \Pi^-$). Therefore,
\begin{eqnarray*}
\left(\op_\eps( c_{\eps,R}^{+,out})v^\eps,v^\eps\right)_{L^2(\R^{3}_{s,z})}&=&\left(\op_\eps(b_{\eps,R}^{+,out})v^\eps_2,v^\eps_2\right)_{L^2(\R^{3}_{s,z})}+\mathcal O(\eta_\eps),\\
\left(\op_\eps( c_{\eps,R}^{-,out})v^\eps,v^\eps\right)_{L^2(\R^{3}_{s,z})}&=&\left(\op_\eps( b_{\eps,R}^{-,out}) v^\eps_1,v^\eps_1\right)_{L^2(\R^{3}_{s,z})}+\mathcal O(\eta_\eps)\end{eqnarray*}
with
$b^{\pm,out}_{\eps,R}= b^\pm (s,z,\zeta) .$

\medskip

The symbols $c^{\pm,in}_{\eps,R}$ are supported in the region $\{s<0\}$, i.e. before the transitions. By (3) of Assumptions~\ref{ass}, we know that the mode minus is negligible when $s<0$. By~(\ref{eq:Pitilde}), in $\{s<0\}$,  the mode plus corresponds to the component $v^\eps_1$ and we deduce  $v^\eps_2=\mathcal O(\eta_\eps)$. Therefore, we have
\begin{eqnarray*}
&\;\;\left(\op_\eps(c_{\eps,R}^{+,in}) v^\eps,v^\eps\right)_{L^2\R^{3}_{s,z})}=\left(\op_\eps\left(\left(1- T_{LZ} \left(\frac{z_1}{\sqrt\eps}\right)\right)  b_{\eps,R}^{+,in}\right)v^\eps_1,v^\eps_1\right)_{L^2(\R^{3}_{s,z})}+\mathcal O(\eta_\eps),&\\
\nonumber
&\left(\op_\eps(c_{\eps,R}^{-,in}) v^\eps,v^\eps\right)_{L^2(\R^{3}_{s,z})}=\left(\op_\eps\left( T_{LZ} \left(\frac{z_1}{\sqrt\eps}\right)  b_{\eps,R}^{-,in}\right)v^\eps_1,v^\eps_1\right)_{L^2(\R^{3}_{s,z})}+\mathcal O(\eta_\eps),&
\end{eqnarray*}
with
$$
b^{+,in}_{\eps,R}  =  b^+ \left(s+\aleph, z,\zeta\right)
,\qquad
b^{-,in}_{\eps,R}  =   b^- \left(s+\aleph, z,\zeta\right).
$$

\medskip 

We now use  Proposition~\ref{prop:scat} in order to relate $v^\eps_{1,2}$ with $k^{\eps,\pm}_{1,2}$. For the term in function of $c^{+,out}_{\eps,R}$, we have
$$\displaylines{
\left(\op_\eps( c_{\eps,R}^{+,out})v^\eps,v^\eps\right)_{L^2(\R^{3}_{s,z})}
= \mathcal O(R^2\sqrt\eps) \hfill\cr\hfill 
+\left(\op_\eps\left(b_{\eps,R}^{+,out}\right){\rm e}^{-i\frac{s^{2}}{2\eps}}
\left|{\tfrac{s}{\sqrt \eps}}\right|^{-i\frac{z_1^2}{2\eps}}
k^{\eps,+}_{2}\right.,
\left.{\rm e}^{-i\frac{s^{2}}{2\eps}}
\left|{\tfrac{s}{\sqrt \eps}}\right|^{-i\frac{z_1^2}{2\eps}}
k^{\eps,+}_{2}\right)_{L^2(\R^{3}_{s,z})}
\cr}$$
We use the following Lemma (see  Lemma~8 and Lemma~9 in~\cite{FG2}) in order to commute the pseudo differential operator and the phases.

\begin{lemma}\label{lem:chi(KK*)} In ${\mathcal L}\left(L^2\left(\R^{3}_{s,z}\right)\right)$, we have for $|z_1|\leq R\sqrt\eps$
$$
\left|{\tfrac{s}{\sqrt\eps}}\right|^{\pm
i\frac{z_1^2}{2\eps }} 
\op_\eps\!\left(b\!\left(s,z,\zeta\right)\right)
\left|{\tfrac{s}{\sqrt\eps}}\right|^{\mp i\frac{z_1^2}{
2\eps}}=
\op_\eps\!\left(b\!\left(s,z,\zeta\right)\right)+\mathcal O(\sqrt\eps|\ln\eps|).$$
\end{lemma}

As a consequence, we obtain 
$$
\left(\op_\eps( c_{\eps,R}^{+,out})v^\eps,v^\eps\right)_{L^2(\R^{3}_{s,z})}
=
\left(\op_\eps\!\left(b^{+,out}_{\eps,R}\right)k^{\eps,+}_{2},k^{\eps,+}_{2}\right)_{L^2(\R^{3}_{s,z})}
+ \mathcal O(\eta_\eps).$$
Using the change of variable $s\mapsto s+\aleph$, we find  
$$
\left(\op_\eps( c_{\eps,R}^{+,out})v^\eps,v^\eps\right)_{L^2(\R^{3}_{s,z})}
=
\left(\op_\eps\!\left(b^{+,in}_{\eps,R}\right)k^{\eps,+}_{2},k^{\eps,+}_{2}\right)_{L^2(\R^{3}_{s,z})}
+ \mathcal O(\eta_\eps).$$
By Proposition~\ref{prop:scat}, we have 
$$k^{\eps,+}_2=
\frac{z_1}{\sqrt\eps}b\left(\frac{z_1^2}{\eps}\right) k_1^{\eps,-}+a\left(\frac{z_1^2}{\eps}\right)k_2^{\eps,-}$$
Besides, 
 $b^\pm(s+\aleph)$ is supported in $\{s<0\}$
and, as we said before, by (3) of Assumptions~\ref{ass}, near $\{\sigma+s=0,\,\,s<0\}$  (i.e. near minus trajectories entering in the hopping zone), 
$
v^\eps_2(z,s)  = 
\mathcal O(\eta_\eps)
$.  Therefore,  because of the link between $v^\eps_2$ and $k_2^{\eps,-}$ in the region $\{s<0\}$, 
$
k_2^{\eps,-}(z) = \mathcal O(\eta_\eps)
$ and, using  the relation
$
\lambda^2|b(\lambda)|^2=1-a(\lambda)^2= 1-{\rm e}^{-\pi\lambda},
$ 
we obtain
$$
\frac{z_1}{\sqrt\eps}\, \overline b\left(\frac{z_1^2}{\eps}\right)\, 
\op_\eps\!\left(b^{+,in}_{\eps,R}\right)
\frac{z_1}{\sqrt\eps}\, \overline b\left(\frac{z_1^2}{\eps}\right)
=\op_\eps\left(
\left((1-T_{LZ}\left(\frac{z_1}{\sqrt\eps}\right)\right)\, 
b^{+,in}_{\eps,R}\right)+ \mathcal O(\sqrt\eps)
$$
in ${\mathcal L}(L^2(\R^{3}_{s,z}))$.
Finally, in view of the relations satisfied by $c_{\eps,R}^{\pm,in}$,  and arguing as before, we can conclude that
$$\displaylines{
\left(\op_\eps(c_{\eps,R}^{+,out})v^\eps,v^\eps\right)_{L^2(\R^{3}_{s,z})}  = 
\left( \op_\eps\left( \left(1-T_{LZ}\left(\frac{z_1}{\sqrt\eps}\right)\right)b^{+,in}_{\eps,R}\right) k^{\eps,-}_1,k^{\eps,-}_1\right)_{L^2(\R^{3}_{s,z})}\cr
\hfill
+\mathcal O(\eta_\eps)\quad\cr\hfill 
 =  \left(\op_\eps(c_{\eps,R}^{+,in})v^\eps,v^\eps\right)_{L^2(\R^{d+1})}+\mathcal O(\eta_\eps).\cr}$$
Arguing similarly for $c_{\eps,R}^{-,out}$, we obtain~(\ref{claim}).

\qed


\section{Appendix: Pseudo-differential calculus}

In this Appendix, we recall a few results of symbolic calculus that we use in this article and we 
prove the estimate~(\ref{eq:CV}), which is at the core of these results.

\medskip

The estimate~(\ref{eq:CV}) relies on the Schur Lemma. With $f\in L^2(\R^d)$, we associate  
$${\mathcal F}_\eps(f)(\xi): =(2\pi\eps)^{-d}\widehat f \left({\xi\over\eps}\right) = (2\pi\eps)^{-d/2} \int f(x) {\rm e}^{-{i \over\eps}\xi\cdot x} dx,$$
and we observe that for any $a\in{\mathcal C}_0^\infty(\R^{2d}) $, 
\begin{eqnarray*}
\left(\op_\eps(a)f,f\right) & =& (2\pi\eps)^{-d}
\int a\left({x+y\over 2},\xi\right){\rm e}^{{i\over \eps}\xi\cdot(x-y) }f(y)\overline f(x) dx\,dy\,d\xi\\
& =& (2\pi\eps)^{-3d}
\int a\left({x+y\over 2},\xi\right){\rm e}^{{i\over \eps}\left(\xi\cdot(x-y) +i\eta\cdot y-ix\cdot\zeta\right)}\\
& & \qquad\qquad\qquad \times\, 
\widehat f\left({\eta\over\eps}\right)\overline {\widehat f}\left({\zeta\over \eps}\right) dx\,dy\,d\xi\, d\eta\,d\zeta\\
& =& (2\pi\eps)^{-2d}
\int a\left(X,\xi\right){\rm e}^{{i\over \eps}\left(\xi\cdot v +i\eta\cdot\left(X-{v\over 2}\right)-i\zeta\left(X+{v\over 2}\right) \right)}\\
& & \qquad\qquad\qquad \times\, 
 {\mathcal F}_\eps (f)(\eta)\overline{{\mathcal F}_\eps( f)}(\zeta) dX\,dv\,d\xi\,d\eta\,d\zeta\\
& =& (2\pi\eps)^{-d}
\int a\left(-X,{\eta+\zeta\over 2}\right){\rm e}^{{i\over \eps}X\cdot(\zeta-\eta)} {\mathcal F}_\eps (f)(\eta)\overline{{\mathcal F}_\eps( f)}(\zeta) dX\,d\eta\,d\zeta\\
&=&  \left(K_\eps {\mathcal F}_\eps(f),{\mathcal F}_\eps(f)\right),
 \end{eqnarray*}
where $K_\eps$ is the operator of kernel $k_\eps(\xi,\xi')$,
$$k_\eps(\xi,\xi')={1\over \eps^d} \widetilde a \left({\xi+\xi'\over 2},{\xi-\xi'\over\eps}\right),$$
with 
$$\widetilde a (\xi,v)=(2\pi)^{-d} \int a(-x,\xi) {\rm e}^{ix\cdot v} dx.$$
By the Plancherel theorem, the norm of $\op_\eps(a)$ and of $K_\eps$ are the same. Besides, 
using the  Schur Lemma, we obtain 
\begin{eqnarray*}
\| K_\eps\| & \leq & {\rm Max} \left( \sup_{\xi\in\R^d} \int |k_\eps(\xi,\xi')|d\xi',  \sup_{\xi'\in\R^d} \int |k_\eps(\xi,\xi')|d\xi
\right)\\
& \leq & \int_{\R^d} \sup_{\xi\in\R^d} \left| \widetilde a (\xi,v)\right| dv\\
& \leq & C\, \sup_{|\beta|\leq d+1}\,\sup_{\xi\in\R^d} \int \left| \partial_x^\beta a(x,\xi) \right| dx.
\end{eqnarray*}
for some constants $C$ independent of $a$ and $\eps$, which gives the result. 

\medskip

By use of the same techniques and of the Taylor formula, one can prove the following result about the composition of pseudo differential operators (see for example section 4.1 of chapter 2 in~\cite{AFF}):

\begin{proposition}\label{prop:symbol}
Let $a,b\in{\mathcal C}_0^\infty(\R^{2d},\C^{N,N})$, $N\in\N$, then
$$\op_\eps(a)\op_\eps(b)  = \op_\eps(ab)+{\eps\over 2i} \op_\eps(\{a,b\})+\eps^2 R_\eps,$$
with $\{a,b\}=\nabla_\xi a \cdot \nabla _x b-\nabla _xa\cdot \nabla_\xi b$ and 
$$\| R_\eps\|_{{\mathcal L}(L^2(\R^d)}\leq C \,\sup_{|\alpha|+|\beta|=2} \sup_{  |\gamma|\leq d+1}\,\sup_{\xi\in\R^d}\left( \int \left| \partial_\xi^\alpha \partial_x^{\beta+\gamma}  a(x,\xi) \right| dx\right)\left( \int \left| \partial_\xi^\beta \partial_x^{\alpha+\gamma}  b(x,\xi) \right| dx\right)$$
for some constant $C>0$ independent of $a$, $b$ and $\eps$.
\end{proposition}

\medskip

Let us give a few comments on symbolic calculus involving two-scaled symbols of the form 
$$a_{\eps,R}(x,\xi)=a\left(x,\xi,{f(x,\xi)\over R\sqrt\eps}\right)$$
for some smooth function $f(x,\xi)$ and smooth bounded function $a(x,\xi,\eta)$ compactly supported in variables $(x,\xi)$ uniformly in $\eta$ and with bounded derivatives in~$\eta$.
Consider the unitary scaling operator $T_\eps$ defined by 
$$\forall u\in L^2(\R^d),\;\;T_\eps u(x)=\eps^{1/4} u(x\sqrt\eps).$$
Note that this scaling operator is at the core of Calder\'on-Vaillancourt proof (\cite{CV}).
Then, we observe that the relation 
$$\op_\eps(a_{\eps,R})= T_\eps^*\op_1\left(a\left(x\sqrt\eps ,\xi\sqrt\eps,{f(x\sqrt\eps,\xi\sqrt\eps )\over R\sqrt\eps}\right)\right)T_\eps,$$
yield the uniform boundedness of the operator $\op_\eps(a_{\eps,R})$ on $L^2(\R^d)$ by the standard Calder\'on-Vaillancourt estimate
$$\exists C,N>0,\;\; \forall a\in{\mathcal C}_0^\infty(\R^{2d}),\;\; \| \op_1(a)\| _{{\mathcal L}(L^2(\R^d))}\leq C \,\sup_{|\alpha|+|\beta|\leq N}\sup_{(x,\xi)\in\R^{2d}}\left| \partial_x^\alpha \partial_\xi^\beta a (x,\xi) \right|.$$
Besides, by using similarly the operator $T_\eps$, the reader will convince oneself that Proposition~\ref{prop:symbol} holds with rest terms of size $\sqrt\eps$ as soon as one and only one of the involved symbols is two-scaled. That is the precise reason why, following the construction of Fourier Integral Operators (as performed in~\cite{FG1} for example), equation~(\ref{OIF-1}) extend to two-scaled symbols and writes~(\ref{OIF-2}).


\end{document}